\documentclass[journal]{IEEEtran}

\usepackage{setspace}
\usepackage{fancyhdr}

\usepackage{bm}
\usepackage{amssymb}
\usepackage{amsmath}
\usepackage{amsthm}
\newtheorem{theorem}{Theorem}
\newtheorem{lemma}{Lemma}

\newtheorem{definition}{Definition}

\newtheorem{assumption}{Assumption}

\newtheorem{observation}{Observation}
\usepackage{graphicx}

\newcommand{\minitab}[2][l]{\begin{tabular}{#1}#2\end{tabular}}

\usepackage{dsfont}

\usepackage{cite}

\usepackage{subfigure}

\usepackage{caption}

\usepackage{algorithm}
\usepackage{algorithmic}
\usepackage{multirow}
\usepackage{array}

\usepackage{comment}

\usepackage{stfloats}

\usepackage{url}

\usepackage{color}

\begin{document}
\title{Auction-Based Coopetition between \\LTE Unlicensed and Wi-Fi}





\author{Haoran Yu, George Iosifidis, Jianwei Huang, \IEEEmembership{Fellow, IEEE}, and Leandros Tassiulas, \IEEEmembership{Fellow, IEEE}\vspace{-1cm}
\thanks{Manuscript received April 29, 2016; revised October 16, 2016; accepted November 7, 2016. 
This work is supported by the General Research Fund CUHK 14202814 established under the University Grant Committee of the Hong Kong Special Administrative Region, China, the CUHK Global Scholarship Programme for Research Excellence, and the CUHK Overseas Research Attachment Programme. 
This paper was presented in part at IEEE WiOpt, Tempe, AZ, USA, May 2016 \cite{haoran2016wiopt} (Corresponding author: Jianwei Huang).}
\thanks{H. Yu and J. Huang (\{hryu, jwhuang\}@ie.cuhk.edu.hk) are with the Department of Information Engineering, the Chinese University of Hong Kong, Hong Kong, China.} 
\thanks{G. Iosifidis (iosifidg@tcd.ie) is with the School of Computer Science and Statistics, and CONNECT, Trinity College Dublin, Ireland.} 
\thanks{L. Tassiulas (leandros.tassiulas@yale.edu) is with the Department of Electrical Engineering and the Yale Institute for Network Science, Yale University, USA.}
}

\normalsize %

\maketitle
\thispagestyle{empty}

\setcounter{page}{1}

\begin{abstract}
Motivated by the recent efforts in extending LTE to the unlicensed spectrum, we propose a novel spectrum sharing framework for the \emph{coopetition} (\emph{i.e.}, cooperation and competition) between LTE and Wi-Fi in the unlicensed band. Basically, the LTE network can choose to work in one of the two modes: in the \emph{competition mode}, it randomly accesses an unlicensed channel, and interferes with the Wi-Fi access point using the same channel; in the \emph{cooperation mode}, it delivers traffic for the Wi-Fi users in exchange for the {{exclusive}} access of the corresponding channel. Because the LTE network works in an interference-free manner in the \emph{cooperation mode}, {{it can achieve a much larger data rate than that in the \emph{competition mode},}} which allows it to effectively serve both its own users and the Wi-Fi users. 
{{We design a {{second-price}} reverse auction mechanism, which enables the LTE provider and the Wi-Fi access point owners (APOs) to effectively negotiate the operation mode. 
Specifically, the LTE provider is the auctioneer (buyer), and the APOs are the bidders (sellers) who compete to sell their channel access opportunities to the LTE provider. 
In Stage I of the auction, the LTE provider announces a reserve rate, which is the maximum data rate that it is willing to allocate to the APOs in the \emph{cooperation mode}. In Stage II of the auction, the APOs submit their bids, which indicate the data rates that they would like the LTE provider to offer in the \emph{cooperation mode}.}}
We show that the auction involves \emph{allocative externalities}, \emph{i.e.}, the cooperation between the LTE provider and one APO benefits other APOs who are not directly involved in this cooperation. As a result, 
{{a particular APO's willingness to cooperate is affected by its belief about other APOs' willingness to cooperate.}} This makes our analysis much more challenging than that of the conventional second-price auction, where bidding truthfully is a weakly dominant strategy. 
{{We show that the APOs have a unique form of the equilibrium bidding strategies in Stage II, based on which we analyze the LTE provider's optimal reserve rate in Stage I.}} 
Numerical results show that our framework improves the payoffs of both the LTE provider and the APOs comparing with a benchmark scheme. {{In particular, our framework increases the LTE provider's payoff by $70\%$ on average when the LTE provider has a large throughput and a small data rate discounting factor.}} Moreover, our framework leads to a close-to-optimal social welfare under a large LTE throughput.

\end{abstract}


\normalsize %



\section{Introduction}
\subsection{Motivations}
\IEEEPARstart{T}{he} proliferation of mobile devices is leading to an explosion of global mobile traffic, which is estimated to reach 30.6 exabytes per month by 2020 \cite{Cisco2016}. To accommodate this {{rapidly}} growing mobile traffic, 3GPP has been working on proposals to enable LTE to operate in the unlicensed 5GHz band \cite{3GPP}.{\footnote{{The LTE unlicensed technology can also work in the 3.5GHz band \cite{kim2015design}. However, since the available spectrum resources for the LTE technology in the 5GHz band ($500$MHz) are much more than those in the 3.5GHz band ($80$MHz), we focus on the interaction between the LTE and Wi-Fi in the 5GHz in this paper.}}} By extending LTE to the unlicensed spectrum, the LTE provider can significantly expand its network capacity, and tightly integrate its control over the licensed and unlicensed bands \cite{Senza}. Furthermore, since the LTE technology has an efficient framework of traffic management (\emph{e.g.,} congestion control), it is capable of achieving a much higher spectral efficiency than Wi-Fi networks in the unlicensed spectrum, if there is no competition between these two technologies \cite{forum}.
Key market players, such as AT\&T, Verizon, T-Mobile, Qualcomm, and Ericsson, have already demonstrated the potential of LTE in the unlicensed band through experiments \cite{forum}, and have formed several forums (\emph{e.g.}, LTE-U Forum \cite{LTEforum} and EVOLVE \cite{EVOLVE}) to promote this promising LTE unlicensed technology.

A key technical challenge for LTE working in the unlicensed spectrum is that it can significantly degrade the Wi-Fi network performance if there is no effective co-channel interference avoidance mechanism. 
To address this issue, industries have proposed two major mechanisms for LTE/Wi-Fi coexistence: 
(a) Qualcomm's carrier-sensing adaptive transmission (CSAT) scheme \cite{Qualcomm}, where the LTE transmission follows a periodic on/off pattern creating interference-free zones for Wi-Fi during certain periods, and (b) Ericsson's  ``Listen-Before-Talk'' (LBT) scheme \cite{Ericsson}, where LTE transmits only when it senses the channel being idle for at least certain duration. 
{{However, field tests revealed that these solutions often perform below expectations in practice.}} 
In particular, a series of experiments by Google revealed that both mechanisms severely affect the performance of Wi-Fi \cite{Google}: for the CSAT mechanism, since Wi-Fi is not designed in anticipation of LTE's activity, it cannot respond well to LTE's on-off cycling, and its transmission is severely affected; for the LBT mechanism, it is challenging to choose the proper backoff time and transmission length for LTE to fairly coexist with Wi-Fi. 
{{Therefore, beyond these coexistence mechanisms, there is a need for a novel framework that can effectively explore the potential cooperation opportunity between LTE and Wi-Fi to directly avoid the co-channel interference. This motivates our study in this work.}}

\subsection{Contributions}
{{Unlike previous solely technical coexistence mechanisms that focused on the fair competition between LTE and Wi-Fi, we design a novel coopetition framework. The basic idea is that the two types of networks (LTE and Wi-Fi) should explore the potential benefits of cooperation before deciding whether to enter head-to-head competition. 
Under certain conditions (\emph{e.g.}, the co-channel interference heavily reduces the data rates of both LTE and Wi-Fi), it would be more beneficial for both types of networks to reach an agreement on the cooperation; 
otherwise, they will compete with each other based on a typical coexistence mechanism (\emph{e.g.}, CSAT or LBT).}} 

{{In our coopetition framework,}} 
the LTE network works in either the \emph{competition mode} or the \emph{cooperation mode}. For the \emph{competition mode}, the LTE network simply shares the access of a channel with the corresponding Wi-Fi access point.{\footnote{We consider a general coexistence scheme between LTE and Wi-Fi. Hence, our model applies to both the CSAT and the LBT mechanisms.}} 
For the \emph{cooperation mode}, the LTE network \emph{exclusively} occupies a Wi-Fi access point's channel and the corresponding Wi-Fi access point does not transmit, which avoids the co-channel interference and hence generates a high LTE data rate. Meanwhile, the Wi-Fi access point \emph{onloads} its users to the LTE network,{\footnote{{With the industrial standardization efforts (\emph{e.g.}, Hotspot 2.0 \cite{4Ginteg}), there is a trend of tightly integrating the Wi-Fi technology with the cellular networks. This enables various forms of cooperations between the cellular and Wi-Fi network providers. A successful example is Wi-Fi data offloading, where the cellular network providers offload their cellular traffic to the third-party Wi-Fi networks to relieve the cellular congestion \cite{iosifidis2015iterative,gao2014bargaining,dong2014ideal,lu2014easybid}. In terms of the practical implementation, one advantage of the data onloading over the Wi-Fi data offloading is that the data onloading can be more secure and better protect the mobile users' privacy. This is because the cellular networks usually provide better security guarantees than the Wi-Fi networks.}}} which serves the Wi-Fi access point's users with some data rates based on the access point's request. 
Since LTE usually achieves a much higher spectral efficiency than Wi-Fi \cite{forum,canousing}, such a cooperation {{can potentially}} lead to a \emph{win-win} situation for both networks.

{{In our work, we want to answer the following two questions: (1) \emph{How would LTE and Wi-Fi negotiate over which mode (competition mode or cooperation mode) that LTE would use?} (2) \emph{If the LTE network works in the cooperation mode, how much Wi-Fi traffic should it serve?} Addressing these questions is challenging because of the following reasons: (i) given the increasingly large penetration of Wi-Fi technology, there are usually {{multiple}} Wi-Fi networks in range. {{As we will show in our analysis, the cooperation between the LTE network and one Wi-Fi network imposes a positive externality to other Wi-Fi networks not involved in the cooperation}}; (ii) there is no centralized decision maker in such a system, and different networks have conflicting interests as each of them wants to maximize the total data rate received by its own users; (iii) the throughput of a network (LTE or Wi-Fi) when it exclusively occupies a channel is its private information not known by others, which makes the coordination difficult.}}


{{To address these issues, in Section \ref{sec:model}, we design a mechanism that operates with minimum signaling and computations, and can be implemented in an almost real-time fashion. Specifically, the mechanism is based on}} 
a reverse auction where the LTE provider is the auctioneer (buyer) and wants to \emph{exclusively} obtain the channel from one of the Wi-Fi access point owners {{(APOs, sellers)}}.{\footnote{We consider one LTE network and multiple Wi-Fi access points, since the LTE network has a larger coverage than the Wi-Fi access points, and the Wi-Fi access points are already very popular and exist in many areas.}} 
{{We define the payoff of a network (LTE or Wi-Fi) as the total data rate received by its users.}} 
{{In Stage I of the auction,}} the LTE provider announces the maximum data rate (\emph{i.e.,} reserve rate) that it is willing to allocate for serving users of the winning APO. By optimizing the reserve rate, the LTE provider can affect the APOs' willingness of cooperation, and hence maximize its expected payoff. 
{{In Stage II of the auction,}} given the reserve rate, the APOs report whether they are willing to cooperate and what are the data rates that they request from the LTE provider. 
{{Different APOs may have different requests, since they can have different data rates when exclusively occupying their channels.}} 
If no APO wants to cooperate, the LTE network works in the \emph{competition mode}, and randomly accesses an APO's channel {{(based on a coexistence mechanism like CSAT or LBT)}}; otherwise, it works in the \emph{cooperation mode}, {{and cooperates with the APO that requests the lowest data rate from the LTE provider}}. Such an auction mechanism is particularly challenging to analyze since it induces \emph{positive allocative externalities} \cite{jehiel2000auctions}: the cooperation between the LTE provider and one APO will benefit other APOs not involved in this collaboration, because other APOs can avoid the potential interference generated by the LTE network under the \emph{competition mode}.

{{In Section \ref{sec:stageII:APO}, we analyze the APOs' equilibrium strategies in Stage II of the auction, given the LTE provider's reserve rate in Stage I. We show that an APO always has a unique form of the bidding strategy at the equilibrium under a given reserve rate. However, such a unique form of the bidding strategy may have different closed-form expressions based on different intervals of the reserve rate.}} 
Furthermore, our study shows that for some APOs, the data rates they request from the LTE provider are lower than the rates they can obtain by themselves without the LTE's interference. Intuitively, such a low request motivates the LTE network to work in the \emph{cooperation mode} rather than the \emph{competition mode}. In the latter case, the APOs may receive even lower data rates due to the potential co-channel interference from the LTE network.

{{In Section \ref{sec:stageI:LTE}, we analyze the LTE provider's equilibrium choice of reserve rate in Stage I of the auction, by anticipating the APOs' equilibrium strategies in Stage II.}} 
The LTE network's expected payoff has different closed-form function forms, over different intervals of the reserve rate. We analyze the optimal reserve rate by jointly considering all the reserve rate intervals. We show that when the LTE network's throughout exceeds a threshold, it will choose a reasonably large reserve rate and cooperate with some APOs; otherwise, it will restrict the reserve rate to a small value, and eventually work in the \emph{competition mode}.

The main contributions of this work are as follows:
\begin{itemize}
\item \emph{Proposal of the LTE/Wi-Fi coopetition framework:} 
{{We propose a coopetition framework that explores the cooperation opportunity between LTE and Wi-Fi in order to determine whether they should directly compete with each other. Unlike previously proposed LTE/Wi-Fi coexistence mechanisms, our framework can avoid the data rate reduction when there is a cooperation opportunity between LTE and Wi-Fi. 
Furthermore, our framework can be implemented without revealing the private throughput information of the networks.}} 
\item \emph{Equilibrium analysis of the auction with allocative externalities:} {{We provide rigorous analysis for an auction mechanism with positive allocative externalities that involves more than two bidders. To the best of our knowledge, this is the first work studying such a mechanism in auction theory. Moreover, our work introduces a methodology for modeling and analyzing the allocative dependencies that arise increasingly often in wireless systems.}}
\item \emph{Characterization of the optimal reserve rate:} We analyze the reserve rate that maximizes the LTE network's payoff, and investigate its relation with the LTE throughput. Through simulation, we show that the optimal reserve rate is non-increasing in the LTE's data rate discounting factor, and non-decreasing in the LTE throughput, the number of APOs, and the APOs' data rate discounting factor. 
\item \emph{Performance evaluation of the LTE/Wi-Fi coopetition framework:} Numerical results show that our framework achieves larger LTE's and APOs' payoffs comparing with {{a state-of-the-art benchmark scheme, which only considers the competition between LTE and APOs.}} In particular, our framework increases the LTE's payoff by $70\%$ on average when the LTE has a large throughput and a small data rate discounting factor. Furthermore, our framework leads to a close-to-optimal social welfare for a large LTE throughput.


\end{itemize}
\subsection{Related Work}

{{This paper is an extension of our conference paper \cite{haoran2016wiopt}, where we considered a basic model with two APOs. 
In this paper, we generalize the model by considering an arbitrary number of APOs, which substantially extends the scope of the paper and the applicability of the results, but also significantly complicates the analysis. 
Furthermore, in this paper, we investigate the impact of the number of APOs on the LTE provider's and the APOs' strategies, and compare our auction-based scheme with a state-of-the-art benchmark scheme through simulation. We also extensively discuss the generalization of our work to more complicated scenarios (\emph{e.g.}, multi-LTE scenario).}}

{{Several recent studies focused on the spectrum sharing problems for the LTE unlicensed technology.}} 
{{Cano \emph{et al.} in \cite{canousing} and Zhang \emph{et al.} in \cite{zhanglte} discussed the major challenges for the LTE/Wi-Fi coexistence.}} 
{{References \cite{cavalcante2013performance,rupasinghe2014licensed} provided performance evaluations for the LTE/Wi-Fi coexistence.}} 
{{Li \emph{et al.} in \cite{li2015modelingJ} applied stochastic geometry to characterize the main performance metrics (\emph{e.g.,} SINR coverage probability) for the neighboring LTE and Wi-Fi networks in the unlicensed spectrum. 
Jeon \emph{et al.} in \cite{jeon2014lte} applied a fluid network model to analyze the interference between the LTE and Wi-Fi. 
Chen \emph{et al.} in \cite{chen2016cellular} jointly considered the Wi-Fi data offloading and the spectrum sharing between the LTE and Wi-Fi. 
Cano \emph{et al.} in \cite{cano2016fair} addressed the fair coexistence problem for general scheduled and random access transmitters that share the same channel.}} 
{{Cano \emph{et al.} in \cite{cano2015coexistence} studied the LTE network's channel access probability in the CSAT mechanism to ensure the fairness between LTE and Wi-Fi.}} Zhang \emph{et al.} in \cite{zhangmodeling} proposed a new LBT-based MAC protocol that allows LTE to friendly coexist with Wi-Fi. 
{{Guan \emph{et al.} in \cite{guan4cu}} investigated the LTE provider's joint channel selection and fractional spectrum access problem with the consideration of the fairness between LTE and Wi-Fi.} 
{{Zhang \emph{et al.} in \cite{zhang2015hierarchical} analyzed the spectrum sharing among multiple LTE providers in the unlicensed spectrum through a hierarchical game.}} 
{{However, these studies did not consider the cooperation between LTE and Wi-Fi. We include the existing studies on LTE/Wi-Fi coexistence like \cite{cano2015coexistence,zhangmodeling} as part of our framework (\emph{i.e.}, in the competition mode), and also consider the new possibility of cooperation between LTE and Wi-Fi (\emph{i.e.}, in the cooperation mode).}}

In terms of the auction with \emph{allocative externalities}, the most relevant works are \cite{jehiel2000auctions} and \cite{bagwell2003case}. Jehiel and Moldovanu in \cite{jehiel2000auctions} provided a systematic study of the second-price forward auction with allocative externalities. They characterized the bidders' bidding strategies at the equilibrium for general payoff functions. However, they did not prove the uniqueness of the equilibrium strategies. Bagwell \emph{et al.} in \cite{bagwell2003case} studied a special example in the WTO system, where the retaliation rights were allocated through a first-price forward auction among different countries. The auction involves positive allocative externalities, and the authors showed the uniqueness of the countries' bidding strategies. 
Both \cite{jehiel2000auctions} and \cite{bagwell2003case} only studied two bidders in the auction. In contrast, we consider an auction with an arbitrary number of bidders, and show the impact of the number of bidders on the auction outcome. 
Furthermore, the bidders' equilibrium strategies have different expressions under different reserve rates announced by the auctioneer, which makes our analysis of the optimal reserve rate {{much more challenging than}} \cite{jehiel2000auctions} and \cite{bagwell2003case}.
\section{System Model}\label{sec:model}
\subsection{Basic Settings}
We consider a time-slotted system, where the length of each time slot corresponds to several minutes. We assume that the system is quasi-static, \emph{i.e.}, the system parameters (which involve mostly time average values) remain constant during each time slot, but can change over time slots. Our analysis focuses on the interaction between LTE and Wi-Fi networks in a single generic time slot.{\footnote{{Since the LTE unlicensed technology (time-division duplex mode) supports both the uplink and downlink transmissions \cite{zhanglte,Senza}, the LTE network is able to onload both the APOs' uplink and downlink traffic. Our framework works for both the uplink scenario (the networks only have uplink traffic) and downlink scenario (the networks only have downlink traffic). For example, in the uplink scenario, all throughputs in our model correspond to the networks' uplink throughputs. For the most general scenario, where the networks serve uplink and downlink traffic simultaneously, each network should choose its strategy by considering both the uplink and downlink transmissions, and we leave the analysis of this scenario as our future work.}}} 
We consider one LTE small cell network and a set ${\cal K}\triangleq \left\{1,2,\ldots,K\right\}$ ($K\ge2$) of Wi-Fi access points. 
The LTE small cell network is owned by an LTE provider,{\footnote{{In Section \ref{subsec:discussion:multiLTE}, we will discuss the extension to the scenario where there are multiple LTE providers.}}} and the $k$-th ($k\in{\cal K}$) Wi-Fi access point is owned by APO $k$. We assume that the APOs occupy different unlicensed channels so that they do not interfere with each other. We use channel $k$ to represent the channel occupied by APO $k$. 
The LTE small cell network has a larger coverage area than the Wi-Fi access points \cite{Qualcomm,forum}. Furthermore, it can work in one of the $K$ channels, and cause interference to the corresponding access point in the channel.{\footnote{For ease of exposition, we use ``LTE provider'' and ``LTE network'' interchangeably. Similarly, we use ``APO'' and ``access point'' interchangeably.}} The assumption that the APOs occupy different channels simplifies the problem and helps us gain key insights into the proposed auction framework. {{In Section \ref{subsec:discussion:APOshare}, we will discuss the extension to the scenario where different APOs can share the same channel.}} 

{\bf APOs' Rates:} {{We consider fully loaded APOs,{\footnote{{Since the length of each time slot corresponds to several minutes, we assume that a network has enough traffic to serve during a time slot and will not complete its service within a time slot. This assumption simplifies the problem, and helps us understand the fundamental benefit of organizing an auction to onload the Wi-Fi traffic to the LTE network. Many papers made similar saturation assumptions to analyze the network performance \cite{bianchi2000performance,cali2000dynamic,wu2002performance}. {{In the future work, we will study the scenario where the networks do not have full loads, and can precisely predict their traffic loads in the next few minutes.}}}}} and use $r_k$ to denote the throughput that APO $k\in{\cal K}$ can achieve to serve its users when it \emph{exclusively} occupies channel $k$ (without the interference from the LTE network).}} 
{{The value of $r_k$ in the time slot that we are interested in is the private information of APO $k$. The LTE provider and the other $K-1$ APOs only know the probability distribution of $r_k$. Specifically, we assume that $r_k$ is a continuous random variable drawn from interval $\left[r_{\min},r_{\max}\right]$ ($r_{\min},r_{\max}\ge0$), and follows a probability distribution function (PDF) $f\left(\cdot\right)$ and a cumulative distribution function (CDF) $F\left(\cdot\right)$.{\footnote{We assume that all $r_k$ ($k\in{\cal K}$) follow the same distribution, and hence both functions $f\left(\cdot\right)$ and $F\left(\cdot\right)$ are independent of index $k$. We will study problem with the non-identical variable $r_k$ in our future work.}} Moreover, we assume that $f\left(\cdot\right)>0$ for all $r\in\left[r_{\min},r_{\max}\right]$.}} 


{\bf LTE's Dual Modes:} {{We consider a fully loaded LTE network,}} and assume that it achieves a channel independent throughput of $R_{\rm LTE}>0$ when it \emph{exclusively} occupies one of the $K$ channels (without the interference from the APOs).{\footnote{{As we will show in the analysis, the APOs make their decisions based on the LTE provider's reserve rate $C$ instead of the throughput $R_{\rm LTE}$. In other words, the APOs do not need to know the value of $R_{\rm LTE}$. Therefore, we do not need to assume a probability distribution of $R_{\rm LTE}$ to model the APOs' knowledge of $R_{\rm LTE}$.}}} The LTE provider can operate its network in one of the following modes:

-\emph{competition mode:} the LTE provider randomly chooses each channel $k\in{\cal K}$ with an equal probability and coexists with APO $k$. {{Since our main focus is the design of the auction framework, the LTE provider simply coexists with APO $k$ based on a typical coexistence mechanism (\emph{e.g.}, CSAT or LBT) and setting (\emph{e.g.}, the LTE's backoff time and transmission length in the LBT)}}. The co-channel interference decreases both the data rates of the LTE provider and the corresponding APO. We use ${\delta^{\rm LTE}}\in\left(0,1\right)$ and ${\eta^{\rm APO}}\in\left(0,1\right)$ to denote the LTE's and the APO's data rate discounting factors, respectively;{\footnote{{Based on \cite{Google,cavalcante2013performance,rupasinghe2014licensed}, the data rate reduction of the APO due to the co-channel interference is much heavier than that of the LTE. Hence, factor ${\eta^{\rm APO}}$ is usually smaller than ${\delta^{\rm LTE}}$. The values of ${\eta^{\rm APO}}$ and ${\delta^{\rm LTE}}$ depend on the concrete coexistence mechanisms and settings. For example, the study in \cite{Google} showed that ${\eta^{\rm APO}}$ ranges from $0.1$ to $0.5$ given different LTE off time under the CSAT mechanism. 
In this work, we assume that the LTE provider adopts the same mechanism (\emph{e.g.}, CSAT or LBT) and settings (\emph{e.g.}, LTE off time in CSAT) when coexisting with any APO. Hence, the LTE provider has the same discounting factor ${\delta^{\rm LTE}}$ for the coexistence with any APO, and the APOs have the same discounting factor ${\eta^{\rm APO}}$.}}}

-\emph{cooperation mode:} the LTE provider reaches an agreement with APO $k\in{\cal K}$, where APO $k$ stops transmission and the LTE provider exclusively occupies channel $k$. In this case, there is no co-channel interference, and the LTE provider's data rate is simply $R_{\rm LTE}$. As a compensation, the LTE provider will serve APO $k$'s users with a guaranteed data rate $r_{\rm pay}\in\left[0,R_{\rm LTE}\right]$.{\footnote{{{{{According to \cite{Cisco2016}, in 2015, $88$\% of global mobile devices are the mobile phones (including the smartphones, non-smartphones, and phablets), which have the cellular interfaces. Only $12$\% of global mobile devices are the tablets, laptops, and other devices that may not have the cellular interfaces. Therefore, we assume that all the mobile devices served by the APOs (during the considered time slot) have the cellular interfaces and hence can be onloaded to the LTE network if needed. In Section \ref{subsec:discussion:dumb}, we will discuss the extension to the scenario where some mobile devices (\emph{e.g.}, laptops) do not have the cellular interfaces.}}}}}} 
The remaining $K-1$ APOs occupy their own channels, and are not interfered by the LTE provider. Which APO the LTE provider chooses to cooperate with and what the value $r_{\rm pay}$ should be will be determined through a reverse auction design in the next subsection.

\subsection{Second-Price Reverse Auction Design}\label{subsec:auctionframe}
We design a second-price reverse auction, where the LTE provider is the auctioneer (buyer) and the APOs are the bidders (sellers). 
The auction is held at the beginning of each time slot. 
The private type of APO $k$ is $r_k$ (\emph{i.e.}, the data rate when it exclusively occupies channel $k$), {{and APO $k$'s item for sale is the right of onloading APO $k$'s traffic. When the LTE provider obtains the item from APO $k$, the LTE provider can onload APO $k$'s traffic and exclusively occupy channel $k$.}} 
Since we assume that the LTE provider cannot occupy more than one channel at the same time, the LTE provider is only interested in obtaining one item from one of the APOs.{\footnote{{Since the LTE unlicensed technology is still in an early stage of development, the existing relevant experiments and studies focused on the situation where the LTE network can only utilize a single unlicensed channel \cite{Google,forum,Qualcomm}. In the future, it is likely that the LTE network can aggregate multiple unlicensed channels through the carrier aggregation technology \cite{zhang2010carrier}.}}} Different from the conventional reverse auction where the auctioneer pays the winner money to obtain the item, here the LTE serves the winning APO's users with the rate $r_{\rm pay}$ as the payment.

\begin{figure}[t]
  \centering
  \includegraphics[scale=0.5]{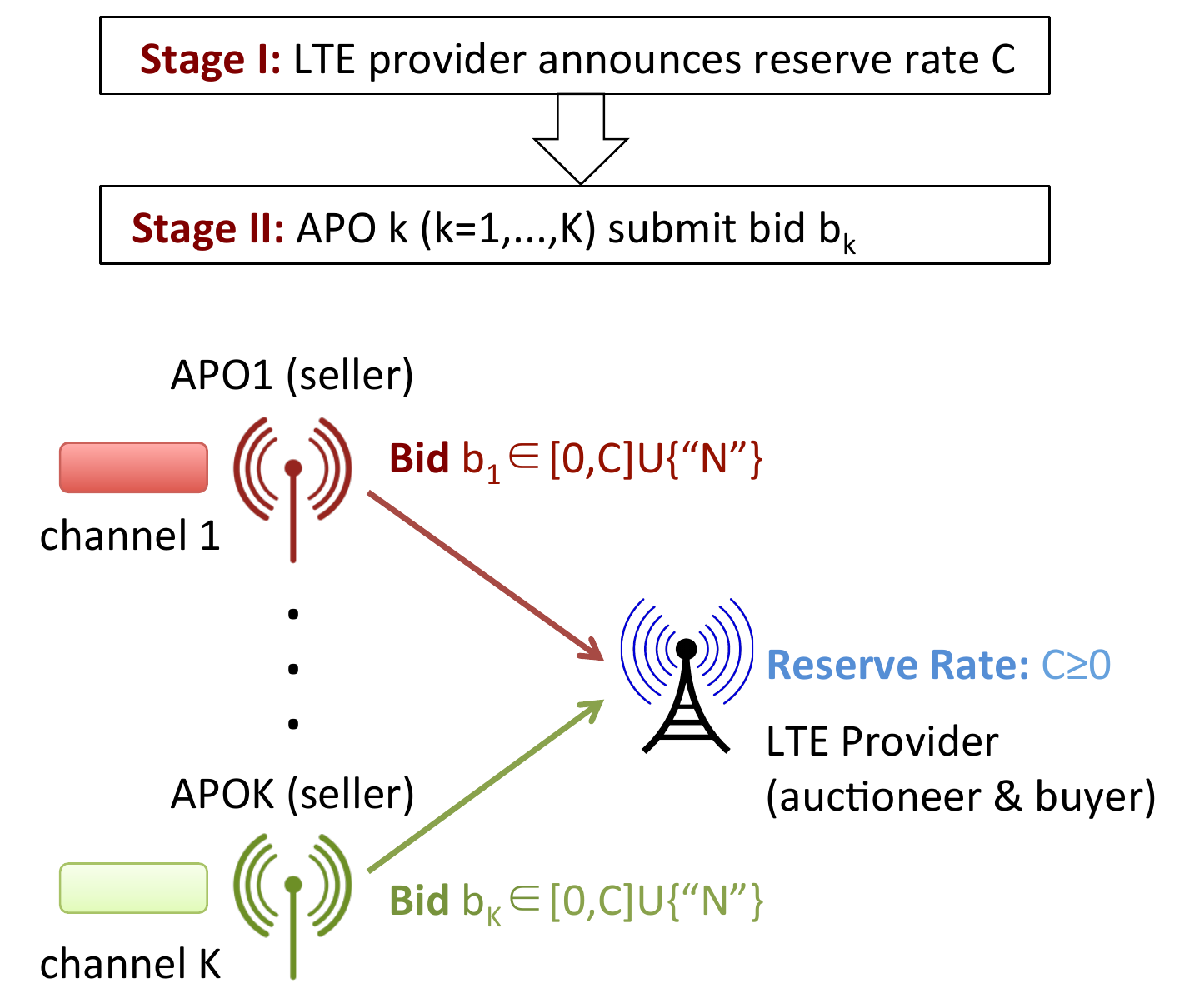}
  \vspace{-0.2cm}
  \caption{Illustration of The Reverse Auction.}
  \label{fig:auction}
  \vspace{-0.7cm}
\end{figure}

{\bf Reserve Rate and Bids:} {{In Stage I of the auction,}} the LTE provider announces its reserve rate $C\in\left[0,\infty\right)$, which corresponds to the maximum data rate that it is willing to accept to serve the winning APO's users. {{In Stage II of the auction,}} after observing the reserve rate $C$, APO $k$ submits a bid $b_k\in\left[0,C\right]\cup\left\{``\rm N \textquotedblright\right\}$: (a) $b_k\in\left[0,C\right]$ indicates the data rate that APO $k$ requests the LTE provider to serve APO $k$'s users; (b) $b_k=``\rm N \textquotedblright$ means that APO $k$ does not want to sell its item (\emph{i.e.}, the right of {{onloading APO $k$'s traffic}}) to the LTE provider.{\footnote{{{If APO $k$ bids any value greater than the reserve rate $C$, the LTE provider will not cooperate with APO $k$ based on the definition of $C$. Hence, any bid greater than $C$ leads to the same result to APO $k$. In order to facilitate the description, we use $``\rm N \textquotedblright$ to represent any bid greater than $C$.}} Intuitively, if the reserve rate $C$ is very small, APO $k$ is more likely to bid $``\rm N \textquotedblright$. In this case, APO $k$ can achieve an expected data rate {{(considering all possible auction results)}} higher than that when onloading the users to the LTE provider.}} We define the vector of APOs' bids as ${\bm b}\triangleq \left(b_k,\forall k\in{\cal K}\right)$. 
The auction design is illustrated in Fig. \ref{fig:auction}.


{\bf Auction Outcomes:} Next we discuss the auction outcomes based on the different values of ${\bm b}$ and $C$. For ease of exposition, we define the comparison between $``\rm N \textquotedblright$ and any bid $b_k$ as
\begin{align}
\min\left\{{``\rm N \textquotedblright},b_k\right\}=\left\{\begin{array}{ll}
{b_k,} & {\rm if~}{b_k\in\left[0,C\right],}\\
{{``\rm N \textquotedblright},} & {{\rm if~}{b_k={``\rm N \textquotedblright}}.}
\end{array} \right.
\end{align} 
Furthermore, we use ${\cal I}_{\min}$ to denote the set of APOs with the minimum bid, and define it as 
\begin{align}
{\cal I}_{\min}\triangleq \left\{i\in{\cal K}:i=\arg\min_{k\in{\cal K}}{b_k} \right\}.
\end{align}

The auction has the following possible outcomes:

(a) When $\left|{\cal I}_{\min}\right|=1$,{\footnote{Condition $\left|{\cal I}_{\min}\right|=1$ implies $\min_{k\in{\cal K}}{b_k} \in \left[0,C\right]$ as we have $K\ge2$ APOs.}} then APO $i= \arg\min_{k\in{\cal K}}{b_k}$ is the winner, and leaves channel $i$ to the LTE provider. The LTE provider works in the \emph{cooperation mode} and exclusively occupies channel $i$. Furthermore, the LTE serves APO $i$'s users with a rate $r_{\rm pay}=\min\left\{C,b_1,\ldots,b_{i-1},b_{i+1},\ldots,b_{K}\right\}$, which is the lowest rate among the reserve rate and all the {\emph{other}} APOs' bids, based on the rule of the second-price auction. {{In this case, the allocated rate $r_{\rm pay}$ is greater than the winning APO's bid (\emph{i.e.}, $\min_{k\in{\cal K}}{b_k}$)}};

(b) When $\min_{k\in{\cal K}}{b_k} \in \left[0,C\right]$ and $\left|{\cal I}_{\min}\right|>1$, the LTE provider works in the \emph{cooperation mode}, randomly chooses an APO from set ${\cal I}_{\min}$ with the probability $\textstyle \frac{1}{\left|{\cal I}_{\min}\right|}$ to exclusively occupy the corresponding channel, and serves the APO's users with a rate $r_{\rm pay}=\min_{k\in{\cal K}}{b_k}$. {{In this case, the allocated rate $r_{\rm pay}$ equals the winning APO's bid}};

(c) When $\min_{k\in{\cal K}}{b_k} =``\rm N \textquotedblright$,{\footnote{In this case, all APOs bid $``\rm N \textquotedblright$.}} the LTE provider works in the \emph{competition mode}, randomly chooses one of the $K$ channels with the probability $\textstyle\frac{1}{K}$, and shares the channel with the corresponding APO.{\footnote{{{Because the LTE provider does not have the private information $r_k$, it cannot differentiate the channels.}} We consider a specific protocol where the LTE provider randomly accesses each channel with an equal probability in the \emph{competition mode}.}}
\vspace{-0.4cm}
\subsection{LTE Provider's Payoff}\label{subsec:LTEpayoff}
Based on the summary of auction outcomes in the last subsection, we can write $r_{\rm pay}$ as a function of $\bm b$ and $C$:


\vspace{-0.2cm}
\begin{align}
\nonumber
& r_{\rm pay}\left({{\bm b},C}\right)= \\
& \left\{\begin{array}{ll}
{\min\!\left\{\!C,\!\min_{k\ne i,k\in{\cal K}}b_k\right\}\!,} &  {\rm if~}{\left|{\cal I}_{\min}\right|=1,}\\
{\min_{k\in{\cal K}}{b_k},} &  {\rm if~}{\min_{k\in{\cal K}}\!{b_k}\! \in\!\left[0,\!C\right] {~\!\rm and~\!\!}\left|{\cal I}_{\min}\right|\!\!>\!\!1,}\\
{0,} & {\rm if~}{\min_{k\in{\cal K}}{b_k} =``\rm N \textquotedblright.}
\end{array} \right.\label{equ:rpay}
\end{align}

We define the LTE provider's payoff as the data rate that it can allocate to its own users, and compute it as:{\footnote{Notice that $\min_{k\in{\cal K}}{b_k} \in \left[0,C\right]$ contains two possible situations: (i) $\left|{\cal I}_{\min}\right|=1$; (ii) $\min_{k\in{\cal K}}{b_k} \in\left[0,C\right] {~\rm and~}\left|{\cal I}_{\min}\right|>1$.}} 
\begin{align}
\Pi^{\rm LTE}\left({\bm b},C\right)=\left\{\begin{array}{ll}
{R_{\rm LTE}-r_{\rm pay}\left({{\bm b},C}\right),} & {\rm if~}{\min_{k\in{\cal K}}{b_k} \in\left[0,C\right],}\\
{{\delta^{\rm LTE}} R_{\rm LTE},} & {\rm if~}{\min_{k\in{\cal K}}{b_k} =``\rm N \textquotedblright.}
\end{array} \right.\label{equ:LTEpayoff}
\end{align}
Equation (\ref{equ:LTEpayoff}) captures two possible situations: 
(a) when the minimum bid lies in $\left[0,C\right]$, the LTE provider works in the \emph{cooperation mode}, exclusively occupies a channel, and obtains a total data rate of $R_{\rm LTE}$. Since the LTE provider needs to allocate a rate of $r_{\rm pay}\left({{\bm b},C}\right)$ to the winning APO's users, its payoff is $R_{\rm LTE}-r_{\rm pay}\left({{\bm b},C}\right)$; 
(b) when all APOs bid $``{\rm N} \textquotedblright$, the LTE provider works in the \emph{competition mode}, and ${\delta^{\rm LTE}}\in\left(0,1\right)$ captures the discount in the LTE provider's data rate due to the interference from the Wi-Fi APO in the same channel.
\vspace{-0.4cm}
\subsection{APOs' Payoffs and Allocative Externalities}
We define the payoff of APO $k\in{\cal K}$ as the data rate that its users receive: when APO $k$ cooperates with the LTE provider, these users are served by the LTE provider; otherwise, they are served by APO $k$. Based on the summary of auction outcomes in Section \ref{subsec:auctionframe} and the definition of $r_{\rm pay}\left({\bm b},C\right)$ in (\ref{equ:rpay}), we summarize APO $k$'s expected payoff as follows:
\begin{align}
\nonumber
& \Pi^{\rm APO}_{k}\left({\bm b},C\right)=\\
& \left\{\begin{array}{ll}
{r_k,} & {\rm if~}{b_k>\min_{j\in{\cal K}}{b_j},}\\
{\frac{1}{\left|{\cal I}_{\min}\right|}r_{\rm pay}\left({\bm b},C\right)\!+\!\frac{\left|{\cal I}_{\min}\right|-1}{\left|{\cal I}_{\min}\right|}r_k,} & {\rm if~}{b_k\!=\! \min_{j\in{\cal K}}{b_j}\!\in\!\left[0,C\right],}\\
{\frac{K-1+{\eta^{\rm APO}}}{K}r_k,} & {\rm if~}{\min_{j\in{\cal K}}{b_j}=``{\rm N} \textquotedblright.}
\end{array} \right.\label{equ:APOpayoff}
\end{align}

Equation (\ref{equ:APOpayoff}) summarizes three possible situations: (a) when $b_k>\min_{j\in{\cal K}}{b_j}$, the LTE provider exclusively occupies a channel from one of the APOs (other than APO $k$) with the minimum bid. As a result, APO $k$ can exclusively occupy its own channel $k$, and serve its users with rate $r_k$; (b) when $b_k= \min_{j\in{\cal K}}{b_j}\in\left[0,C\right]$, the LTE provider cooperates with APO $k$ and one of the other APOs with the minimum bid with the probability ${\frac{1}{\left|{\cal I}_{\min}\right|}}$ and the probability $1-{\frac{1}{\left|{\cal I}_{\min}\right|}}$ ($1\le {\left|{\cal I}_{\min}\right|}\le K$), respectively. Hence, APO $k$'s users receive rate $r_{\rm pay}\left({\bm b},C\right)$ and rate $r_k$ with the probability ${\frac{1}{\left|{\cal I}_{\min}\right|}}$ and the probability $1-{\frac{1}{\left|{\cal I}_{\min}\right|}}$, respectively. In this case, the expected data rate that APO $k$'s users receive is ${\frac{1}{\left|{\cal I}_{\min}\right|}r_{\rm pay}\left({\bm b},C\right)+\frac{\left|{\cal I}_{\min}\right|-1}{\left|{\cal I}_{\min}\right|}r_k,}$; (c) when $\min_{j\in{\cal K}}{b_j}=``{\rm N} \textquotedblright$, there is no winner in the auction, and the LTE provider randomly chooses one of the $K$ channels to coexist with the corresponding APO. With the probability ${\frac{1}{K}}$, APO $k$ coexists with the LTE provider and has a data rate of ${{\eta^{\rm APO}}}r_k$; with the probability $1-\frac{1}{K}$, APO $k$ has a data rate of $r_k$ by exclusively occupying channel $k$. In this case, the expected data rate that APO $k$'s users receive is $\frac{K-1+{\eta^{\rm APO}}}{K}r_k$. 

We note that APO $k$ does not win the auction in either of the following two cases: $b_k>\min_{j\in{\cal K}}{b_j}$ and $\min_{j\in{\cal K}}{b_j}=``{\rm N} \textquotedblright$. However, the APO $k$'s payoff is different in these two cases: it obtains a payoff of $r_k$ when $b_k>\min_{j\in{\cal K}}{b_j}$, and achieves a smaller payoff of $\frac{K-1+{\eta^{\rm APO}}}{K}r_k$ when $\min_{j\in{\cal K}}{b_j}=``{\rm N} \textquotedblright$. That is to say, even if APO $k$ does not win the auction, it is more willing to see the other APOs winning (\emph{i.e.}, $b_k>\min_{j\in{\cal K}}{b_j}$) rather than losing the auction (\emph{i.e.}, $\min_{j\in{\cal K}}{b_j}=``{\rm N} \textquotedblright$). This shows \emph{positive allocative externalities} of the auction, which make our problem substantially different from conventional auction problems. At the equilibrium of the conventional second-price auction, bidders bid truthfully according to their private values, regardless of other bidders' valuations. With allocative externalities in our problem, when APO $k$ evaluates its payoff when losing the auction, it needs to consider whether the other APOs win the auction or not. Hence, the distributions of the other APOs' valuations (types) affect APO $k$'s strategy. As we will show in the following sections, this leads to a special structure of APOs' bidding strategies at the equilibrium, and bidding truthfully is no longer a dominate strategy.

\begin{table}[t]\small
\centering
\caption{Main Notations}\label{table:notation}
\begin{tabular}{|c|p{6cm}|}
\hline
{\minitab[c]{${\cal K},K$}} & {The set of APOs and its cardinality}\\
\hline
{\minitab[c]{$r_k$}} & {APO $k$'s {throughput without interference (private valuation, also called \emph{type})}}\\
\hline
{\minitab[c]{$r_{\min}$, $r_{\max}$}} & {Lower and upper bounds of $r_k$, $k\in{\cal K}$}\\
\hline
{\minitab[c]{$f\left(\cdot\right)$, $F\left(\cdot\right)$}} & {PDF and CDF of $r_k$, $k\in{\cal K}$}\\
\hline
{\minitab[c]{$R_{\rm LTE}$}} & {LTE provider's throughput without interference}\\
\hline
{\minitab[c]{${\eta^{\rm APO}}$}} & {APOs' data rate discounting factor}\\
\hline
{\minitab[c]{${\delta^{\rm LTE}}$}} & {LTE provider's data rate discounting factor}\\
\hline
{\minitab[c]{$C$}} & {LTE provider's reserve rate (\emph{decision variable})}\\
\hline
{\minitab[c]{$b_k$}} & {APO $k$'s bid (\emph{decision variable})}\\
\hline
{\minitab[c]{${\cal I}_{\min}$}} & {The set of APOs with the minimum bid}\\
\hline
{\minitab[c]{$\Pi^{\rm LTE}\left({\bm b},C\right)$}} & {LTE provider's payoff}\\
\hline
{\minitab[c]{$r_{\rm pay}\left({\bm b},C\right)$}} & {Data rate LTE allocates to the winning APO}\\
\hline
{\minitab[c]{$\Pi^{\rm APO}_{k}\left({\bm b},C\right)$}} & {APO $k$'s payoff}\\
\hline
\end{tabular}
\vspace{-0.5cm}
\end{table}


We summarize the main notations in Table \ref{table:notation}. For the parameters and distributions that characterize the APOs, $r_k$ is APO $k$'s private information, and the remaining information, \emph{i.e.}, $K,r_{\min},r_{\max},f\left(\cdot\right),F\left(\cdot\right),$ and ${\eta^{\rm APO}}$, is publicly known to all the APOs and the LTE provider. For the parameters that characterize the LTE provider, \emph{i.e.}, $R_{\rm LTE}$ and ${\delta^{\rm LTE}}$, as we will see in later sections, they will not affect the APOs' strategies. Therefore, they can be either known or unknown to the APOs.

{{Next we analyze the auction by backward induction. In Section \ref{sec:stageII:APO}, we analyze the APOs' equilibrium strategies in Stage II, given the LTE provider's reserve rate $C$ in Stage I. In Section \ref{sec:stageI:LTE}, we analyze the LTE provider's equilibrium reserve rate $C^*$ in Stage I by anticipating the APOs' equilibrium strategies in Stage II.}}

\vspace{-0.2cm}
\section{{{Stage II: APOs' Equilibrium Bidding Strategies}}}\label{sec:stageII:APO}
{{In this section, we assume that the reserve rate $C$ of the LTE provider in Stage I is given, and analyze the APOs' equilibrium strategies in Stage II. 
In Section \ref{subsec:definitionEQ}, we define the equilibrium for the APOs under a given $C$. 
In Sections \ref{subsec:stageII:1}, \ref{subsec:lowC}, \ref{subsec:stageII:3}, and \ref{subsec:stageII:4}, we analyze the APOs' equilibrium strategies by considering different intervals of $C$. 
In Section \ref{sec:comparison}, we summarize the results for the APOs' equilibrium strategies.}}
\vspace{-0.3cm}
\subsection{{Definition of Symmetric Bayesian Nash Equilibrium}}\label{subsec:definitionEQ}
We focus on the symmetric Bayesian Nash equilibrium (SBNE), which is defined as follows. 
\begin{definition}
Under a reserve rate $C$, a bidding strategy function $b^*\left(r,C\right)$, $r\in\left[r_{\min},r_{\max}\right]$, constitutes a symmetric Bayesian Nash equilibrium if 
relation (\ref{equ:defineEQ}) holds 
for all $s_k\in \left[0,C\right]\cup\left\{{``{\rm N}\textquotedblright}\right\}$, all $r_k\in  \left[r_{\min},r_{\max}\right]$, and all $k\in{\cal K}$.
\end{definition}

\begin{figure*}[t]
\begin{align}
\nonumber
& \mathbb{E}_{{\bm r}_{-k}}\left\{\Pi^{\rm APO}_k\left(\left(b^*\left(r_1,C\right),\ldots,b^*\left(r_{k-1},C\right),b^*\left(r_k,C\right),b^*\left(r_{k+1},C\right),\ldots,b^*\left(r_K,C\right)\right),C\right)|r_k\right\}\ge\\
& \mathbb{E}_{{\bm r}_{-k}}\left\{\Pi^{\rm APO}_k\left(\left(b^*\left(r_1,C\right),\ldots,b^*\left(r_{k-1},C\right),s_k,b^*\left(r_{k+1},C\right),\ldots,b^*\left(r_K,C\right)\right),C\right)|r_k\right\},\label{equ:defineEQ}
\end{align}
\hrulefill
\end{figure*}

Since it is the symmetric equilibrium, all the APOs apply the same bidding strategy function $b^*\left(r,C\right)$ at the equilibrium. The left hand side of inequality (\ref{equ:defineEQ}) stands for APO $k$'s expected payoff when it bids $b^*\left(r_k,C\right)$. The expectation is taken with respect to ${\bm r}_{-k}\triangleq \left(r_j,\forall j\ne k,j\in{\cal K}\right)$, which denotes all the other APOs' types and is unknown to APO $k$. Inequality (\ref{equ:defineEQ}) implies that APO $k\in{\cal K}$ cannot improve its expected payoff by unilaterally changing its bid from $b^*\left(r_k,C\right)$ to any $s_k\in \left[0,C\right]\cup\left\{{``{\rm N}\textquotedblright}\right\}$.
\vspace{-0.4cm}
\subsection{{{APOs' Equilibrium When $C\in\left[r_{\min},r_{\max}\right)$}}}\label{subsec:stageII:1}
We assume that the reserve rate $C$ is given from $\left[r_{\min},r_{\max}\right)$,{\footnote{{We first analyze the case where $C\in\left[r_{\min},r_{\max}\right)$, because it has the most complicated equilibrium analysis. We can apply a similar analysis approach in this section to the other cases.}}} and show the unique form of bidding strategy that constitutes an SBNE. We first introduce the following lemma (the proofs of all lemmas and theorems can be found in the appendix{{}}).
\begin{lemma}\label{lemma:rT} 
The following equation admits at least one solution $r$ in $\left(C,r_{\max}\right)$:
\begin{multline}
\sum_{n=1}^{K-1}{\binom{K-1}{n} \left(F\left(r\right)-F\left(C\right)\right)^n\left(1-F\left(r\right)\right)^{K-1-n}\frac{C-r}{n+1}} \\
+\left(1-F\left(r\right)\right)^{K-1}\left(C-\frac{K-1+{\eta^{\rm APO}}}{K}r\right)=0,\label{equ:rT}
\end{multline}
where $F\left(\cdot\right)$ is the CDF of random variable $r_k$, $k\in{\cal K}$. We denote the solutions $r$ in $\left(C,r_{\max}\right)$ as $r_1^t\left(C\right), r_2^t\left(C\right),\ldots, r_M^t\left(C\right)$, where $M=1,2,\ldots,$ is the number of solutions.
\end{lemma}

Based on the definition of $r_1^t\left(C\right), r_2^t\left(C\right),\ldots, r_M^t\left(C\right)$ in Lemma \ref{lemma:rT}, we introduce the following theorem.
\begin{theorem}\label{theorem:equilibrium}
Consider an $r_T\left(C\right)\in\left(C,r_{\max}\right)$ that belongs to the set of $\left\{r_1^t\left(C\right), r_2^t\left(C\right),\ldots, r_M^t\left(C\right)\right\}$, then the following bidding strategy $b^*$ constitutes an SBNE:
\begin{align}
{b^*}\left(r_k,C\right)=\left\{\begin{array}{ll}
{{\rm any\!~value\!~in} \left[0,\!r_{\min}\right]\!,} & {\rm if~}{r_k=r_{\min},}\\
r_k, & {{\rm if~}{r_k\in\left(r_{\min},C\right]},}\\
C, & {{\rm if~}{ r_k\in\left(C,r_T\left(C\right)\right)},}\\
C{\rm~or~}``{\rm N}\textquotedblright, &  {\rm if~}{r_k=r_T\left(C\right),}\\
``{\rm N} \textquotedblright, & {{\rm if~} r_k\in\left(r_T\left(C\right),r_{\max}\right],}
\end{array} \right.\label{equ:equilibrium}
\end{align}
for all $k\in{\cal K}$.
\end{theorem}
We illustrate the structure of strategy $b^*$ in Fig. \ref{fig:S1}, in which we notice that

\vspace{-3mm}
\begin{figure}[h]
  \centering
  \includegraphics[scale=0.4]{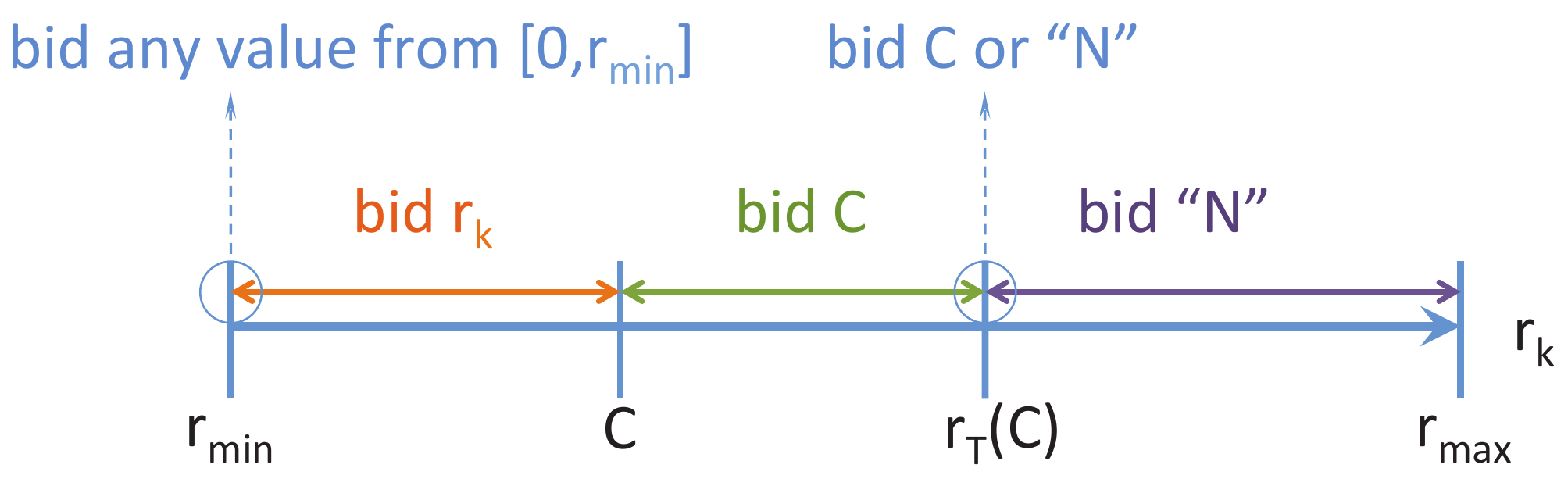}
  \vspace{-2mm}
  \caption{Bidding Strategy Structure at SBNE When $C\in\left[r_{\min},r_{\max}\right)$.}
  \label{fig:S1}
  \vspace{-4mm}
\end{figure}

(a) For an APO $k$ with type $r_k\in\left(r_{\min},C\right]$, it bids $r_k$. In other words, APO $k$ requests the LTE provider to serve APO $k$'s users with at least the rate that APO $k$ can achieve by exclusively occupying channel $k$;

(b) For an APO $k$ with type $r_k\in\left(C,r_T\left(C\right)\right)$, it bids $C$. Since $C<r_k$, the data rate APO $k$ requests from the LTE provider is smaller than the rate that APO $k$ achieves by exclusively occupying channel $k$. Recall that the feasible bid should be from $\left[0,C\right]\cup \left\{``{\rm N} \textquotedblright\right\}$. If APO $k$ bids $``{\rm N} \textquotedblright$, there is a chance that all the other APOs also bid $``{\rm N} \textquotedblright$, which makes the LTE provider work in the \emph{competition mode} and leads to a payoff of $\frac{K-1+{\eta^{\rm APO}}}{K}r_k$ to APO $k$ based on (\ref{equ:APOpayoff}). In order to avoid such a situation, APO $k$ would bid $C$, and ensure that its payoff is at least $C$;{\footnote{Specifically, based on (\ref{equ:rpay}), if APO $k$ bids $C$ and wins the auction, its payoff will be $C$; if APO $k$ bids $C$ but loses the auction, its payoff will be $r_k>C$.}}

(c) For an APO $k$ with type $r_k\in\left(r_T\left(C\right),r_{\max}\right]$, it bids $``{\rm N} \textquotedblright$. Similar as case (b), there is a chance that all the other APOs also bid $``{\rm N} \textquotedblright$, and APO $k$ obtains a payoff of $\frac{K-1+{\eta^{\rm APO}}}{K}r_k$. However, with $r_k\in\left(r_T\left(C\right),r_{\max}\right]$, the value $\frac{K-1+{\eta^{\rm APO}}}{K}r_k$ is already large enough so that there is no need for APO $k$ to lower its bid from $``{\rm N} \textquotedblright$ to any value from $\left[0,C\right]$.

There are two special points in (\ref{equ:equilibrium}): 

(d) For an APO $k$ with $r_k=r_{\min}$, it has the same payoff if it bids any value from $\left[0,r_{\min}\right]$. This is because with probability one, APO $k$ wins the auction.{\footnote{Notice that for any APO $j\ne k, j\in{\cal K}$, the probability that $r_j=r_{\min}$ is zero based on the continuous distribution of $r_j$. In other words, with probability one, $r_j$ is from the interval $\left(r_{\min},r_{\max}\right]$. Based on (\ref{equ:equilibrium}), APO $j\ne k$ bids from $\left(r_{\min},C\right]\cup \left\{``{\rm N} \textquotedblright\right\}$ and APO $k$ wins the auction.}} From (\ref{equ:rpay}) and (\ref{equ:APOpayoff}), APO $k$'s payoff is $\min\left\{C,b_1,\ldots,b_{k-1},b_{k+1},\ldots,b_{K}\right\}$, which does not depend on APO $k$'s bid $b_k$ {and is always no smaller than $r_{\min}$};

(e) For an APO $k$ with $r_k=r_T\left(C\right)$, it has the same expected payoff under bids $C$ and $``{\rm N} \textquotedblright$.

It is easy to show that $b^*\left(r_k,C\right)$ in (\ref{equ:equilibrium}) is not a dominant strategy for the APOs. For example, if APO $k$'s type $r_k\in\left(C,r_T\left(C\right)\right)$ and $\min_{j\in{\cal K},j\ne k}b_j=C$, bidding $``{\rm N} \textquotedblright$ generates a larger payoff to APO $k$ than bidding $b^*\left(r_k,C\right)=C$. This result is different from that of the conventional second-price auction, where bidding the truthful valuation constitutes an equilibrium, and is also the weakly dominant strategy for the bidders.

Notice that equation (\ref{equ:rT}) may admit multiple solutions, \emph{i.e.}, $M>1$. Based on Theorem \ref{theorem:equilibrium}, each solution $r_m^t$, $m=1,2,\ldots,M$, corresponds to a strategy $b^*$ defined in (\ref{equ:equilibrium}).

In the following theorem, we show the unique form of bidding strategy under an SBNE.

\begin{theorem}\label{theorem:combine:unique}
The strategy function in (\ref{equ:equilibrium}) is the unique form of bidding strategy that constitutes an SBNE.
\end{theorem}

The sketch of the proof is as follows: first, we show the necessary conditions that a bidding strategy needs to satisfy to constitute an SBNE; second, we show that the function in (\ref{equ:equilibrium}) is the only function that satisfies all these conditions. We leave the detailed proof in Appendices \ref{appendix:sec:preliminary} and \ref{appendix:sec:theorem2}. 

\subsection{APOs' Equilibrium When $C\in\left[0,\frac{K-1+{\eta^{\rm APO}}}{K}r_{\min}\right]$}\label{subsec:lowC}
We assume that the reserve rate $C$ is given from interval $\left[0,\frac{K-1+{\eta^{\rm APO}}}{K}r_{\min}\right]$, and summarize the form of the bidding strategy at the SBNE in the following theorem.

\begin{theorem}\label{theorem:lowC}
When $C\in\left[0,\frac{K-1+{\eta^{\rm APO}}}{K}r_{\min}\right)$, there is a unique SBNE, where ${b^*}\left(r_k,C\right)=``{\rm N}\textquotedblright$, $k\in{\cal K}$, for all $r_k\in\left[r_{\min},r_{\max}\right]$; when $C=\frac{K-1+{\eta^{\rm APO}}}{K}r_{\min}$, 
a strategy function constitutes an SBNE if and only if it is in the following form:
\begin{align}
{b^*}\left(r_k,C\right)\!=\! \left\{\begin{array}{ll}
{\rm \!any\!~value\!~in} \left[0,\!C\right]{\rm{~\!or~\!}}``{\rm N} \textquotedblright\!, & {{\rm if~}{ r_k=r_{\min}},}\\
``{\rm N}\textquotedblright, & {{\rm if~} r_k\!\in\left(r_{\min},\!r_{\max}\right].}
\end{array} \right.\label{equ:lowCbid}
\end{align}
\end{theorem}

When $C\in\left[0,\frac{K-1+{\eta^{\rm APO}}}{K}r_{\min}\right]$, the LTE provider only wants to allocate a limited data rate to the winning APO's users. In this case, the APOs bid $``{\rm N} \textquotedblright$ with probability one.{\footnote{In Theorem \ref{theorem:lowC}, when $C=\frac{K-1+{\eta^{\rm APO}}}{K}r_{\min}$, the APO with type $r_{\min}$ can bid any value. However, the probability for an APO to have the type $r_{\min}$ is zero due to the continuous distribution of $r$.}}

\vspace{-0.2cm}
\subsection{APOs' Equilibrium When $C\in\left(\frac{K-1+{\eta^{\rm APO}}}{K}r_{\min},r_{\min}\right)$}\label{subsec:stageII:3}
We assume that the reserve rate $C$ is given from interval $\left(\frac{K-1+{\eta^{\rm APO}}}{K}r_{\min},r_{\min}\right)$, and show that the bidding strategy that constitutes an SBNE has a unique form. First, we introduce the following lemma.

\begin{lemma}\label{lemma:rX} 
The following equation admits at least one solution $r$ in $\left(r_{\min},r_{\max}\right)$:
\begin{multline}
\sum_{n=1}^{K-1}{\binom{K-1}{n} F^n\left(r\right)\left(1-F\left(r\right)\right)^{K-1-n}\frac{C-r}{n+1}}\\
+\left(1-F\left(r\right)\right)^{K-1}\left(C-\frac{K-1+{\eta^{\rm APO}}}{K}r\right)=0,\label{equ:rX}
\end{multline}
where $F\left(\cdot\right)$ is the CDF of random variable $r_k$, $k\in{\cal K}$. We denote the solutions $r$ in $\left(r_{\min},r_{\max}\right)$ as $r_1^x\left(C\right),r_2^x\left(C\right),\ldots,r_L^x\left(C\right)$, where $L=1,2,\ldots,$ is the number of solutions.
\end{lemma}

Based on the definition of $r_1^x\left(C\right),r_2^x\left(C\right),\ldots,r_L^x\left(C\right)$ in Lemma \ref{lemma:rX}, we introduce the following theorem.

\begin{theorem}\label{theorem:middleC}
When $C\in\left(\frac{K-1+{\eta^{\rm APO}}}{K}r_{\min},r_{\min}\right)$, consider an $r_X\left(C\right)\in\left(r_{\min},r_{\max}\right)$ that belongs to the set of $\left\{r_1^x\left(C\right),r_2^x\left(C\right),\ldots,r_L^x\left(C\right)\right\}$, then the following bidding strategy $b^*$ constitutes an SBNE:
\begin{align}
{b^*}\left(r_k,C\right)=\left\{\begin{array}{ll}
C, & {{\rm if~}{ r_k\in\left[r_{\min},r_X\left(C\right)\right)},}\\
C{\rm~or~}``{\rm N} \textquotedblright, &  {\rm if~}{r_k=r_X\left(C\right),}\\
``{\rm N} \textquotedblright, & {{\rm if~} r_k\in\left(r_X\left(C\right),r_{\max}\right],}
\end{array} \right.\label{equ:simplyequilibrium}
\end{align}
where $k\in{\cal K}$. Furthermore, such a bidding strategy $b^*$ is the unique form of bidding strategy that constitutes an SBNE.

\end{theorem}
The bidding strategy in (\ref{equ:simplyequilibrium}) is similar to that in (\ref{equ:equilibrium}), except that here it only has two regions instead of three regions. Specifically, here there are no APOs that bid their types $r_k$. This is because here the reserve rate $C$ is smaller than $r_{\min}$, hence bidding any type $r_k\in\left[r_{\min},r_{\max}\right]$ is not feasible. We illustrate the structure of strategy function $b^*$ in Fig. \ref{fig:S2}.

\begin{figure}[h]
  \centering
  \includegraphics[scale=0.4]{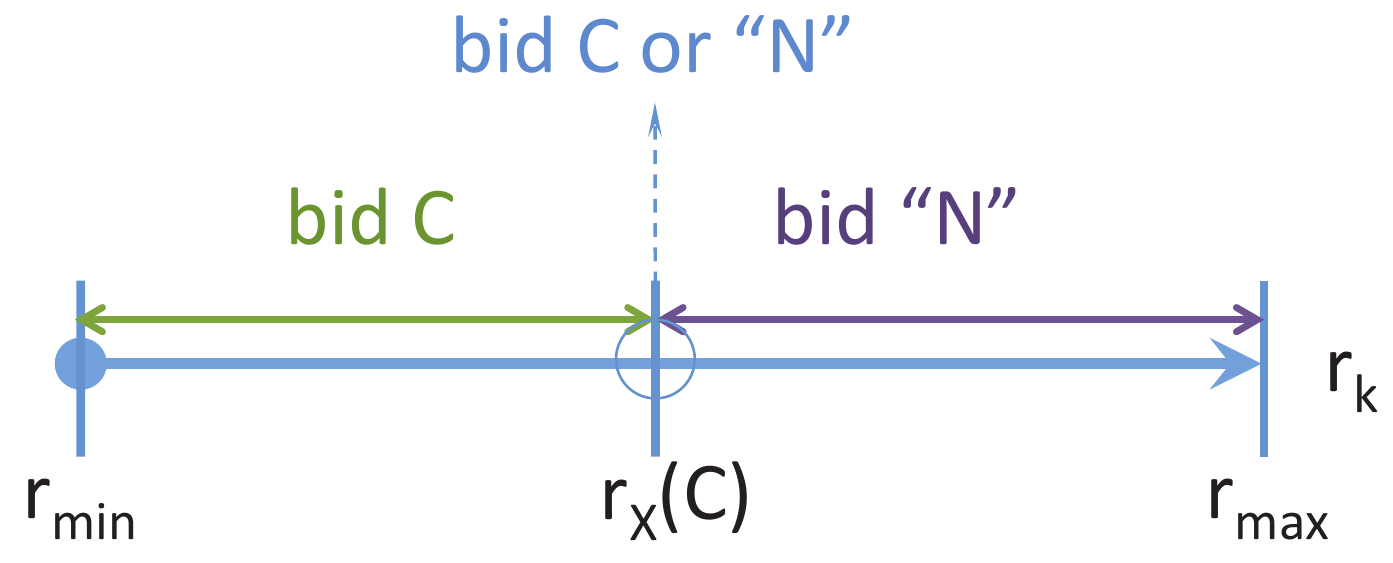}
  \caption{Bidding Strategy Structure at SBNE When $C\in\left(\frac{K-1+{\eta^{\rm APO}}}{K}r_{\min},r_{\min}\right)$.}
  \label{fig:S2}
  \vspace{-4mm}
\end{figure}

Similar as the equilibrium analysis for $C\in\left[r_{\min},r_{\max}\right)$, here equation (\ref{equ:rX}) may admit multiple solutions, \emph{i.e.}, $L>1$, in which case each solution $r_l^x$, $l=1,2,\ldots,L$, corresponds to a strategy $b^*$ defined in (\ref{equ:simplyequilibrium}).
\vspace{-4mm}
\subsection{APOs' Equilibrium When $C\in\left[r_{\max},\infty\right)$}\label{subsec:stageII:4}
We assume that the reserve rate $C$ is given from interval $C\in\left[r_{\max},\infty\right)$, and show the unique form of bidding strategy that constitutes an SBNE in the following theorem.
\begin{theorem}\label{theorem:highC}
When $C\in\left[r_{\max},\infty\right)$, a strategy function constitutes an SBNE if and only if it is in the following form ($k\in{\cal K}$):
\begin{align}
{b^*}\!\left(r_k,\!C\right)\!=\!\!\left\{\begin{array}{ll}
{{\rm \!any~\!value~\!in} \left[0,\!r_{\min}\right]\!,} & {\rm \!if~\!}{r_k=r_{\min},}\\
r_k, & {{\rm \!if~\!}{r_k\!\in\!\left(r_{\min},\!r_{\max}\right)\!},}\\
{\rm \!any~\!value~\!in}\left[r_{\max},\!C\right]\!\cup\!\left\{\!``{\rm N}\textquotedblright\right\}\!, &  {\rm \!if~\!}{r_k=r_{\max}.}\\
\end{array} \right.\label{equ:EQhighC}
\end{align}
\end{theorem}

When $C\in\left[r_{\max},\infty\right)$, the LTE provider is willing to allocate a large data rate to the winning APO's users. Based on (\ref{equ:EQhighC}), all APOs bid values from $\left[0,C\right]$ with probability one.{\footnote{Notice that the probability for an APO to have the type $r_{\max}$ is zero due to the continuous distribution of $r$.}}

\vspace{-0.5cm}
\subsection{{{Summary of APOs' Equilibriums}}}\label{sec:comparison}
\vspace{-0.2cm}

\begin{figure}[h]
  \centering
  \includegraphics[scale=0.38]{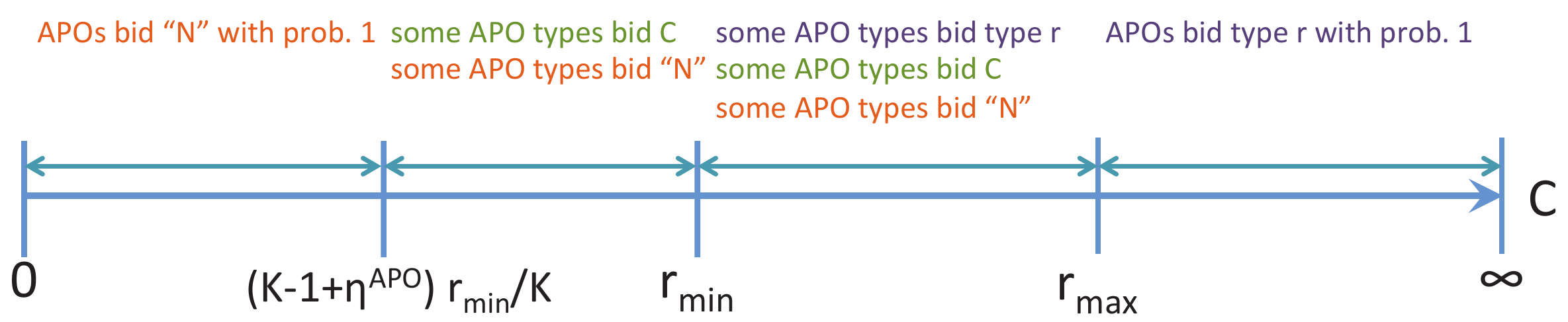}
  \caption{APOs' Strategies under Different $C$.}
  \label{fig:fourC}
\end{figure}

{{Based on Sections \ref{subsec:stageII:1}, \ref{subsec:lowC}, \ref{subsec:stageII:3}, and \ref{subsec:stageII:4}, there is always a unique form of APO $k$'s bidding strategy $b^*\left(r_k,C\right)$ at the SBNE for any reserve rate $C\in\left[0,\infty\right)$.}} 
{{We summarize the APOs' strategies under different intervals of $C$ in Fig. \ref{fig:fourC}.}}{\footnote{{{When $C\in\left[r_{\min},r_{\max}\right)$, an APO with type $r_k=r_{\min}$ can bid any value from $\left[0,r_{\min}\right]$ at the equilibrium based on (\ref{equ:equilibrium}). Since the probability for an APO to have the type $r_{\min}$ is zero, the strategy of this particular APO type is not shown in Fig. \ref{fig:fourC}.}}}} 
We find that some APO types bid the reserve rate $C$ in Fig. \ref{fig:fourC} when $C\in\left(\frac{K-1+{\eta^{\rm APO}}}{K}r_{\min},r_{\max}\right)$. This is due to the unique feature of the auction with \emph{allocative externalities}: {{first, if none of the other APOs submits its bid from interval $\left[0,C\right]$, these types of APOs prefer to cooperate with the LTE provider rather than {{to interfere with}} the LTE in the \emph{competition mode}; second, if at least one of the other APOs submits its bid from interval $\left[0,C\right]$, these types of APOs prefer to occupy their own channels rather than to cooperate with the LTE provider, as the LTE will not generate interference to their channels in this case. The first reason motivates these APO types to bid from interval $\left[0,C\right]$, and the second reason motivates these APO types to reduce their chances of winning the auction as much as possible. As a result, these APO types bid the reserve rate $C$ at the equilibrium.}}

\section{{{Stage I: LTE Provider's Reserve Rate}}}\label{sec:stageI:LTE}
{{In this section, we analyze the LTE provider's optimal reserve rate by anticipating APOs' equilibrium strategies in Stage II. 
In Section \ref{subsec:stageI:0}, we define the LTE provider's expected payoff. 
In Section \ref{subsec:stageI:1}, we compute the LTE provider's expected payoff based on different intervals of $C$. In Section \ref{subsec:stageI:2}, we formulate the LTE provider's payoff maximization problem. In Section \ref{subsec:stageI:3}, we analyze the LTE provider's optimal reserve rate $C^*$.}}

\vspace{-4mm}
\subsection{{{Definition of LTE Provider's Expected Payoff}}}\label{subsec:stageI:0}
We first make the following assumption on the CDF of an APO's type.

\begin{assumption}\label{assumption:unique}
Under the cumulative distribution function $F\left(\cdot\right)$, (a) equation (\ref{equ:rT}) has a unique solution in $\left(C,r_{\max}\right)$, {i.e.}, $M=1$, and (b) equation (\ref{equ:rX}) has a unique solution in $\left(r_{\min},r_{\max}\right)$, {i.e.}, $L=1$.
\end{assumption}
Assumption \ref{assumption:unique} implies that $r_T\left(C\right)$ and $r_X\left(C\right)$ are unique. Such an assumption is mild. 
{{When $K=2$, we have proved that Assumption \ref{assumption:unique} holds for the uniform distribution. 
For a general $K$, we have run simulation and shown that Assumption \ref{assumption:unique} holds for both the uniform distribution and truncated normal distribution. 
The details of the proof and simulation can be found in Appendices \ref{appendix:sec:K2} and \ref{appendix:sec:Kgeneral}, respectively.}}


Based on Theorem \ref{theorem:equilibrium} and Theorem \ref{theorem:combine:unique},  the uniqueness of $r_T\left(C\right)$ implies the unique expression of APOs' bidding strategy $b^*$ for $C\in\left[r_{\min},r_{\max}\right)$. Similarly, from Theorem \ref{theorem:middleC}, the uniqueness of $r_X\left(C\right)$ implies the unique expression of strategy $b^*$ for $C\in\left(\frac{K-1+{\eta^{\rm APO}}}{K}r_{\min},r_{\min}\right)$. 

We define the LTE provider's expected payoff as
\begin{align}
{\bar \Pi}^{\rm LTE}\!\left(C\right)\! \triangleq \!\mathbb{E}_{\bm r}\left\{{\Pi^{\rm LTE}\!\left(\!\left(b^*\!\left(r_1,C\right),\ldots,b^*\!\left(r_K,C\right)\!\right)\!,\!C\right)\!}\right\}\!,\label{equ:expectedpayoff}
\end{align}
where ${\bm r}\triangleq \left(r_k,\forall k\in{\cal K}\right)$ denotes the types of all APOs, and $b^*\left(r_k,C\right)$, $k\in{\cal K}$, is given in (\ref{equ:equilibrium}), (\ref{equ:lowCbid}), (\ref{equ:simplyequilibrium}), and (\ref{equ:EQhighC}) based on the different intervals of $C$. 

\vspace{-4mm}
\subsection{{{Computation of LTE Provider's Expected Payoff}}}\label{subsec:stageI:1}
{{Since $b^*\left(r_k,C\right)$ in (\ref{equ:expectedpayoff}) has different expressions for four different intervals of $C$, we characterize ${\bar \Pi}^{\rm LTE}\left(C\right)$ based on these four intervals of $C$.}}

\subsubsection{$C\in\left[0,\frac{K-1+{\eta^{\rm APO}}}{K}r_{\min}\right]$}\label{subsec:LTEpayoffI}
The APOs submit their bids according to strategy $b^*$ in (\ref{equ:lowCbid}). It is easy to find that $b^*\left(r_k,C\right)=``{\rm N}\textquotedblright$ with probability one for all $k\in{\cal K}$, and hence the LTE provider always works in the \emph{competition mode}. Based on (\ref{equ:LTEpayoff}), we can compute ${\bar \Pi}^{\rm LTE}\left(C\right)$ as
\begin{align}
{\bar \Pi}^{\rm LTE}\left(C\right)={\delta^{\rm LTE}} R_{\rm LTE}.\label{equ:LTEpayoff:smallC}
\end{align}
\subsubsection{$C\in\left(\frac{K-1+{\eta^{\rm APO}}}{K}r_{\min},r_{\min}\right)$}\label{subsec:LTEpayoffII}
The APOs' bidding strategy is summarized in (\ref{equ:simplyequilibrium}). Hence, the probabilities for an APO with a random type to bid $C$ and $``{\rm N}\textquotedblright$ are $F\left(r_X\left(C\right)\right)$ and $1-F\left(r_X\left(C\right)\right)$, respectively. 
Therefore, we can compute ${\bar \Pi}^{\rm LTE}\left(C\right)$ as
\begin{align}
\nonumber
{\bar \Pi}^{\rm LTE}\left(C\right) = &\left(1-F\left(r_X\left(C\right)\right)\right)^{K} {\delta^{\rm LTE}} R_{\rm LTE} + \\
& \left(1-\left(1-F\left(r_X\left(C\right)\right)\right)^K \right)  \left(R_{\rm LTE}-C\right).\label{equ:LTEpayoff:relow}
\end{align}
That is to say: (a) when all the APOs bid $``{\rm N}\textquotedblright$, the LTE provider works in the \emph{competition mode}, and obtains a payoff of ${\delta^{\rm LTE}} R_{\rm LTE}$; (b) when at least one APO bids $C$, the LTE provider works in the \emph{cooperation mode}, and allocates a rate of $C$ to the winning APO's users. 
\subsubsection{$C\in\left[r_{\min},r_{\max}\right)$}\label{subsec:LTEpayoffIII}
The APOs' strategy is given in (\ref{equ:equilibrium}).
We can compute ${\bar \Pi}^{\rm LTE}\left(C\right)$ as
\begin{align}
\nonumber
{\bar \Pi}^{\rm LTE}\left(C\right) = & \left(1-F\left(r_T\left(C\right)\right)\right)^K {\delta^{\rm LTE}} R_{\rm LTE} + \\
& \left(1-\left(1-F\left(r_T\left(C\right)\right)\right)^K \right)  R_{\rm LTE}- {\bar r}_{\rm pay}\left(C\right).\label{equ:LTEpayoff:typicalC}
\end{align}
Here, ${\bar r}_{\rm pay}\left(C\right)$ is defined as the expected data rate that the LTE provider allocates to the winning APO's users, and is given as (the details of computing ${\bar r}_{\rm pay}\left(C\right)$ can be found in Appendix \ref{appendix:sec:rpay})
\begin{align}
\nonumber
{\bar r}_{\rm pay}\left(C\right)\triangleq & K\left(K-1\right) \int_{r_{\min}}^{C} r f\left(r\right) F\left(r\right)  \left(1-F\left(r\right)\right)^{K-2} dr \\
\nonumber
& +KCF\left(C\right)\left(1-F\left(C\right)\right)^{K-1} \\
&+C\left(\left(1-F\left(C\right)\right)^K-\left(1-F\left(r_T\left(C\right)\right)\right)^K\right).\label{equ:rpay:multi}
\end{align}
\subsubsection{$C\in\left[r_{\max},\infty\right)$}\label{subsec:LTEpayoffIV}
Based on (\ref{equ:EQhighC}), the APOs bid values from $\left[0,C\right]$ with probability one, and the LTE provider always works in the \emph{cooperation mode}. Then we can compute ${\bar \Pi}^{\rm LTE}\left(C\right)$ as 
\begin{align}
\nonumber
& {\bar \Pi}^{\rm LTE}\left(C\right)=R_{\rm LTE}\\
& -K\left(K-1\right) \int_{r_{\min}}^{r_{\max}} r  f\left(r\right) F\left(r\right) \left(1-F\left(r\right)\right)^{K-2} dr.\label{equ:LTEpayoff:largeC}
\end{align}

\subsection{LTE Provider's Payoff Maximization Problem}\label{subsec:stageI:2}
Based on ${\bar \Pi}^{\rm LTE}\left(C\right)$ derived in Section \ref{subsec:stageI:1}, we can verify that ${\bar \Pi}^{\rm LTE}\left(C\right)$ is continuous for $C\in\left[0,\infty\right)$. The LTE provider determines the optimal reserve rate by solving
{{
\begin{align}
\max_{C\in\left[0,\infty\right)} {~\bar \Pi}^{\rm LTE}\left(C\right) {~~~~~~\rm s.t.~~} b_{\max}\left(C\right) \le R_{\rm LTE},\label{equ:optABC}
\end{align}}}
where we define 
\begin{align}
b_{\max}\left(C\right) \triangleq \max \left\{ b^*\left(r_k,C\right)\in\left[0,C\right]:{r_k\in\left[r_{\min},r_{\max}\right]} \right\},
\end{align}
which is the maximum possible bid (except $``{\rm N}\textquotedblright$) from the APOs at the SBNE under $C$. Constraint $b_{\max}\left(C\right) \le R_{\rm LTE}$ ensures that the LTE provider has enough capacity to satisfy the bid from the winning APO. 

\vspace{-0.3cm}
\subsection{LTE Provider's Optimal Reserve Rate}\label{subsec:stageI:3}
{{In the following theorem, we characterize the optimal reserve rate $C^*$ that solves problem (\ref{equ:optABC}) for a general distribution function $F\left(\cdot\right)$ that satisfies Assumption \ref{assumption:unique}.}}

\begin{theorem}\label{theorem:optimalCcase}
The LTE provider's optimal reserve rate $C^*$ satisfies the following properties:\\
(1) When $R_{\rm LTE}\le\frac{K-1+{\eta^{\rm APO}}}{K\left(1-{\delta^{\rm LTE}}\right)}r_{\min}$, {{$C^*$}} can be any value from $\left[0,\frac{K-1+{\eta^{\rm APO}}}{K}r_{\min}\right]$;\\
(2) When $\frac{K-1+{\eta^{\rm APO}}}{K\left(1-{\delta^{\rm LTE}}\right)}r_{\min}<\! R_{\rm LTE}\le \! r_{\max}$, {{$C^*$ can be chosen from}} $\left(\frac{K-1+{\eta^{\rm APO}}}{K}r_{\min},\!R_{\rm LTE}\right]$;\\
(3) When $R_{\rm LTE}\!>\!\max\left\{r_{\max},\frac{K-1+{\eta^{\rm APO}}}{K\left(1-{\delta^{\rm LTE}}\right)}r_{\min}\!\right\}$, {{$C^*$ can be chosen from}} $\left(\!\frac{K-1+{\eta^{\rm APO}}}{K}r_{\min},\!r_{\max}\!\right]$.
\end{theorem}

{{
When $R_{\rm LTE}\le\frac{K-1+{\eta^{\rm APO}}}{K\left(1-{\delta^{\rm LTE}}\right)}r_{\min}$, the LTE provider does not have enough capacity to satisfy any APO's request. Specifically, $R_{\rm LTE}\le \frac{K-1+{\eta^{\rm APO}}}{K\left(1-{\delta^{\rm LTE}}\right)}r_{\min}$ is equivalent to $\left(1-{\delta^{\rm LTE}}\right)R_{\rm LTE}\le \frac{K-1+{\eta^{\rm APO}}}{K}r_{\min}$. Here, $\left(1-{\delta^{\rm LTE}}\right)R_{\rm LTE}$ stands for the additional increase in the LTE network's capacity when it works in the \emph{cooperation mode}. Based on (\ref{equ:equilibrium}), (\ref{equ:lowCbid}), (\ref{equ:simplyequilibrium}), and (\ref{equ:EQhighC}), $\frac{K-1+{\eta^{\rm APO}}}{K}r_{\min}$ is the lower bound of the data rate that any APO with type in $\left(r_{\min},r_{\max}\right]$ may request from the LTE provider. Therefore, when $R_{\rm LTE}\le\frac{K-1+{\eta^{\rm APO}}}{K\left(1-{\delta^{\rm LTE}}\right)}r_{\min}$, the additional gain in the LTE network's capacity under cooperation cannot cover the request from any APO, and the LTE provider sets $C^*\in \left[0,\frac{K-1+{\eta^{\rm APO}}}{K}r_{\min}\right]$ to work in the \emph{competition mode}.

When $\frac{K-1+{\eta^{\rm APO}}}{K\left(1-{\delta^{\rm LTE}}\right)}r_{\min}<R_{\rm LTE}\le r_{\max}$, the LTE network's capacity can cover the requests from the APOs that bid small values. Hence, the LTE provider chooses $C^*$ above $\frac{K-1+{\eta^{\rm APO}}}{K}r_{\min}$ to accept these APOs' bids. Meanwhile, the LTE provider has to choose $C^*$ no larger than $R_{\rm LTE}$, otherwise it does not have enough capacity to satisfy the APOs that bid large values. 

When $R_{\rm LTE}\!>\!\max\left\{r_{\max},\frac{K-1+{\eta^{\rm APO}}}{K\left(1-{\delta^{\rm LTE}}\right)}r_{\min}\right\}$, since the maximum possible bid from the APOs is $r_{\max}$, the LTE provider always has enough capacity to satisfy the APOs' requests. In this case, the LTE provider chooses $C^*$ from $\left(\frac{K-1+{\eta^{\rm APO}}}{K}r_{\min},r_{\max}\right]$, and $C^*$ is no longer constrained by the LTE throughput $R_{\rm LTE}$.}}

{{Next we discuss the choice of $C^*$ based on Theorem \ref{theorem:optimalCcase}. When $R_{\rm LTE}\!\le\!\frac{K-1+{\eta^{\rm APO}}}{K\left(1-{\delta^{\rm LTE}}\right)}r_{\min}$, Theorem \ref{theorem:optimalCcase} indicates that any value from interval $\!\left[0,\!\frac{K-1+{\eta^{\rm APO}}}{K}r_{\min}\!\right]$ is the optimal $C^*$ for a general distribution function $F\left(\cdot\right)$.}} However, when $R_{\rm LTE}>\frac{K-1+{\eta^{\rm APO}}}{K\left(1-{\delta^{\rm LTE}}\right)}r_{\min}$, it is difficult to characterize the closed-form expression of $C^*$ even under a specific function $F\left(\cdot\right)$. 
This is because (i) it is difficult to solve equations (\ref{equ:rT}) and (\ref{equ:rX}) and obtain the closed-form expressions of $r_T\left(C\right)$ and $r_X\left(C\right)$, respectively, and (ii) the expression of ${\bar \Pi}^{\rm LTE}\left(C\right)$ in (\ref{equ:LTEpayoff:typicalC}) is complicated and hard to analyze. Therefore, we determine the optimal $C^*$ numerically for $R_{\rm LTE}>\frac{K-1+{\eta^{\rm APO}}}{K\left(1-{\delta^{\rm LTE}}\right)}r_{\min}$. Specifically, we have the following observation from the simulation.
\begin{observation}\label{observation:0}
${\bar \Pi}^{\rm LTE}\left(C\right)$ is strictly unimodal for $C\in\left(\frac{K-1+{\eta^{\rm APO}}}{K}r_{\min},r_{\max}\right]$.
\end{observation}

\begin{figure*}[t]
  \centering
  \begin{minipage}[t]{.49\linewidth}
  \centering
  \includegraphics[scale=0.3]{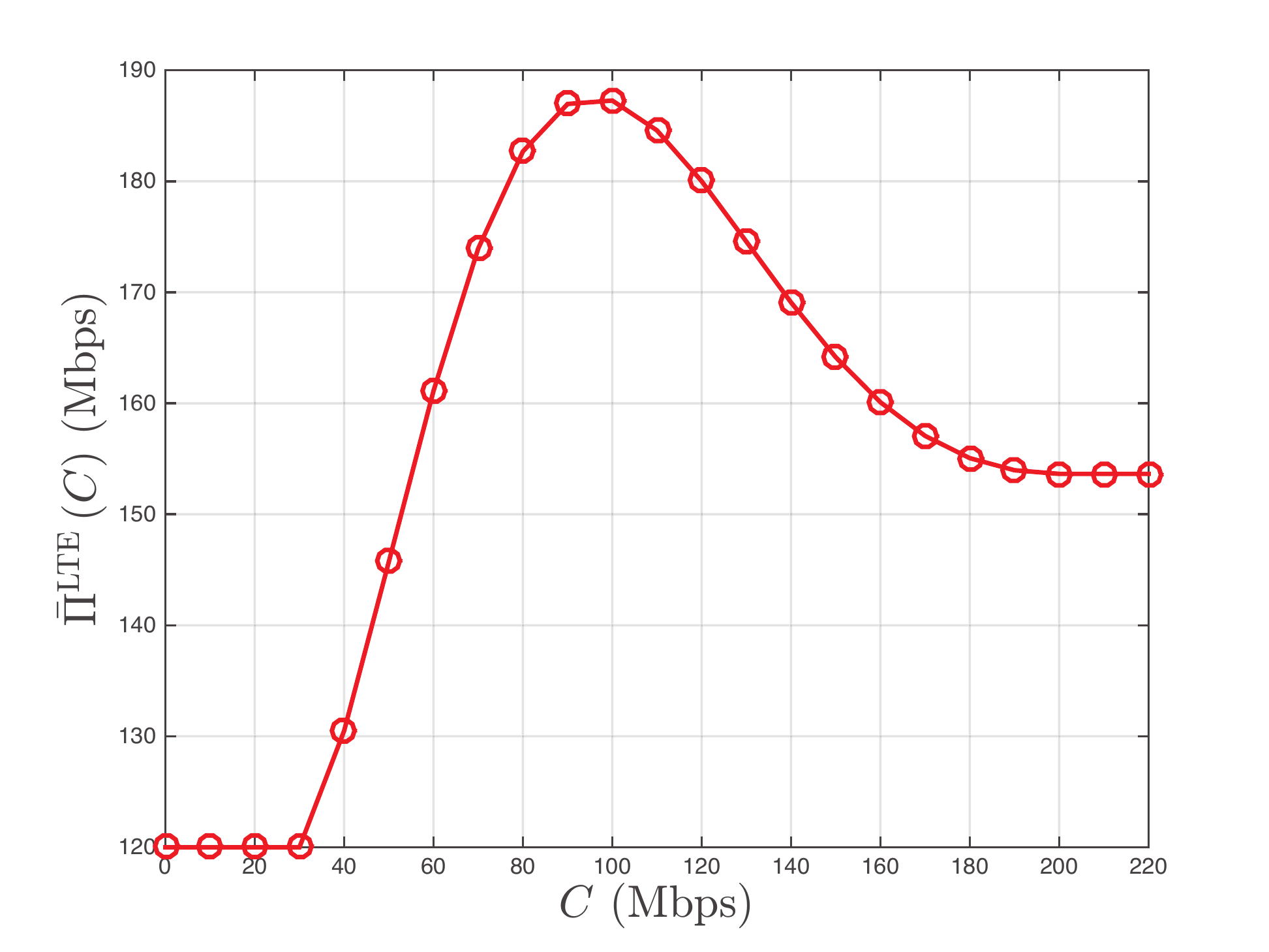}
  \vspace{-2mm}
  \caption{Example of Function ${\bar \Pi}^{\rm LTE}\left(C\right)$.}
  \label{fig:simu:0}
  \vspace{-5mm}
  \end{minipage}
  \begin{minipage}[t]{.49\linewidth}
  \centering
  \includegraphics[scale=0.3]{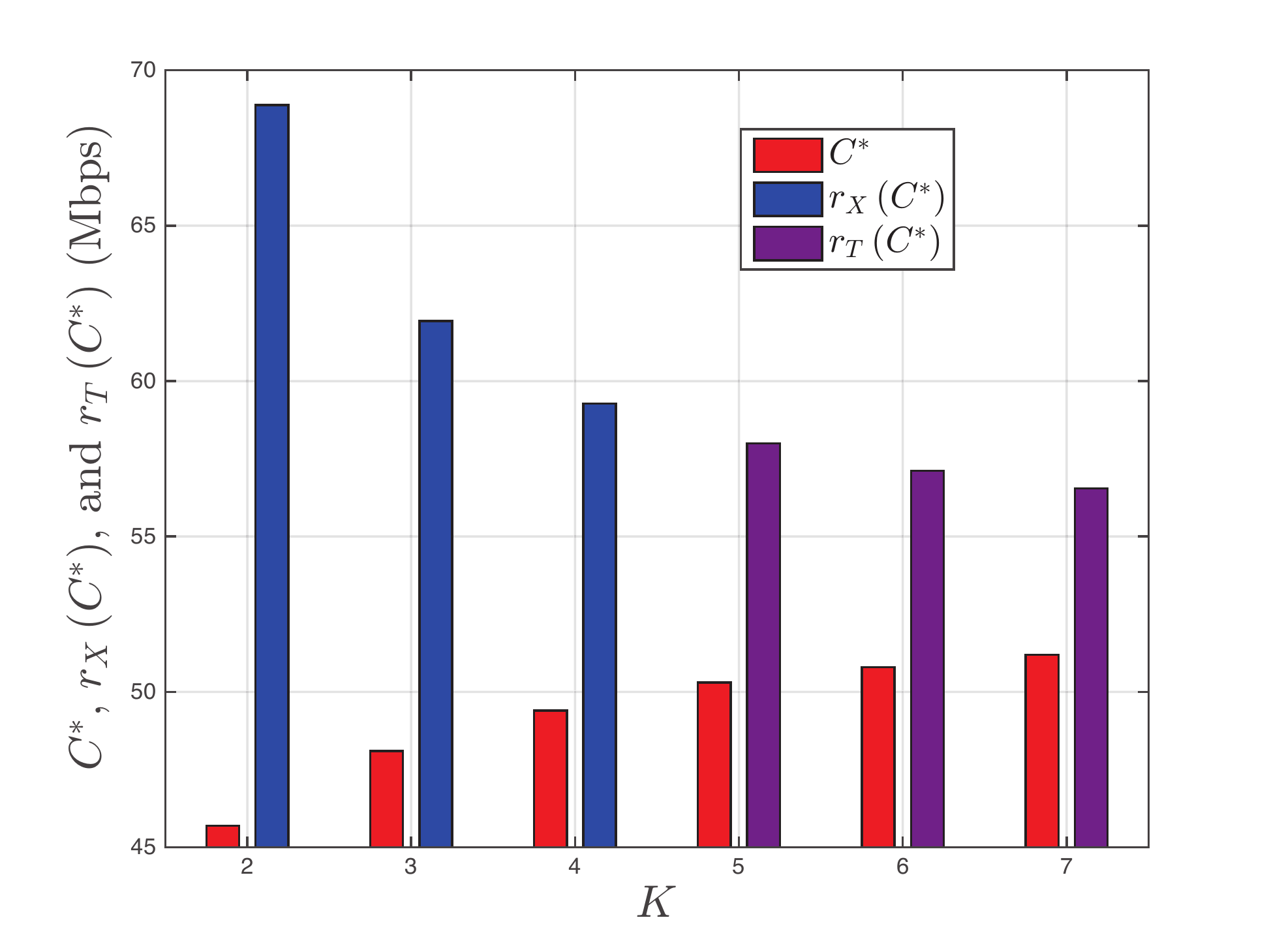}
  \vspace{-2mm}
  \caption{Impact of $K$ on LTE Provider's and APOs' Strategies.}
  \label{fig:simu:A}
  \vspace{-5mm}
  \end{minipage}
\end{figure*}
We have verified Observation \ref{observation:0} for the uniform distribution function $F\left(\cdot\right)$ and the truncated normal distribution function $F\left(\cdot\right)$. In Fig. \ref{fig:simu:0}, we illustrate an example of ${\bar \Pi}^{\rm LTE}\left(C\right)$, where $K=2$, ${\delta^{\rm LTE}}=0.4$, ${\eta^{\rm APO}}=0.3$, $R_{\rm LTE}=300{\rm~Mbps}$, and $r_k\in\left[50~{\rm Mbps},200~{\rm Mbps}\right]$ follows a truncated normal distribution.{\footnote{We choose $R_{\rm LTE}=300{\rm~Mbps}$ because the peak LTE throughput ranges from $250~{\rm Mbps}$ to $370~{\rm Mbps}$ based on \cite{canousing}. {{Moreover, we choose ${\eta^{\rm APO}}$ smaller than ${\delta^{\rm LTE}}$, because the degeneration of Wi-Fi's data rate due to the co-channel interference is usually heavier than that of the LTE, as we discussed in Section \ref{sec:model}.}}}} Based on Theorem \ref{theorem:optimalCcase} and Observation \ref{observation:0}, when $\frac{K-1+{\eta^{\rm APO}}}{K\left(1-{\delta^{\rm LTE}}\right)}r_{\min}<R_{\rm LTE}\le r_{\max}$ and $R_{\rm LTE}>\max\left\{r_{\max},\frac{K-1+{\eta^{\rm APO}}}{K\left(1-{\delta^{\rm LTE}}\right)}r_{\min}\right\}$, we can use the Golden Section method \cite{bertsekas1999nonlinear} on interval $\left(\frac{K-1+{\eta^{\rm APO}}}{K}r_{\min},R_{\rm LTE}\right]$ and interval $\left(\frac{K-1+{\eta^{\rm APO}}}{K}r_{\min},r_{\max}\right]$, respectively, to determine the optimal $C^*$.


\vspace{-0.3cm}
\section{Numerical Results}\label{sec:numerical}
In this section, we first investigate the impacts of the system parameters on the LTE's optimal reserve rate, the LTE's expected payoff, and the APOs' equilibrium strategies. 
{{Then we compare our auction-based spectrum sharing scheme with a state-of-the-art benchmark scheme. Specifically, we randomly pick the APOs, and implement both schemes. We compare several criteria (such as the LTE provider's payoff, the APOs' total payoff, and the social welfare) achieved by our auction-based scheme and the benchmark scheme.}}

\subsection{Influences of System Parameters}
\subsubsection{Influence of $K$} We first study the impact of the number of APOs $K$ on the LTE provider's and APOs' strategies. We choose $R_{\rm LTE}=95{\rm~Mbps}$, ${\delta^{\rm LTE}}=0.4$, and ${\eta^{\rm APO}}=0.3$, and assume that $r_k$, $k\in{\cal K}$, follows a truncated normal distribution. Specifically, we obtain the distribution of $r_k$ by truncating the normal distribution ${\cal N}\left(125~{\rm Mbps},2500~{{\rm Mbps}^2}\right)$ to interval $\left[50~{\rm Mbps},200~{\rm Mbps}\right]$. We change $K$ from $2$ to $7$, and determine the corresponding optimal reserve rate $C^*$ numerically based on the approach discussed in Section \ref{subsec:stageI:3}. 

We plot $C^*$ against $K$ in Fig. \ref{fig:simu:A}, and observe that $C^*$ increases with $K$. This is because that the probability of a particular APO being interfered by the LTE in the \emph{competition mode} decreases with the number of APOs. Hence, the APOs are less willing to cooperate with the LTE provider under a larger $K$, and the LTE provider needs to increase $C^*$ to attract the APOs.

In Fig. \ref{fig:simu:A}, we observe that $C^*\in\left(\frac{K-1+{\eta^{\rm APO}}}{K}r_{\min},r_{\min}\right)$ for $2\le K \le 4$. Based on (\ref{equ:simplyequilibrium}), in this case, APOs with types in $\left[r_{\min},r_X\left(C^*\right)\right)$ and $\left(r_X\left(C^*\right),r_{\max}\right]$ bid $C^*$ and $``{\rm N}\textquotedblright$, respectively. To study the impact of $K$ on the APOs' strategies, we plot $r_X\left(C^*\right)$ for $2\le K \le 4$ in Fig. \ref{fig:simu:A}. We observe that $r_X\left(C^*\right)$ decreases with $K$. This means that when $K$ increases from $2$ to $4$, more APOs bid $``{\rm N}\textquotedblright$ instead of $C^*$. On the other hand, we find that $C^*\in\left[r_{\min},r_{\max}\right)$ for $5\le K \le 7$. Based on (\ref{equ:equilibrium}), in this case, APOs with types in $\left[C^*,r_T\left(C^*\right)\right)$ and $\left(r_T\left(C^*\right),r_{\max}\right]$ bid $C^*$ and $``{\rm N}\textquotedblright$, respectively. We plot $r_T\left(C^*\right)$ for $5\le K \le 7$, and observe that $r_T\left(C^*\right)$ decreases with $K$. Since $C^*$ increases with $K$, it is easy to conclude that when $K$ increases from $5$ to $7$, fewer APOs bid $C^*$, and more APOs bid $``{\rm N}\textquotedblright$. Combining the observations for $2\le K \le 4$ and $5\le K \le 7$, we summarize that the increase of $K$ makes more APOs switch from bidding $C^*$ to bidding $``{\rm N}\textquotedblright$. The reason is that each APO has a smaller chance to be interfered by the LTE in the \emph{competition mode} under a larger $K$. Therefore, the APOs with large $r_k$ are less willing to cooperate with the LTE provider, and more APOs bid $``{\rm N}\textquotedblright$ instead of $C^*$.

We summarize the observations for Fig. \ref{fig:simu:A} as follows.
\begin{observation}\label{observation:1}
When the number of APOs increases, (i) the LTE provider's optimal reserve rate $C^*$ increases, and (ii) more APOs switch from bidding $C^*$ to bidding $``{\rm N}\textquotedblright$.
\end{observation}

\begin{figure*}[t]
  \centering
  \begin{minipage}[t]{.32\linewidth}
  \centering
  \includegraphics[scale=0.3]{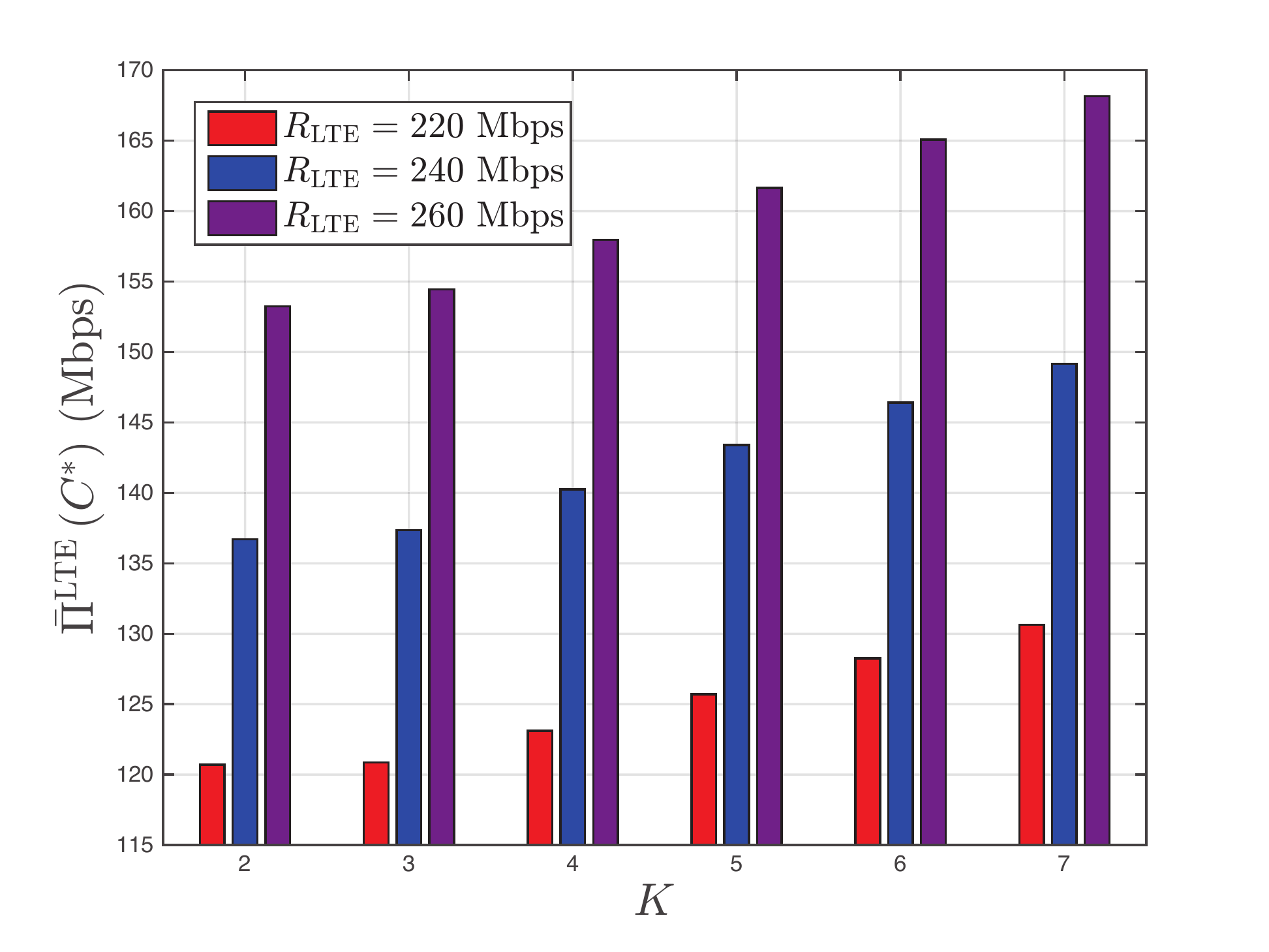}
  \vspace{-3mm}
   \caption{${\bar \Pi}^{\rm LTE}\left(C^*\right)$ (Large $R_{\rm LTE}$).}
  \label{fig:simu:B}
  \vspace{-5mm}
  \end{minipage}
  \begin{minipage}[t]{.32\linewidth}
  \includegraphics[scale=0.3]{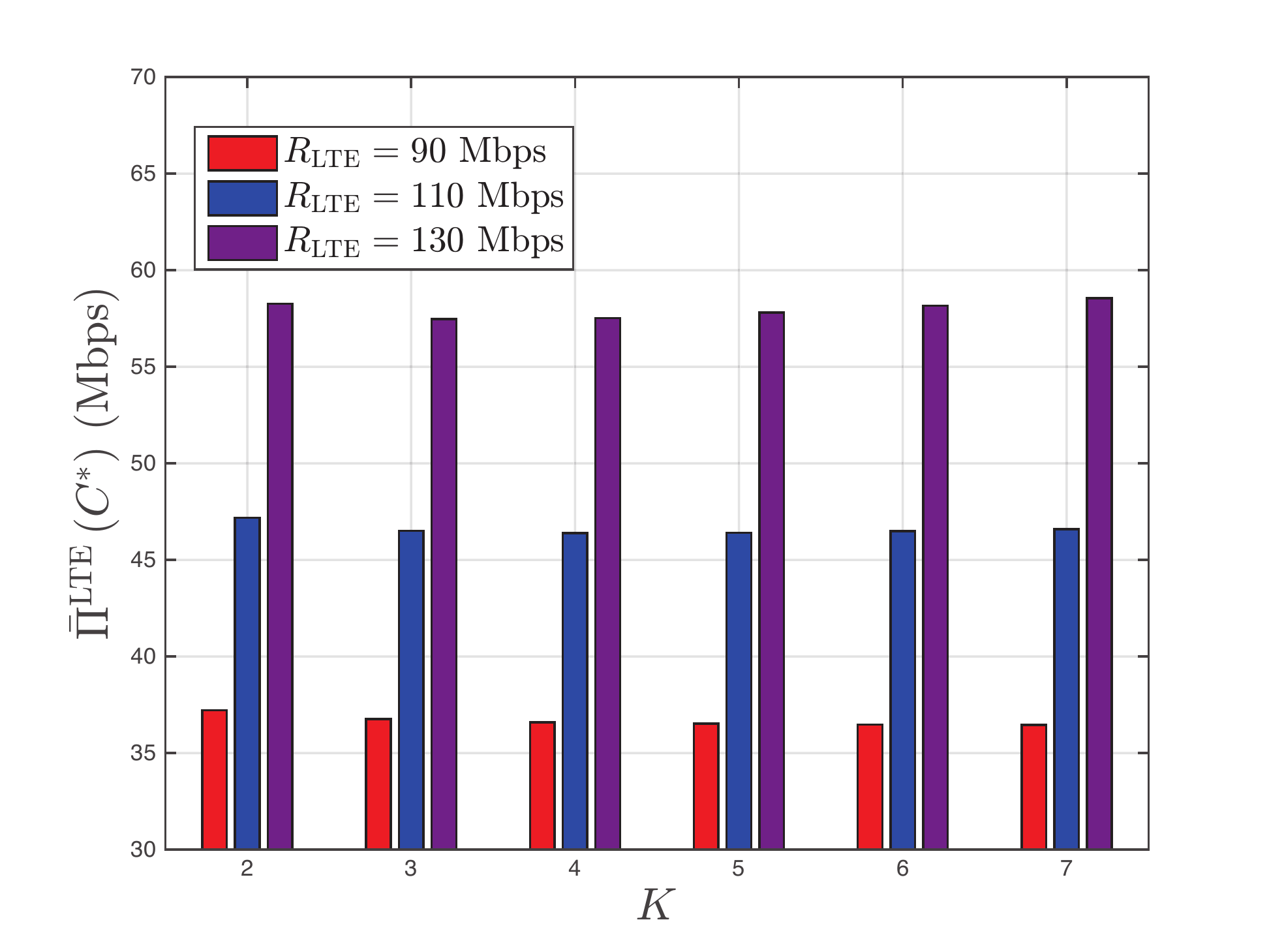}
  \vspace{-3mm}
  \caption{${\bar \Pi}^{\rm LTE}\left(C^*\right)$ (Small $R_{\rm LTE}$).}
  \label{fig:simu:C}
  \vspace{-5mm}
  \end{minipage}
  \begin{minipage}[t]{.32\linewidth}
  \centering
  \includegraphics[scale=0.3]{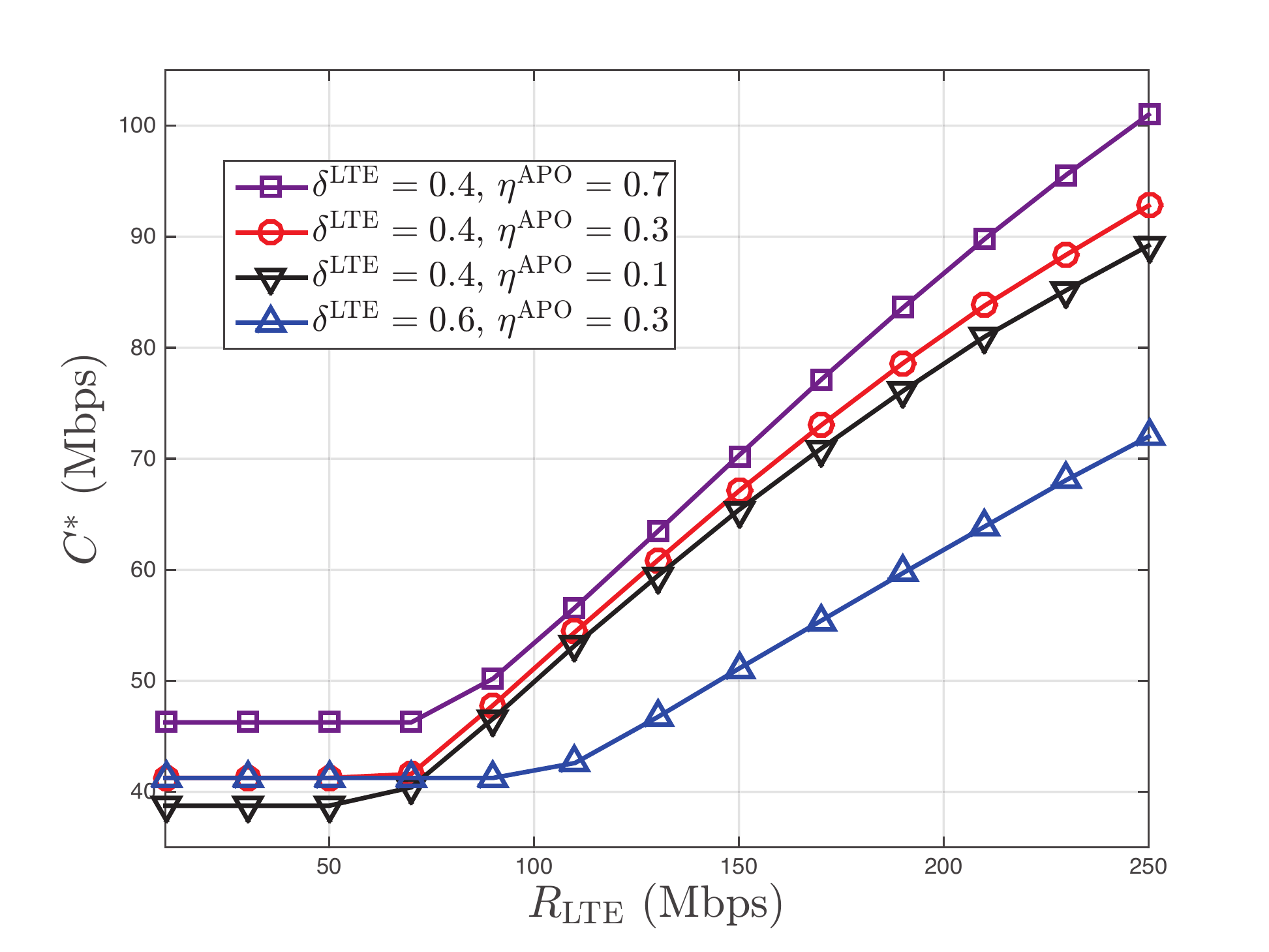}
  \vspace{-3mm}
  \caption{Impacts of $R_{\rm LTE}$, ${\delta^{\rm LTE}}$, and ${\eta^{\rm APO}}$ on $C^*$.}
  \label{fig:simu:D}
  \vspace{-5mm}
  \end{minipage}
\end{figure*}

Next we study the impact of $K$ on the LTE provider's expected payoff ${\bar \Pi}^{\rm LTE}\left(C^*\right)$. The settings of ${\delta^{\rm LTE}}$ and ${\eta^{\rm APO}}$, and the distribution of $r_k$ are the same as those in Fig. \ref{fig:simu:A}. We choose $R_{\rm LTE}=220{\rm~Mbps}$, $240{\rm~Mbps}$, and $260{\rm~Mbps}$, and plot the corresponding ${\bar \Pi}^{\rm LTE}\left(C^*\right)$ against $K$ in Fig. \ref{fig:simu:B}. We observe that ${\bar \Pi}^{\rm LTE}\left(C^*\right)$ increases with $K$ for these values of $R_{\rm LTE}$. Moreover, we choose $R_{\rm LTE}=90{\rm~Mbps}$, $110{\rm~Mbps}$, and $130{\rm~Mbps}$, and plot the corresponding ${\bar \Pi}^{\rm LTE}\left(C^*\right)$ against $K$ in Fig. \ref{fig:simu:C}. Different from Fig. \ref{fig:simu:B}, we find that ${\bar \Pi}^{\rm LTE}\left(C^*\right)$ does not significantly change with $K$ in Fig. \ref{fig:simu:C}. To understand the difference between Fig. \ref{fig:simu:B} and Fig. \ref{fig:simu:C}, we notice that the increase of $K$ has the following two opposite impacts on ${\bar \Pi}^{\rm LTE}\left(C^*\right)$: (i) the probability for the LTE provider to find an APO with a small bid increases, which potentially increases ${\bar \Pi}^{\rm LTE}\left(C^*\right)$; (ii) more APOs bid $``{\rm N}\textquotedblright$ instead of $C^*$ (Observation \ref{observation:1}), which potentially decreases ${\bar \Pi}^{\rm LTE}\left(C^*\right)$. In Fig. \ref{fig:simu:B}, the values of $R_{\rm LTE}$ are large, and the LTE provider can set large reserve rates $C^*$ to attract the APOs. In this situation, the interval of APO types that want to cooperate with the LTE provider is large, and impact (i) plays a dominant role. As a result, ${\bar \Pi}^{\rm LTE}\left(C^*\right)$ increases with $K$ in Fig. \ref{fig:simu:B}. On the other hand, the values of $R_{\rm LTE}$ are small in Fig. \ref{fig:simu:C}, and the LTE provider can only choose small reserve rates $C^*$. Hence, the interval of APO types that want to cooperate with the LTE provider is small. In this situation, impact (ii) becomes as important as impact (i). As a result, ${\bar \Pi}^{\rm LTE}\left(C^*\right)$ does not significantly change with $K$ in Fig. \ref{fig:simu:C}.

We summarize the following observations for Fig. \ref{fig:simu:B} and Fig. \ref{fig:simu:C}.
\begin{observation}\label{observation:2}
When the LTE provider has a large throughput $R_{\rm LTE}$, its expected payoff ${\bar \Pi}^{\rm LTE}\left(C^*\right)$ increases with $K$; otherwise, ${\bar \Pi}^{\rm LTE}\left(C^*\right)$ does not significantly change with $K$.
\end{observation}

\subsubsection{Influences of $R_{\rm LTE}$, ${\delta^{\rm LTE}}$, and ${\eta^{\rm APO}}$} We investigate the impacts of parameters $R_{\rm LTE}$, ${\delta^{\rm LTE}}$, and ${\eta^{\rm APO}}$ on $C^*$. We choose $K=4$, and the distribution of $r_k$ is the same as that in Fig. \ref{fig:simu:A}. We consider four pairs of data rate discounting factors: $\left({\delta^{\rm LTE}},{\eta^{\rm APO}}\right)=\left(0.4,0.7\right),\left(0.4,0.3\right),\left(0.4,0.1\right)$, and $\left(0.6,0.3\right)$. For each pair of $\left({\delta^{\rm LTE}},{\eta^{\rm APO}}\right)$, we change $R_{\rm LTE}$ from $10{\rm~Mbps}$ to $250{\rm~Mbps}$, and determine the corresponding $C^*$ numerically. In Fig. \ref{fig:simu:D}, we plot $C^*$ against $R_{\rm LTE}$ under the different pairs of $\left({\delta^{\rm LTE}},{\eta^{\rm APO}}\right)$. 

Under all four settings, we observe that $C^*$ does not change with $R_{\rm LTE}$ when $R_{\rm LTE}$ is below $\frac{K-1+{\eta^{\rm APO}}}{K\left(1-{\delta^{\rm LTE}}\right)}r_{\min}$. In this case, the LTE provider does not have enough capacity to satisfy the APOs' requests. Based on Theorem \ref{theorem:optimalCcase}, it chooses a small reserve rate, and works in the \emph{competition mode}. When $R_{\rm LTE}$ is above $\frac{K-1+{\eta^{\rm APO}}}{K\left(1-{\delta^{\rm LTE}}\right)}r_{\min}$, $C^*$ increases with $R_{\rm LTE}$. This is because with a larger throughput $R_{\rm LTE}$, the LTE provider is able to allocate a larger data rate to the winning APO, and hence it increases the reserve rate $C^*$ to attract the APOs. 

With ${\delta^{\rm LTE}}=0.4$, we find that $C^*$ increases with ${\eta^{\rm APO}}$ (see the top three curves). This is because under a larger ${\eta^{\rm APO}}$, the APOs are less heavily interfered by the LTE, and hence are less willing to cooperate with the LTE provider. As a result, the LTE provider needs to increase its reserve rate to attract the APOs.

With ${\eta^{\rm APO}}=0.3$, we find that $C^*$ under ${\delta^{\rm LTE}}=0.4$ is no smaller than that under ${\delta^{\rm LTE}}=0.6$. Under a smaller ${\delta^{\rm LTE}}$, the LTE provider is more heavily affected by the interference from Wi-Fi. In this case, the LTE provider chooses a larger reserve rate $C^*$ to motivate the cooperation with the APOs. Furthermore, compared with ${\eta^{\rm APO}}$, we find that the difference in ${\delta^{\rm LTE}}$ leads to a larger difference in $C^*$, which shows that ${\delta^{\rm LTE}}$ has a larger impact on $C^*$ than ${\eta^{\rm APO}}$.

We summarize the observations in Fig. \ref{fig:simu:D} as follows.

\vspace{-0.1cm}

\begin{observation}\label{observation:mono}
The optimal reserve rate $C^*$ is non-decreasing in $R_{\rm LTE}$, increasing in ${\eta^{\rm APO}}$, and non-increasing in ${\delta^{\rm LTE}}$. Moreover, ${\delta^{\rm LTE}}$ has a larger impact on $C^*$ than ${\eta^{\rm APO}}$.
\end{observation}

\vspace{-0.5cm}
\subsection{Comparison with {{The Benchmark Scheme}}}
In this section, we compare our auction-based spectrum sharing scheme with a benchmark scheme. Given a set ${\cal K}$ of APOs, the two schemes work as follows:
\begin{itemize}
\item \emph{Our auction-based scheme:} First, the LTE provider determines $C^*$ numerically based on the approach in Section \ref{subsec:stageI:3}. Second, each APO $k\in{\cal K}$ submits its bid based on the equilibrium strategy $b^*\left(r_k,C^*\right)$ in Section \ref{sec:stageII:APO}. Third, the LTE provider determines its working mode, the winning APO, and the allocated rate based on the auction rule in Section \ref{subsec:auctionframe}. 

\item \emph{Benchmark scheme:} The LTE provider randomly shares a channel with one of the $K$ APOs.{\footnote{{{The existing studies focused on the LTE/Wi-Fi coexistence \cite{cano2015coexistence,zhangmodeling}, and there are no results studying the cooperation between the two types of networks. Hence, we represent the state-of-the-art solution by the benchmark scheme, where the LTE coexists with Wi-Fi. Since the LTE provider does not know the private information $r_k$, it cannot differentiate the channels. Therefore, in the benchmark scheme, the LTE provider will randomly pick a channel to coexist with the corresponding APO.}}}} 

\end{itemize}
For a particular set of APOs, we denote the LTE provider's payoff under our auction-based and the benchmark schemes as $\pi_{\rm a}^{\rm LTE}$ and $\pi_{\rm b}^{\rm LTE}$, respectively.{\footnote{Note that ${\bar \Pi}^{\rm LTE}\left(C^*\right)$ is the expectation of $\pi_{\rm a}^{\rm LTE}$ with respect to the APO types.}} Furthermore, we denote the APOs' total payoff under our auction-based and the benchmark schemes as $\pi_{\rm a}^{\rm APO}$ and $\pi_{\rm b}^{\rm APO}$, respectively. For a given set of APOs, we compute the relative performance gains of our auction-based scheme over the benchmark scheme in terms of the LTE's payoff and the APOs' total payoff as
\begin{align}
\rho^{\rm LTE} \triangleq \frac{\pi_{\rm a}^{\rm LTE}-\pi_{\rm b}^{\rm LTE}}{\pi_{\rm b}^{\rm LTE}}{\rm ~and~} \rho^{\rm APO} \triangleq \frac{\pi_{\rm a}^{\rm APO}-\pi_{\rm b}^{\rm APO}}{\pi_{\rm b}^{\rm APO}}.
\end{align}

\subsubsection{Performance on Average $\rho^{\rm LTE}$ and $\rho^{\rm APO}$} We investigate the average $\rho^{\rm LTE}$ and $\rho^{\rm APO}$. We consider four pairs of data rate discounting factors: $\left({\delta^{\rm LTE}},{\eta^{\rm APO}}\right)=\left(0.4,0.1\right),\left(0.4,0.3\right),\left(0.4,0.7\right)$, and $\left(0.6,0.3\right)$,{\footnote{{In a practical implementation, the values of $\delta^{\rm LTE}$ and $\eta^{\rm APO}$ depend on the applied coexistence mechanism (\emph{e.g.}, LBT or CSAT) and the corresponding settings (\emph{e.g.}, LTE off time in CSAT).}}} and change $R_{\rm LTE}$ from $30{\rm~Mbps}$ to $370{\rm~Mbps}$. The other settings are the same as those in Fig. \ref{fig:simu:D}. Given a pair of $\left({\delta^{\rm LTE}},{\eta^{\rm APO}}\right)$ and a particular value of $R_{\rm LTE}$, we randomly choose $r_k$, $k\in{\cal K}$, based on the truncated normal distribution, implement our auction-based scheme and the benchmark scheme separately, and record the corresponding values of $\rho^{\rm LTE}$ and $\rho^{\rm APO}$. For each pair of $\left({\delta^{\rm LTE}},{\eta^{\rm APO}}\right)$ and each value of $R_{\rm LTE}$, we run the experiment $20,000$ times, and obtain the corresponding average values of $\rho^{\rm LTE}$ and $\rho^{\rm APO}$.
\begin{figure*}[t]
  \centering
  \begin{minipage}[t]{.32\linewidth}
  \includegraphics[scale=0.3]{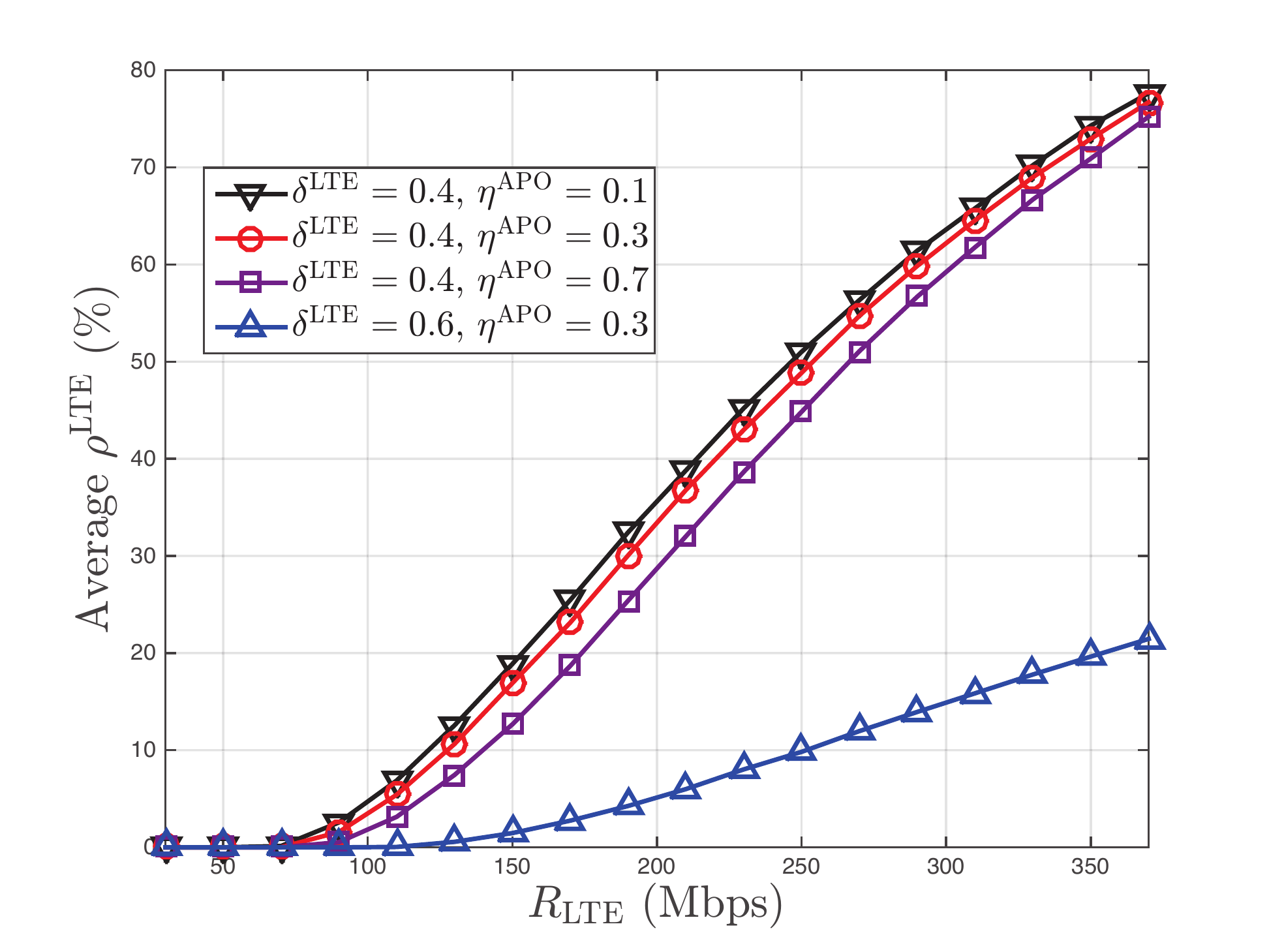}
  \caption{Comparison on LTE's Payoff.}
  \label{fig:simu:E}
  \vspace{-4mm}
  \end{minipage}
  \begin{minipage}[t]{.32\linewidth}
  \centering
  \includegraphics[scale=0.3]{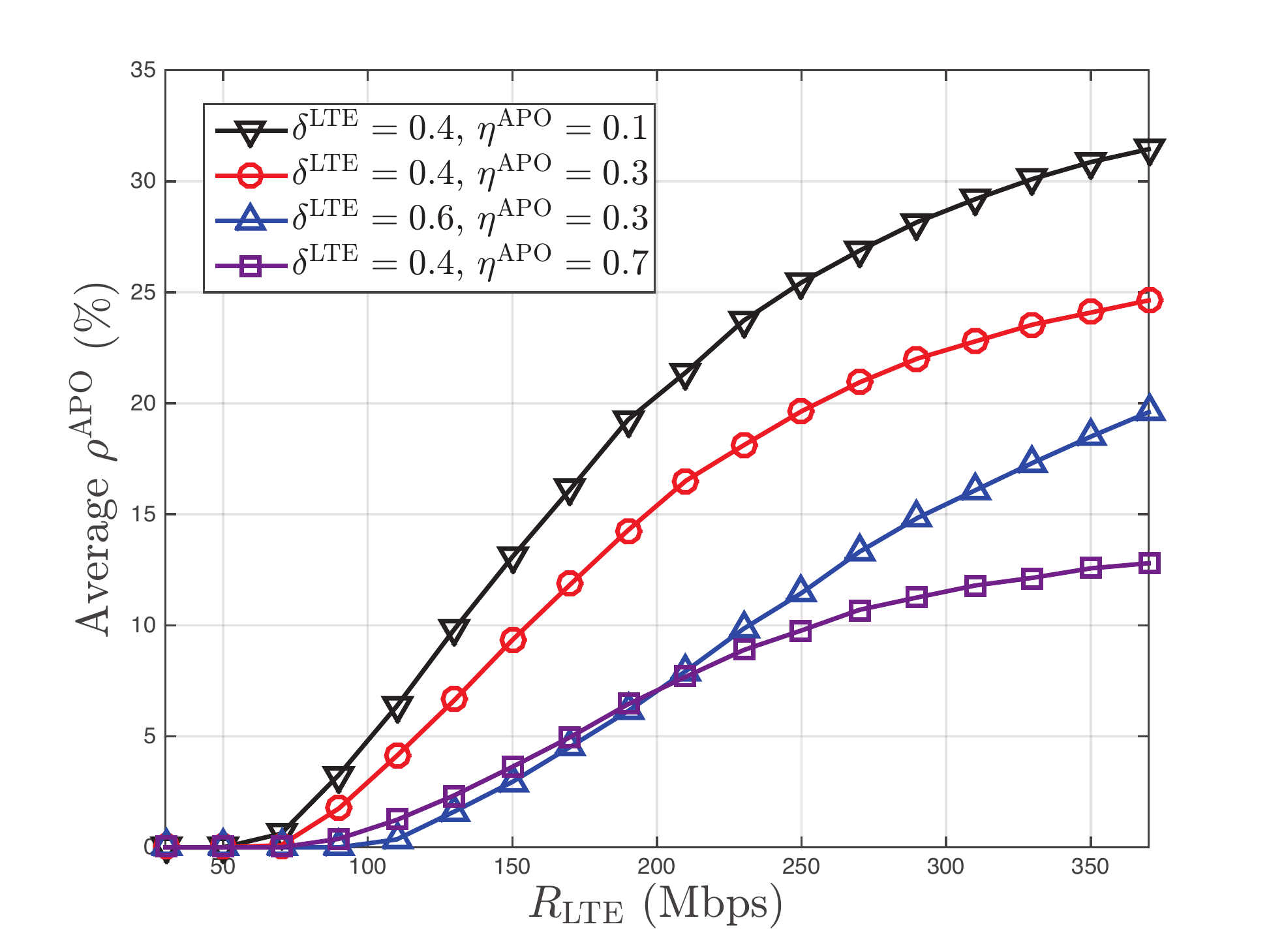}
  \caption{Comparison on APOs' Payoffs.}
  \label{fig:simu:F}
  \vspace{-4mm}
  \end{minipage}
  \begin{minipage}[t]{.32\linewidth}
  \centering
  \includegraphics[scale=0.3]{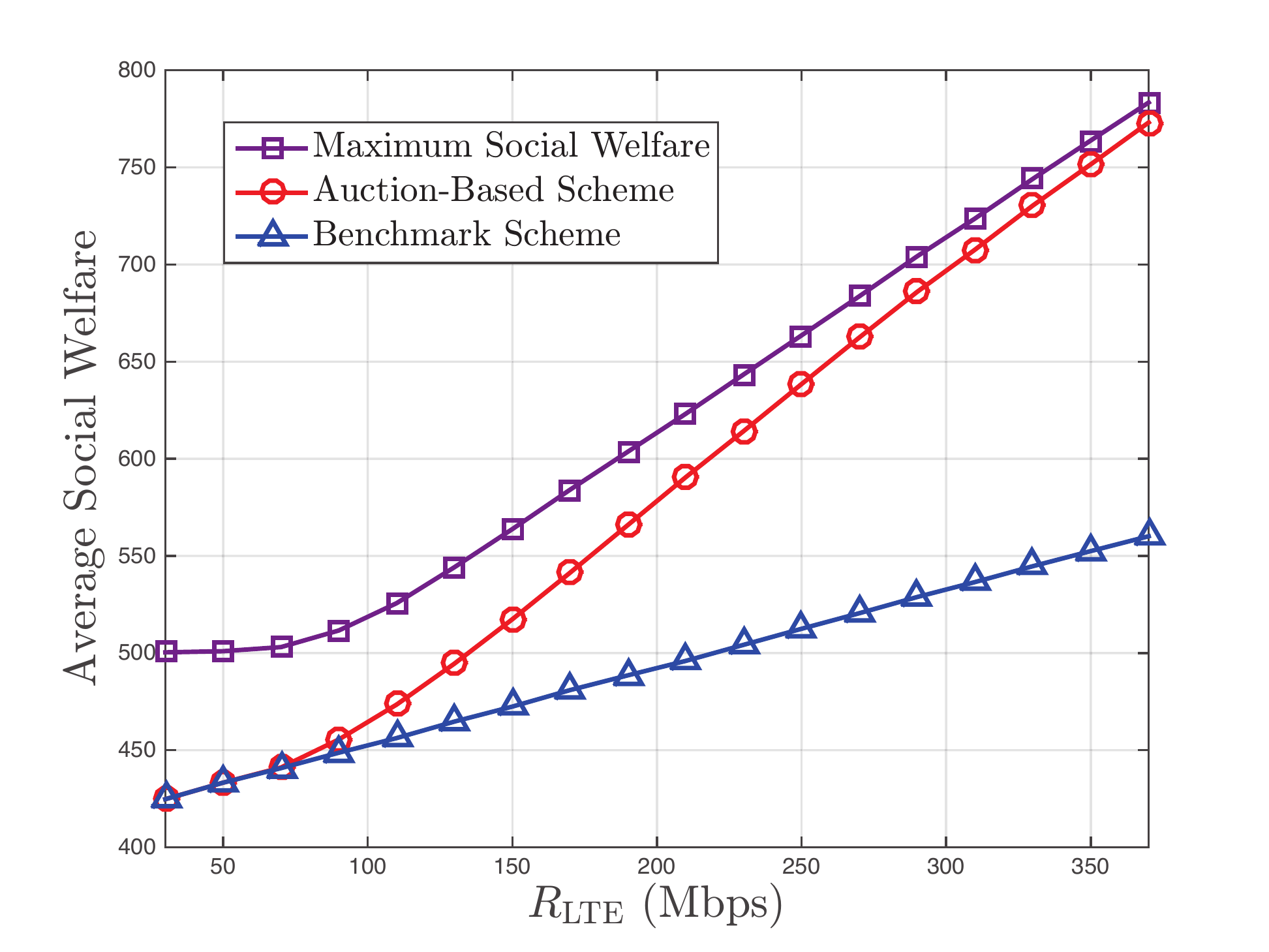}
  \caption{Comparison on Social Welfare.}
  \label{fig:simu:G}
  \vspace{-4mm}
  \end{minipage}
\end{figure*}

In Fig. \ref{fig:simu:E}, we plot the average $\rho^{\rm LTE}$ against $R_{\rm LTE}$ for different $\left({\delta^{\rm LTE}},{\eta^{\rm APO}}\right)$ pairs. First, we observe that the average $\rho^{\rm LTE}$ increases with $R_{\rm LTE}$. In particular, all the average $\rho^{\rm LTE}$ with ${\delta^{\rm LTE}}=0.4$ are above $70\%$ for $R_{\rm LTE}=370{\rm~Mbps}$ (the maximum LTE throughput according to \cite{canousing}). That is to say, our auction-based scheme's performance gain on the LTE provider's payoff is more significant for a larger $R_{\rm LTE}$. The reason is that a larger $R_{\rm LTE}$ enables the LTE provider to set a larger reserve rate, which increases the probability for the cooperation between the LTE provider and the APOs. Second, when ${\eta^{\rm APO}}=0.3$ and ${\delta^{\rm LTE}}$ increases from $0.4$ to $0.6$, the average $\rho^{\rm LTE}$ decreases significantly. Since a larger ${\delta^{\rm LTE}}$ implies that the coexistence with Wi-Fi reduces the LTE's payoff less significantly, the cooperation with Wi-Fi is less beneficial to the LTE provider, which decreases the average $\rho^{\rm LTE}$. Third, when ${\delta^{\rm LTE}}=0.4$ and ${\eta^{\rm APO}}$ changes from $0.1$ to $0.7$, the change in the average $\rho^{\rm LTE}$ is small. Hence, ${\eta^{\rm APO}}$ has a smaller impact on the average $\rho^{\rm LTE}$ comparing with ${\delta^{\rm LTE}}$. We summarize the observations in Fig. \ref{fig:simu:E} as follows.
\begin{observation}
Compared with the benchmark scheme, our auction-based scheme improves the LTE's payoff by 70\% on average under a large $R_{\rm LTE}$ and a small ${\delta^{\rm LTE}}$. Moreover, the performance gain is not sensitive to ${\eta^{\rm APO}}$.
\end{observation}

In Fig. \ref{fig:simu:F}, we plot the average $\rho^{\rm APO}$ against $R_{\rm LTE}$ for different $\left({\delta^{\rm LTE}},{\eta^{\rm APO}}\right)$ pairs. First, we observe that the average $\rho^{\rm APO}$ increases with $R_{\rm LTE}$. Similar as the explanation for $\rho^{\rm LTE}$, this is because a larger $R_{\rm LTE}$ leads to a larger reserve rate, and creates more cooperation opportunities between the LTE provider and the APOs. Second, the average $\rho^{\rm APO}$ is large when both ${\delta^{\rm LTE}}$ and ${\eta^{\rm APO}}$ are small. In this case, there is a heavy interference between the LTE and the APOs in the \emph{competition mode}, and the both of them want to avoid the interference through the cooperation. Therefore, our auction-based scheme is much more efficient, and achieves a large $\rho^{\rm APO}$. We summarize the observations in Fig. \ref{fig:simu:F} as follows.
\begin{observation}
Compared with the benchmark scheme, our auction-based scheme is most beneficial to the APOs for a large $R_{\rm LTE}$ and small ${\delta^{\rm LTE}}$ and ${\eta^{\rm APO}}$.
\end{observation}

\vspace{-0.1cm}

\subsubsection{Performance on Social Welfare} We consider $\left({\delta^{\rm LTE}},{\eta^{\rm APO}}\right)=\left(0.4,0.3\right)$, and choose the same settings as Fig. \ref{fig:simu:E} and Fig. \ref{fig:simu:F} for the other parameters. In Fig. \ref{fig:simu:G}, we plot the average social welfares of the two schemes, and also show the average value of the maximum social welfare. To compute the maximum social welfare for a particular set of APOs, we assume that there is a \emph{centralized decision maker}, who allocates $K$ channels to the LTE provider and the $K$ APOs in a manner that maximizes the social welfare.{\footnote{Specifically, the \emph{centralized decision maker} can choose to: (i) keep the LTE idle, and allocate all channels to the APOs, (ii) keep one APO idle, and allocate all channels to the LTE and the remaining $K-1$ APOs, or (iii) let the LTE share one channel with one APO, and allocate the remaining channels to the remaining $K-1$ APOs.}} For each experiment, we randomly pick a set of APOs and record the social welfare achieved by the \emph{centralized decision maker}. We run the experiment $20,000$ times, and obtain the average value of the maximum social welfare.

When $R_{\rm LTE}$ increases, the social welfare gain of our auction-based scheme over the benchmark scheme increases, and the average social welfare under our auction-based scheme approaches the maximum social welfare. 
This is because when $R_{\rm LTE}$ is large, it is always good for the LTE to exclusively occupy a channel to maximize the social welfare. 
For our auction-based scheme, the increase of $R_{\rm LTE}$ improves the cooperation chance between the LTE and the APOs, and hence increases the probability for the LTE to exclusively occupy a channel.
The result in Fig. \ref{fig:simu:G} shows that in our auction-based scheme, even the LTE provider and APOs make decisions to maximize their own payoffs, and the LTE provider and each APO do not have the complete information on the other APOs' types, the eventual auction outcome leads to a close-to-optimal social welfare for a large $R_{\rm LTE}$. We summarize the observation in Fig. \ref{fig:simu:G} as follows.

\begin{observation}
Our auction-based scheme leads to a close-to-optimal social welfare when $R_{\rm LTE}$ is large.
\end{observation}


\vspace{-0.5cm}


{{

\section{Practical Implementation and Model Extension}\label{sec:discussion}
In this section, we first discuss the practical implementation of our auction framework. In particular, we explain the approach for the LTE provider and APOs to exchange information (\emph{e.g.}, reserve rate and bids). 
Then we discuss some extensions of our model. Specifically, in Section \ref{subsec:discussion:APOshare}, we extend our model to the scenario where different APOs can share the same channel. In Section \ref{subsec:discussion:dumb}, we consider a scenario where some APOs' traffic cannot be onloaded to the LTE network. In Section \ref{subsec:discussion:multiLTE}, we extend our model to the scenario where there are multiple LTE providers. 

In a practical implementation, a centralized broker (\emph{e.g.}, a private company or a company designated by the government) can coordinate the interactions between the LTE provider and APOs \cite{iosifidis2015iterative}. 
Next we briefly introduce the centralized broker with an example from the TV white space networks, which are in the process of commercial trials in the US and UK. In the TV white space networks, a white space database operator (\emph{e.g.}, Google, Microsoft, and SpectrumBridge) serves as the broker to record and update the TV spectrum usage (by TV stations) as well as the secondary access (by non-TV devices) in the same area. 
{{Moreover, the broker controls the spectrum allocated to different secondary service providers to avoid the interference between the secondary service providers' networks.}} This shows that it is possible to coordinate the spectrum sharing of different networks through a broker, even if these networks belong to different operators {{and have overlapping coverages}}. In our auction framework, the LTE provider can announce the reserve rate to the broker at the beginning of each time slot. The APOs that are interested in participating in the auction can communicate with the broker to obtain the reserve rate information and submit their bids to the broker.{\footnote{{In particular, when the APOs have multiple equilibrium strategies $b^*$ under the reserve rate (\emph{i.e.}, $M>1$ or $L>1$), the broker can coordinate the APOs' selection of the equilibrium strategy. Intuitively, the broker will suggest the equilibrium strategy that maximizes the social welfare to the APOs. We are interested in studying the details of this problem in our future work.}}} Then the broker determines the winning APO based on our auction rule, and broadcasts this result to the LTE provider and APOs. With the broker's help, the LTE provider does not need to directly communicate with all surrounding APOs.

\subsection{Extension: Channel Sharing Among APOs}\label{subsec:discussion:APOshare}

In this section, we discuss the extension of our framework to the scenario where different APOs can share the same channel. In this scenario, the LTE provider still determines at most one winning APO in each auction. The major challenge is that when there are other APOs in the winning APO's channel, the LTE provider has to coexist with these remaining APOs (based on the coexistence mechanisms like LBT and CSAT) after onloading the winning APO's traffic. 
Therefore, we need to (i) extend the modeling of the LTE provider's payoff, the APOs' payoffs, and the APOs' types, and (ii) modify the auction rule. In the following, we briefly explain these two aspects.
    
For the modeling, we should first model the impact of the number of APOs in the same channel on the LTE provider's and the APOs' payoffs. Intuitively, the reductions in the LTE provider's and the APOs' payoffs are more severe when there are more APOs using the same channel. Second, we should model the multi-dimensional APO type. In Section \ref{sec:model}, we define the APO type as an APO's throughput without interference (\emph{i.e.}, $r_k$). Here, an APO's type should also include the information of the number of APOs in the same channel. In the equilibrium analysis, we can characterize the APOs' equilibrium strategies by a function that maps an APO type (\emph{i.e.}, throughput and number of APOs in the same channel) to a bid.
    
For the auction rule, the major modification is the rule of determining the winning APO. In Section \ref{sec:model}, the winning APO is always the APO with the lowest bid. However, when different APOs can share the same channel, such a rule is no longer optimal for the LTE provider. This is because the APO with the lowest bid may have many other APOs using the same channel, and hence the benefit for the LTE provider to cooperate with this APO may be small. Therefore, the LTE provider has to consider both the APOs' bids and the number of APOs in each channel to determine the winning APO.

\subsection{Extension: Complex APOs}\label{subsec:discussion:dumb}

In reality, some users' mobile devices, such as the laptops, do not have the LTE interfaces. The existence of these mobile devices prevents the corresponding APOs from participating in the auction and onloading all of their traffic to the LTE network. For ease of exposition, we use the \emph{simple} APOs to represent the APOs who can onload all of their traffic to the LTE network, and use the \emph{complex} APOs to represent the APOs who cannot onload all of their traffic to the LTE network. The \emph{complex} APOs will not participate in the auction and will simply use their original channels. In the following, we explain the impact of the consideration of \emph{complex} APOs on our analysis.

First, if all APOs occupy different channels (the assumption in Section \ref{sec:model}), our current analysis can be directly extended to the case where there are \emph{complex} APOs. Notice that even though the LTE provider can only cooperate with the \emph{simple} APOs in the cooperation mode, it can compete with both the \emph{simple} and \emph{complex} APOs in the competition mode. Therefore, the major change is that when the LTE provider works in the competition mode, the expected payoff of an APO depends on the number of all APOs (\emph{simple} and \emph{complex} APOs), instead of the number of APOs participating in the auction (\emph{simple} APOs). 
Second, if different APOs can share the same channel (the scenario in Section \ref{subsec:discussion:APOshare}), it will be much more challenging to consider the \emph{complex} APOs in the analysis. This is because the \emph{complex} APOs may coexist with the \emph{simple} APOs in the same channel. In this situation, we need to characterize a \emph{simple} APO's equilibrium strategy based on the number of \emph{complex} APOs as well as the number of \emph{simple} APOs in the APO's channel.


\subsection{Extension: Multiple LTE Providers}\label{subsec:discussion:multiLTE}

In this section, we discuss the extension of our framework to the scenario where there are multiple LTE providers. According to \cite{forum}, the LTE networks of different providers can well coexist with each other in the same unlicensed channel. Hence, when there are multiple LTE providers, the focus of our auction framework is still onloading the Wi-Fi APOs' traffic to the LTE networks.
 
\begin{figure}[t]
  \centering
  \includegraphics[scale=0.37]{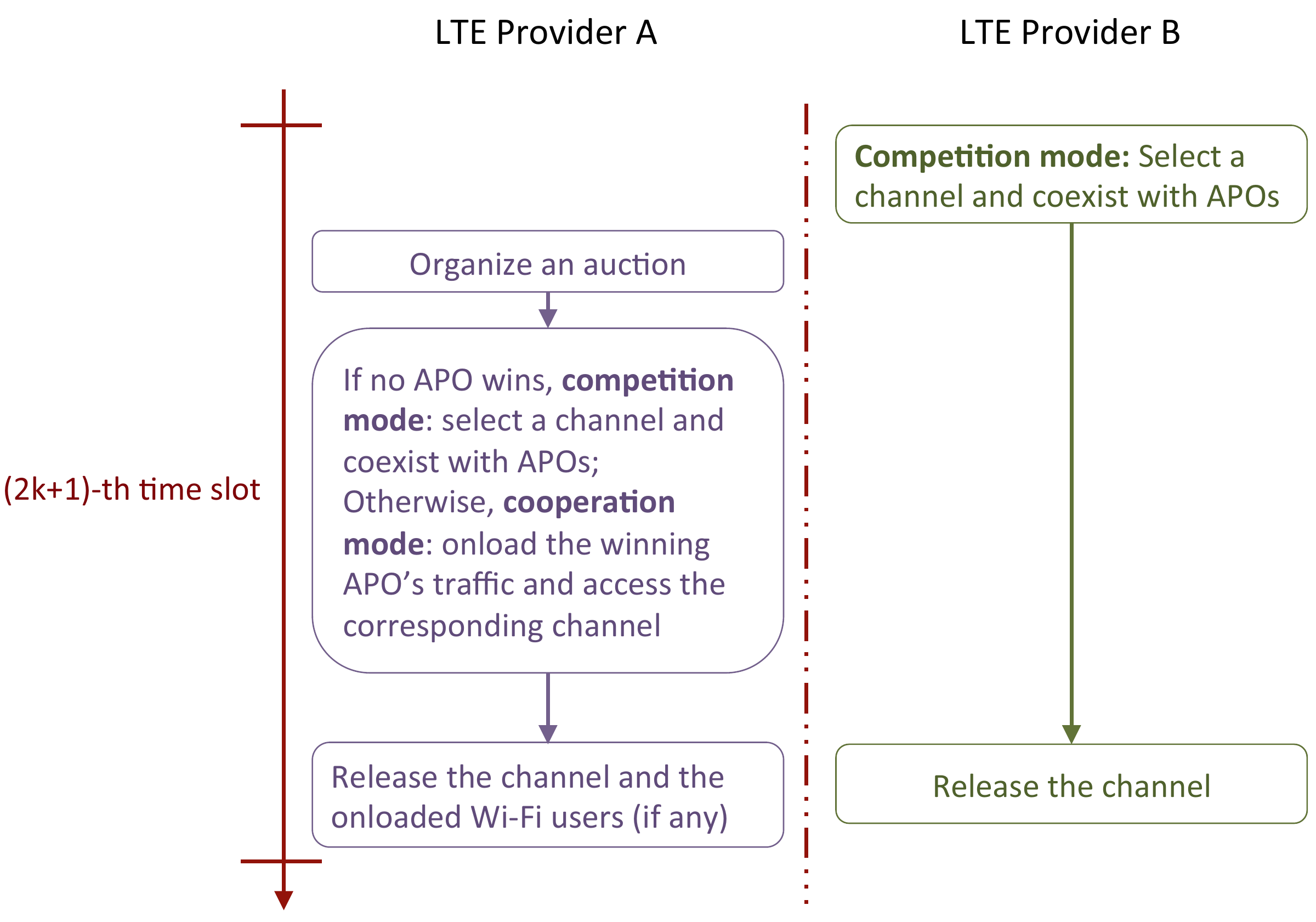}\\
  \caption{An Example of Two LTE Providers.}
  \label{fig:majorrevision}
  \vspace{-0.7cm}
\end{figure}   

When there are multiple LTE providers, they can take turns to organize the auctions, which can be managed by the centralized broker. Suppose that there are two LTE small cell networks in the same area, and they are owned by LTE provider A and LTE provider B, respectively. We illustrate a protocol in Figure \ref{fig:majorrevision}. 
During the odd number $\left(2k+1\right)$-th ($k\in\left\{0,1,\ldots\right\}$) time slot, LTE provider B directly operates in the competition mode, and LTE provider A can send a request to the broker and organize an auction. During the even number $\left(2k+2\right)$-th time slot, the two LTE providers switch their roles: LTE provider A operates in the competition mode, and LTE provider B can organize an auction. We can apply similar protocols to the situations with more than two LTE providers. 

Since the protocol designs for the $\left(2k+1\right)$-th time slot and the $\left(2k+2\right)$-th time slot are symmetric, next we only introduce the protocol design for the $\left(2k+1\right)$-th time slot. At the beginning of the $\left(2k+1\right)$-th time slot, LTE provider B chooses the competition mode, \emph{i.e.}, it selects a channel and coexists with the corresponding APOs. Then LTE provider A organizes an auction: when no APO wants to cooperate with LTE provider A, LTE provider A works in the competition mode, selects a channel, and coexists with the corresponding networks; otherwise, LTE provider A works in the cooperation mode, onloads the winning APO's traffic, and accesses the corresponding channel. At the end of the $\left(2k+1\right)$-th time slot, both LTE provider A and LTE provider B release the channels they use. In particular, LTE provider A also needs to release the onloaded Wi-Fi users if LTE provider A works in the cooperation mode during the $\left(2k+1\right)$-th time slot.

Next we discuss the challenges of analyzing the scenario with multiple LTE providers under the protocol we introduced above. Briefly speaking, when a particular LTE provider organizes an auction, it needs to consider the number of other LTE providers in each channel. This is because the benefit for the LTE provider to cooperate with an APO decreases with the number of other LTE providers using the same channel. In the analysis, we should characterize an APO's equilibrium strategy based on its throughput $r_k$ and the number of LTE providers in the same channel. Furthermore, the auctioneer should consider both the APOs' bids and the number of LTE providers in each channel to determine the winning APO. 
We provide a complete analysis of the scenario where there are multiple LTE providers in Section \ref{sec:supplementary} (supplementary materials).

\section{Conclusion}\label{sec:conclusion}
In this paper, we proposed a framework for LTE's coopetition with Wi-Fi in the unlicensed spectrum. We designed a reverse auction for the LTE provider to exclusively obtain the channel from the APOs by onloading their traffic. 
{{Compared with the existing LTE/Wi-Fi coexistence mechanisms like LBT and CSAT, our auction can potentially avoid the interference between the LTE and APOs.}}
The analysis of the auction is quite challenging as the designed auction involves positive allocative externalities. 
We characterized {{the unique form of the APOs' bidding strategies at the equilibrium}}, and analyzed the optimal reserve rate of the LTE provider. 
Numerical results showed that our framework benefits both the LTE provider and the APOs, and it achieves a close-to-optimal social welfare under a large LTE throughput. 
In our framework, the LTE provider announces the reserve rate and each APO then submits a bid at the beginning of each time slot, where the length of each time slot corresponds to several minutes. 
Compared with the existing LTE/Wi-Fi coexistence mechanisms, our auction framework leads to more signaling overhead. However, if the LTE provider and an APO agree to cooperate, there is no more need for the LTE to frequently sense the channel activities, which removes the related operational overhead during the rest of the time slot.{\footnote{{{For example, in the LBT mechanism, the LTE senses the channel status (busy or idle) every $20$ microseconds; in the CSAT mechanism, the LTE senses the Wi-Fi activity on a time scale of $100$ milliseconds to determine the length of LTE off time {\cite{Qualcomm}}.}}}} Therefore, although our framework generates more signaling overhead initially, it can potentially significantly save the sensing cost (\emph{e.g.}, power) and improve the payoffs of both LTE and Wi-Fi.

{{An interesting observation of our framework}} is that sometimes even if the cooperation mutually benefits the LTE provider and the APOs, these two types of networks do not reach an agreement on the cooperation. The reason is that our framework considers an incomplete information setting. For example, the LTE provider determines the reserve rate to maximize its expected payoff by considering the distribution of ${\bm r}$ (the vector of APOs' types) instead of the actual value of ${\bm r}$. For some ${\bm r}$, such a reserve rate may not be optimal to the LTE provider and can make the LTE provider lose some cooperation chances that mutually benefit both types of networks (we provide an example in Appendix \ref{appendix:sec:example}). Similarly, the incomplete information among the APOs can also lead to the same inefficiency problem. In our future work, we will consider other mechanisms (\emph{e.g.}, bargaining) for the LTE/Wi-Fi coopetition to reduce such an inefficiency.


 




\bibliographystyle{IEEEtran}
\bibliography{bare_conf}


\begin{IEEEbiography}
[{\includegraphics[width=1in,height=1.25in,clip,keepaspectratio]{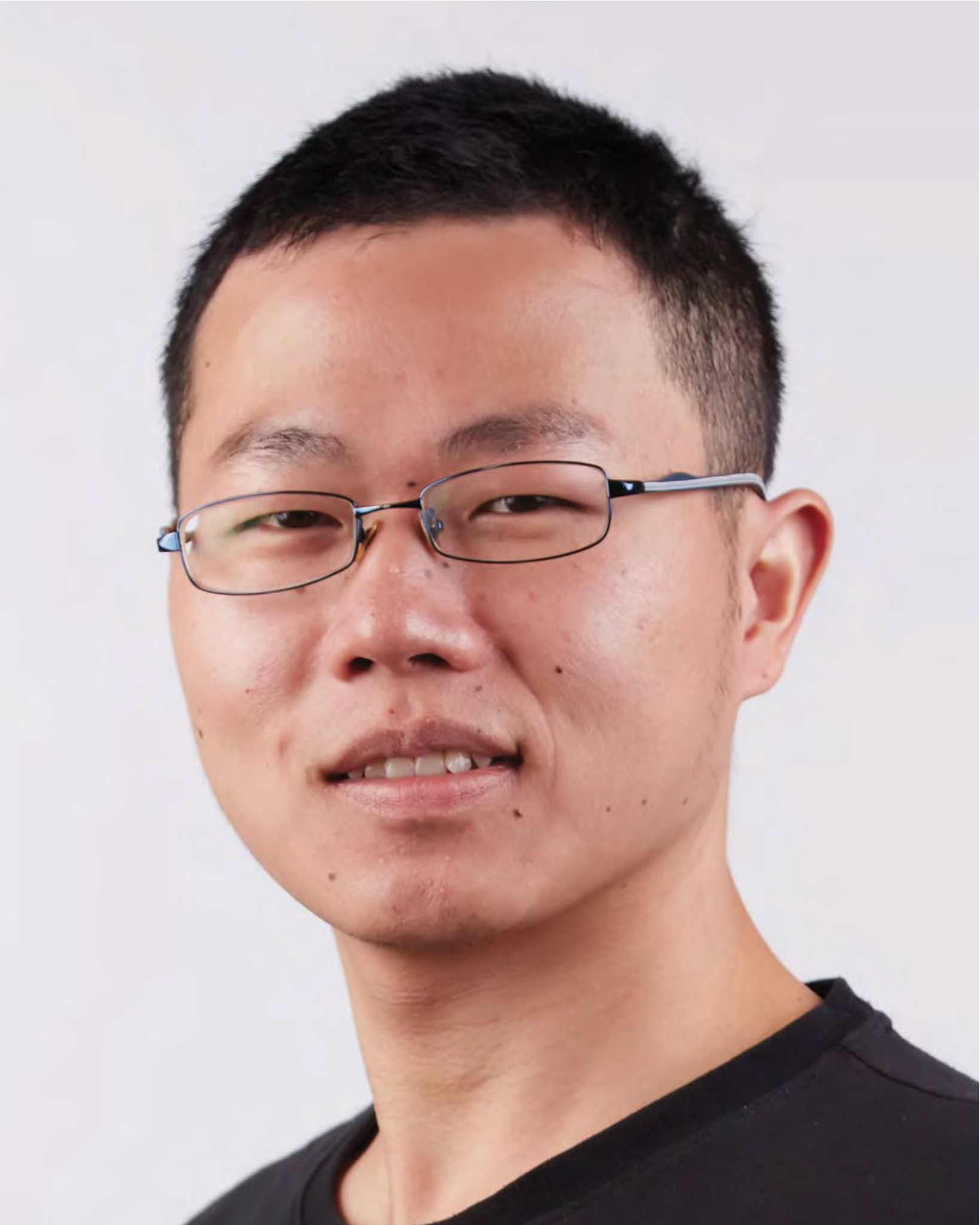}}]
{Haoran Yu} (S'14) received his Ph.D. degree at the Chinese University of Hong Kong in 2016. He was a visiting student in the Yale Institute for Network Science and the Department of Electrical Engineering at Yale University during 2015-2016. He is now a post-doctoral researcher in the Department of Information Engineering at the Chinese University of Hong Kong. His research interests lie in the field of wireless communications and network economics, with current emphasis on cellular/Wi-Fi integration, LTE in unlicensed spectrum, and economics of Wi-Fi networks. He was awarded the Global Scholarship Programme for Research Excellence by the Chinese University of Hong Kong. His paper in IEEE INFOCOM 2016 was selected as a Best Paper Award finalist and one of top 5 papers from 1600+ submissions.
\end{IEEEbiography}

\vspace{-1.55cm}

\begin{IEEEbiography}
[{\includegraphics[width=1in,height=1.25in,clip,keepaspectratio]{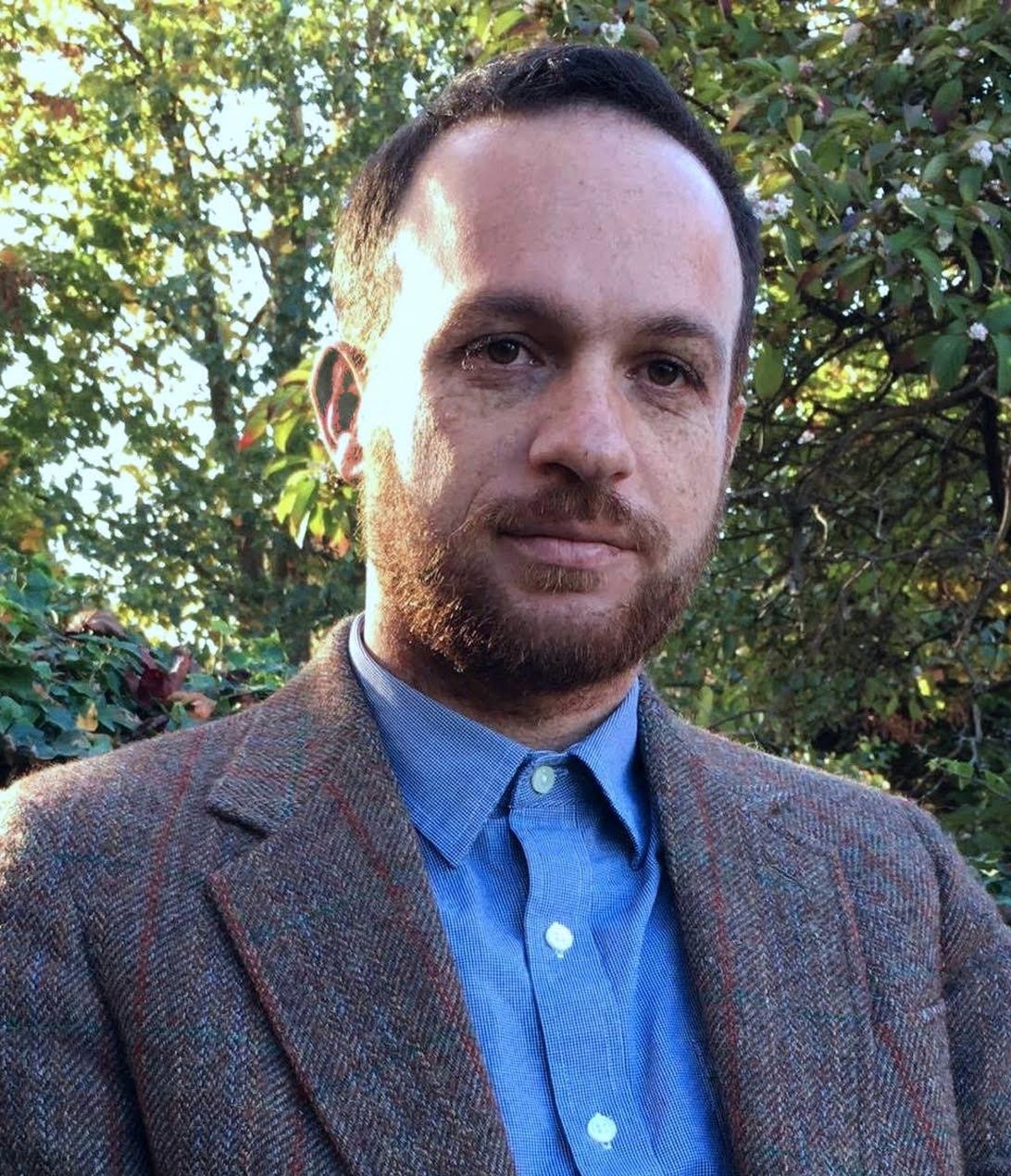}}]
{George Iosifidis} received the Diploma degree in electronics and telecommunications engineering from the Greek Air Force Academy in 2000, and the M.S. and Ph.D. degrees in electrical engineering from University of Thessaly, Greece, in 2007 and 2012, respectively. He worked as a post-doctoral researcher at CERTH, Greece, and Yale University, USA. He is currently the Ussher Assistant Professor in Future Networks with Trinity College Dublin, and also a Funded Investigator with the national research centre CONNECT in Ireland. His research interests lie in the broad area of wireless network optimization and network economics.
\end{IEEEbiography}

\vspace{-1.55cm}

\begin{IEEEbiography}[{\includegraphics[width=1in,height=1.25in,clip,keepaspectratio]{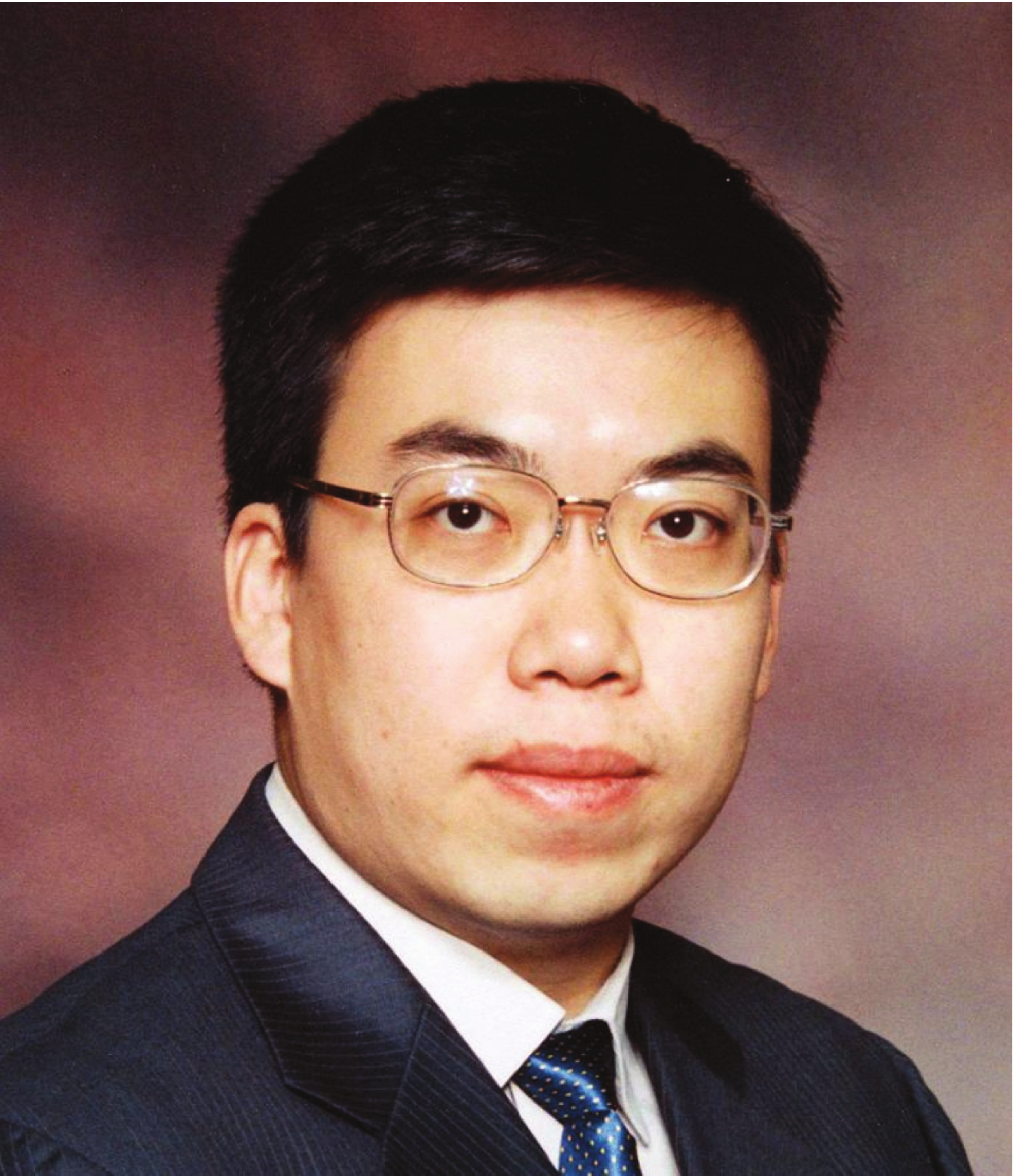}}]
{Jianwei Huang} (S'01-M'06-SM'11-F'16) is an Associate Professor and Director of the Network Communications and Economics Lab (ncel.ie.cuhk.edu.hk), in the Department of Information Engineering at the Chinese University of Hong Kong. He received the Ph.D. degree from Northwestern University in 2005 and worked as a Postdoc Research Associate in Princeton during 2005-2007. He is the co-recipient of 8 international Best Paper Awards, including IEEE Marconi Prize Paper Award in Wireless Communications in 2011. He has co-authored six books, including the first textbook on ``Wireless Network Pricing." He has served as an Associate Editor of IEEE Transactions on Cognitive Communications and Networking, IEEE Transactions on Wireless Communications, and IEEE Journal on Selected Areas in Communications - Cognitive Radio Series. He is the Vice Chair of IEEE ComSoc Cognitive Network Technical Committee and the Past Chair of IEEE ComSoc Multimedia Communications Technical Committee. He is a Fellow of IEEE, a Distinguished Lecturer of IEEE Communications Society, and a Thomson Reuters Highly Cited Researcher in Computer Science.
\end{IEEEbiography}

\vspace{-1.55cm}

\begin{IEEEbiography}
[{\includegraphics[width=1in,height=1.25in,clip,keepaspectratio]{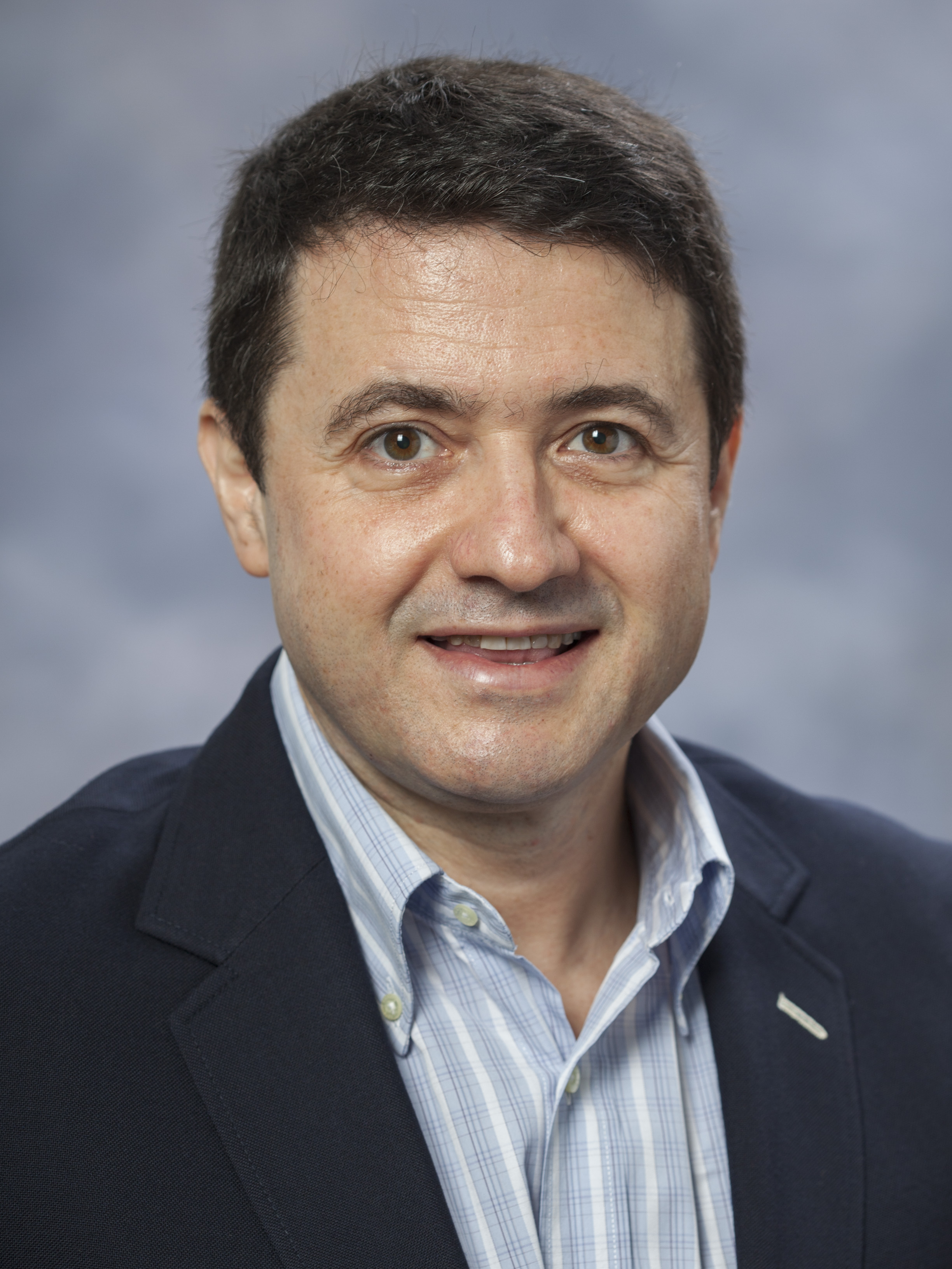}}]
{Leandros Tassiulas} (S'89-M'91-SM'05-F'07) is the John C. Malone Professor of Electrical Engineering and member of the Institute for Network Science at Yale University. His research interests are in the field of computer and communication networks with emphasis on fundamental mathematical models and algorithms of complex networks, architectures and protocols of wireless systems, sensor networks, novel internet architectures and experimental platforms for network research. His most notable contributions include the max-weight scheduling algorithm and the back-pressure network control policy, opportunistic scheduling in wireless, the maximum lifetime approach for wireless network energy management, and the consideration of joint access control and antenna transmission management in multiple antenna wireless systems. Dr. Tassiulas is a Fellow of IEEE (2007). His research has been recognized by several awards including the IEEE Koji Kobayashi computer and communications award (2016), the inaugural INFOCOM 2007 Achievement Award ``for fundamental contributions to resource allocation in communication networks'', the INFOCOM 1994 best paper award, a National Science Foundation (NSF) Research Initiation Award (1992), an NSF CAREER Award (1995), an Office of Naval Research Young Investigator Award (1997) and a Bodossaki Foundation award (1999). He holds a Ph.D. in Electrical Engineering from the University of Maryland, College Park (1991). He has held faculty positions at Polytechnic University, New York, University of Maryland, College Park, and University of Thessaly, Greece.
\end{IEEEbiography}

\newpage

\section{Appendix}

For ease of exposition, we define $b_{\min}^{-k}\triangleq \min_{j\in {\cal K},j\ne k}b_j$ as the minimum bid from the APOs excluding APO $k$.

{{
\subsection{Proof of Lemma \ref{lemma:rT}}
We define function $H\left(r\right)$ as the left hand side of equation (\ref{equ:rT}):
\begin{align}
\nonumber
& H\left(r\right)\triangleq  \left(1-F\left(r\right)\right)^{K-1}\left(C-\frac{K-1+{\eta^{\rm APO}}}{K}r\right)\\
& +\sum_{n=1}^{K-1}{\binom{K-1}{n} \left(F\left(r\right)-F\left(C\right)\right)^n\left(1-F\left(r\right)\right)^{K-1-n}\frac{C-r}{n+1}}.
\end{align}

{\bf Step 1:} we show that function $H\left(r\right)$ is continuous for $r\in\left(-\infty,\infty\right)$. 
Since $F\left(\cdot\right)$ is the cumulative distribution function of a continuous random variable, function $F\left(r\right)$ is continuous for $r\in\left(-\infty,\infty\right)$. In particular, $F\left(r\right)=0$ for $r\in\left(-\infty,r_{\min}\right]$ and $F\left(r\right)=1$ for $r\in\left[r_{\max},\infty\right)$. 
Hence, function $H\left(r\right)$ is continuous for $r\in\left(-\infty,\infty\right)$.

{\bf Step 2:} we compute the values of function $H\left(r\right)$ at points $r=C$ and $r=r_{\max}$. First, we compute $H\left(C\right)$ as follows:
\begin{align}
\nonumber
& H\left(C\right)=\left(1-F\left(C\right)\right)^{K-1}\left(C-\frac{K-1+{\eta^{\rm APO}}}{K}C\right)\\
\nonumber
& +\sum_{n=1}^{K-1}\!{\binom{K-1}{n}\! \left(F\left(C\right)-F\left(C\right)\right)^n\!\left(1-F\left(C\right)\right)^{K-1-n}\!\frac{C-C}{n+1}} \\
& =\left(1-F\left(C\right)\right)^{K-1}\frac{1-{\eta^{\rm APO}}}{K}C.\label{appendix:HC:a}
\end{align}
Since $F\left(\cdot\right)$ is the cumulative distribution function of a random variable, we have $F\left(C\right)\le1$. 
Next we prove $F\left(C\right)<1$ by contradiction. 
Suppose $F\left(C\right)=1$. Because $C\in\left[r_{\min},r_{\max}\right)$ (the assumption of Section III-B), we can always find a $\xi>0$ such that $C+\xi\in\left[r_{\min},r_{\max}\right]$. 
The cumulative distribution function is a non-decreasing function. 
Therefore, we have $F\left(C+\xi\right)\ge F\left(C\right)=1$. 
Together with $F\left(C+\xi\right)\le 1$ (property of cumulative distribution function $F\left(\cdot\right)$), we obtain $F\left(C+\xi \right)= 1$ and $F\left(C+\xi \right)-F\left(C\right)= 0$. 
Recall that $F\left(\cdot\right)$ and $f\left(\cdot\right)$ are the cumulative distribution function and probability density function of a random variable, respectively. We can write $F\left(C+\xi \right)-F\left(C\right)$ as
\begin{align}
F\left(C+\xi \right)-F\left(C\right)=\int_{C}^{C+\xi} f\left(r\right) dr.
\end{align}
Therefore, the result $F\left(C+\xi \right)-F\left(C\right)= 0$ contradicts with the fact that $f\left(r\right)>0$ for all $r\in\left[r_{\min},r_{\max}\right]$ (Section II). Hence, $F\left(C\right)=1$ does not hold, and we obtain $F\left(C\right)<1$.

Based on (\ref{appendix:HC:a}), $F\left(C\right)<1$, and $\eta^{\rm APO}\in\left(0,1\right)$, we conclude that $H\left(C\right)>0$. 

Then we compute $H\left(r_{\max}\right)$ as follows:
\begin{align}
\nonumber
& H\left(r_{\max}\right) \!=\! \left(1-F\left(r_{\max}\right)\right)^{K-1}\!\left(C\!-\frac{K-1+{\eta^{\rm APO}}}{K}r_{\max}\!\right)\!+\\
& \!\!\sum_{n=1}^{K-1}\!\!{\binom{K-1}{n} \!\left(\!F\left(r_{\max}\right)\!-\!F\left(C\right)\right)^n\!\left(\!1\!-\!F\left(r_{\max}\right)\right)^{K\!-1\!-n}\!\frac{C\!-\!r_{\max}}{n+1}}.
\end{align}
Because $F\left(r_{\max}\right)=1$, we can further simplify $H\left(r_{\max}\right)$ as
\begin{align}
\nonumber
H\left(r_{\max}\right) = \left(1-F\left(C\right)\right)^{K-1}\frac{C-r_{\max}}{K}. 
\end{align}
Since $F\left(C\right)<1$ and $C\in\left[r_{\min},r_{\max}\right)$, we conclude that $H\left(r_{\max}\right)<0$.

{\bf Step 3:} We have shown that function $H\left(r\right)$ is continuous for $r\in\left(-\infty,\infty\right)$ and $H\left(C\right)>0>H\left(r_{\max}\right)$. Based on the intermediate value theorem, there is at least one $r\in\left(C,r_{\max}\right)$ satisfying $H\left(r\right)=0$, which completes the proof. 
}}

\subsection{Proof of Theorem \ref{theorem:equilibrium}}
We consider the strategy for APO $k\in{\cal K}$. Assuming that all the other APOs adopt strategy $b^*\left(r_k,C\right)$ defined in (\ref{equ:equilibrium}), we show that choosing strategy $b^*\left(r_k,C\right)$ in (\ref{equ:equilibrium}) maximizes APO $k$'s payoff. We introduce the proofs in the following four parts.

{\bf Part I:} We assume that APO $k$'s type $r_k$ satisfies $r_k\in\left[r_{\min},C\right)$. We discuss the following two cases.

\emph{Case A:} $b_{\min}^{-k}\in\left[0,C\right]$. In this case, the LTE can always find an APO to cooperate with, and the analysis for APO $k$ is equivalent to that of a conventional second-price reverse auction (without allocative externalities). 
Hence, bidding type $r_k$ is APO $k$'s optimal strategy. 

In particular, we study the optimal bidding strategy for APO $k$ with $r_k=r_{\min}$. 
When all the other APOs adopt strategy $b^*\left(r_k,C\right)$ in (\ref{equ:equilibrium}), we can show that $b_{\min}^{-k}>r_{\min}$ with probability one. Therefore, if $r_k=r_{\min}$, bidding any value from set $\left[0,r_{\min}\right)$ and bidding $r_{\min}$ generate the same payoff to APO $k$. In other words, if $r_k=r_{\min}$, bidding any value from $\left[0,r_{\min}\right]$ is APO $k$'s optimal strategy.

\emph{Case B:} $b_{\min}^{-k}=``{\rm N}\textquotedblright$. If APO $k$ bids any value from set $\left[0,C\right]$ (for example, $r_k$), it obtains the same payoff, which equals the reserve rate $C$. If APO $k$ bids $``{\rm N} \textquotedblright$, its payoff will be $\frac{K-1+{\eta^{\rm APO}}}{K}r_k$. 
Since $C>r_k>\frac{K-1+{\eta^{\rm APO}}}{K}r_k$, bidding $r_k$ is one of APO $k$'s optimal strategy. 
In particular, if $r_k=r_{\min}$, bidding any value from set $\left[0,r_{\min}\right)$ is also optimal for APO $k$ (besides bidding $r_{\min}$).

Combining \emph{Case A} and \emph{Case B}, we conclude that when other APOs choose $b^*\left(r_k,C\right)$ defined in (\ref{equ:equilibrium}), it is optimal for APO $k$ with $r_k\in\left[r_{\min},C\right)$ to also choose $b^*\left(r_k,C\right)$ in (\ref{equ:equilibrium}).

{\bf Part II:} We assume that APO $k$'s type $r_k$ satisfies $r_k\in\left[C,r_T\left(C\right)\right)$. We compare bid $C$ with bid $``{\rm N} \textquotedblright$ and any bid smaller than $C$, respectively, and show that bid $C$ is optimal for APO $k$.

First, we compare bid $C$ with bid $``{\rm N} \textquotedblright$. We denote APO $k$'s expected payoff under bid $C$ as $\Pi_a$. Since all the other $K-1$ APOs choose $b^*\left(r_k,C\right)$ in (\ref{equ:equilibrium}), we compute $\Pi_a$ as (\ref{equ:pia}). 
We denote APO $k$'s expected payoff under bid $``{\rm N} \textquotedblright$ as $\Pi_b$, and compute it as (\ref{equ:pib}). 
We compute the difference between $\Pi_a$ and $\Pi_b$ as (\ref{equ:diffpiapib}). 
\begin{figure*}
\begin{align}
\nonumber
\Pi_a = & \left(1-\left(1-F\left(C\right)\right)^{K-1}\right)r_k+\left(1-F\left(r_T\left(C\right)\right)\right)^{K-1}C\\
&+\sum_{n=1}^{K-1}{\binom{K-1}{n} \left(F\left(r_T\left(C\right)\right)-F\left(C\right)\right)^n\left(1-F\left(r_T\left(C\right)\right)\right)^{K-1-n}\frac{C+nr_k}{n+1}}.\label{equ:pia}
\end{align}
\hrulefill
\begin{align}
\nonumber
\Pi_b =& \left(1-\left(1-F\left(C\right)\right)^{K-1}\right)r_k+\left(1-F\left(r_T\left(C\right)\right)\right)^{K-1}\frac{K-1+{\eta^{\rm APO}}}{K}r_k\\
& +\sum_{n=1}^{K-1}{\binom{K-1}{n} \left(F\left(r_T\left(C\right)\right)-F\left(C\right)\right)^n\left(1-F\left(r_T\left(C\right)\right)\right)^{K-1-n}r_k}.\label{equ:pib}
\end{align}
\hrulefill
\begin{align}
\nonumber
\Pi_a-\Pi_b=& \left(1-F\left(r_T\left(C\right)\right)\right)^{K-1}\left(C-\frac{K-1+{\eta^{\rm APO}}}{K}r_k\right)\\
&+\sum_{n=1}^{K-1}{\binom{K-1}{n} \left(F\left(r_T\left(C\right)\right)-F\left(C\right)\right)^n\left(1-F\left(r_T\left(C\right)\right)\right)^{K-1-n}\frac{C-r_k}{n+1}}.\label{equ:diffpiapib}
\end{align}
\hrulefill
\end{figure*}

It is easy to find that $\Pi_a-\Pi_b$ is strictly decreasing with $r_k$. Furthermore, based on the definition of $r_T\left(C\right)$, when $r_k=r_T\left(C\right)$, we have $\Pi_a-\Pi_b=0$. 
Since we assume that APO $k$'s type $r_k$ satisfies $r_k\in\left[C,r_T\left(C\right)\right)$, we have $\Pi_a>\Pi_b$. In other words, bidding $C$ generates a higher payoff to APO $k$ than bidding $``{\rm N} \textquotedblright$.

Second, we compare bid $C$ with any bid $x\in\left[0,C\right)$. It is easy to find that bid $C$ and bid $x$ generate a difference on APO $k$'s payoff only when $b_{\min}^{-k} \in\left[x,C\right]$. Since the other APOs bid according to $b^*\left(r_k,C\right)$ in (\ref{equ:equilibrium}), the probability for $b_{\min}^{-k}=x\in\left[0,C\right)$ is zero. Therefore, we do not need to consider the case where $b_{\min}^{-k}=x$. Next we discuss case $b_{\min}^{-k}\in\left(x,C\right)$ and case $b_{\min}^{-k}=C$, separately:
\begin{itemize}
\item When $b_{\min}^{-k}\in\left(x,C\right)$, bidding $C$ generates a payoff of $r_k$ to APO $k$, and bidding $x$ generates a payoff of $b_{\min}^{-k}$ to APO $k$. Since $b_{\min}^{-k}<C\le r_k$, in this case, bidding $C$ generates a higher payoff to APO $k$ than bidding $x$;
\item When $b_{\min}^{-k}=C$, APO $k$'s expected payoff under bid $C$ lies in interval $\left(C,r_k\right)$, and its payoff under bid $x$ equals $C$. Hence, in this case, bidding $C$ generates a higher greater than that under bid $x$.
\end{itemize}
To conclude, considering all cases, bidding $C$ generates a payoff higher than that under any bid smaller than $C$.

Therefore, when other APOs choose $b^*\left(r_k,C\right)$ in (\ref{equ:equilibrium}), it is optimal for APO $k$ with $r_k\in\left[C,r_T\left(C\right)\right)$ to also choose $b^*\left(r_k,C\right)$ in (\ref{equ:equilibrium}).

{\bf Part III:} We assume that APO $k$'s type $r_k=r_T\left(C\right)$. Based on a similar analysis as {\bf Part II}, we can show that bidding $C$ and bidding $``{\rm N} \textquotedblright$ generate the same expected payoff to APO $k$. Furthermore, these two bids weakly dominate any bid in $\left[0,C\right)$. 

Therefore, when other APOs choose $b^*\left(r_k,C\right)$ in (\ref{equ:equilibrium}), it is optimal for APO $k$ with $r_k=r_T\left(C\right)$ to also choose $b^*\left(r_k,C\right)$ in (\ref{equ:equilibrium}).

{\bf Part IV:} We assume that APO $k$'s type $r_k\in\left(r_T\left(C\right),r_{\max}\right]$. Based on a similar analysis as {\bf{Part II}}, we can show that bid $C$ weakly dominates any bid $x<C$. 

Next, we only need to compare bid $C$ and bid $``{\rm N} \textquotedblright$ for APO $k$. Similar as {\bf{Part II}}, we can compute $\Pi_a$ and $\Pi_b$ for APO $k$'s expected payoffs under bid $C$ and bid $``{\rm N} \textquotedblright$, respectively. From (\ref{equ:diffpiapib}), we can show that $\Pi_a < \Pi_b$ for $r_k\in\left(r_T\left(C\right),r_{\max}\right]$. Therefore, bidding $``{\rm N} \textquotedblright$ generates a higher payoff to APO $k$ than bidding $C$.

To conclude, when other APOs choose $b^*\left(r_k,C\right)$ in (\ref{equ:equilibrium}), it is optimal for APO $k$ with $r_k\in\left(r_T\left(C\right),r_{\max}\right]$ to also choose $b^*\left(r_k,C\right)$ in (\ref{equ:equilibrium}).

Summarizing {\bf Part I}, {\bf Part II}, {\bf Part III}, and {\bf Part IV}, we complete the proof.

\subsection{Preliminary Lemmas}\label{appendix:sec:preliminary}

In order to prove Theorem \ref{theorem:combine:unique}, we introduce some preliminary lemmas in this section. For all these lemma, we assume that $C$ is fixed and chosen from $\in\left[r_{\min},r_{\max}\right)$.

\begin{lemma} (Monotonicity)\label{lemma:monotonicity}
For a strategy function $\tilde b$ which constitutes an SBNE, if $r_L,r_H\in\left[r_{\min},r_{\max}\right]$, $r_L<r_H$ and ${\tilde b}\left(r_H,C\right)\ne``{\rm N} \textquotedblright$, we have ${\tilde b}\left(r_L,C\right)\ne``{\rm N} \textquotedblright$ and ${\tilde b}\left(r_L,C\right)\le {\tilde b}\left(r_H,C\right)$.
\end{lemma}

\begin{proof} 
We introduce the proof in the following two parts.

{\bf Part I:} Now suppose there is an SBNE $\tilde b$, we first show that if $r_L<r_H$ and ${\tilde b}\left(r_H,C\right)\ne``{\rm N} \textquotedblright$, we have ${\tilde b}\left(r_L,C\right)\ne``{\rm N} \textquotedblright$. We prove it by contradiction, \emph{i.e.}, suppose there exist $r_H,r_L$ such that ${\tilde b}\left(r_H,C\right)\ne``{\rm N} \textquotedblright$ and ${\tilde b}\left(r_L,C\right)=``{\rm N} \textquotedblright$.

We denote $s\triangleq {\tilde b}\left(r_H,C\right)\in \left[0,C\right]$. Next we assume APO $k$'s type is $r_k=r_H$ and study its strategy. Since $\tilde b$ constitutes an SBNE and ${\tilde b}\left(r_H,C\right)=s$, if the other APOs choose strategy $\tilde b$, bidding $s$ is optimal to APO $k$.

Next we compute APO $k$'s expected payoff under bid $s$. We define $P_{\rm win}\left(s\right)$ as the probability that APO $k$ wins the auction with bid $s$, assuming the other $K-1$ APOs adopt strategy $\tilde b$. We also define $\Pi_{\rm win}\left(s\right)$ as APO $k$'s expected payoff when it wins the auction with bid $s$, assuming the other $K-1$ APOs adopt strategy $\tilde b$. 
Based on $P_{\rm win}\left(s\right)$ and $\Pi_{\rm win}\left(s\right)$, we compute APO $k$'s expected payoff under bid $s$ as ${P_{\rm win}\left(s\right)} \Pi_{\rm win}\left(s\right) + \left(1-{P_{\rm win}\left(s\right)}\right)r_H$.

Then we compute APO $k$'s expected payoff under bid $``{\rm N} \textquotedblright$. We define $D$ as the probability that $b_{\min}^{-k}\in\left[0,C\right]$, assuming the other $K-1$ APOs adopt strategy $\tilde b$. 
APO $k$'s expected payoff under bid $``{\rm N} \textquotedblright$ is computed as $D r_H+\left(1-D\right) \frac{K-1+{\eta^{\rm APO}}}{K}{r_H}$.

Since ${\tilde b}\left(r_H,C\right)=s$, we have the following relation:
\begin{align}
\nonumber
&{P_{\rm win}\left(s\right)} \Pi_{\rm win}\left(s\right) + \left(1-{P_{\rm win}\left(s\right)}\right)r_H\ge \\
& D r_H+\left(1-D\right) \frac{K-1+{\eta^{\rm APO}}}{K}{r_H}.\label{equ:rH}
\end{align}

By our assumption, we have ${\tilde b}\left(r_L,C\right)=``{\rm N} \textquotedblright$. Next we assume APO $k$'s type is $r_k=r_L$ and study its strategy. When APO $k$ bids $``{\rm N} \textquotedblright$, its expected payoff is $D r_L+\left(1-D\right) \frac{K-1+{\eta^{\rm APO}}}{K}{r_L}$; when APO $k$ bids $s$ (recall that $s= {\tilde b}\left(r_H,C\right)$), its expected payoff is ${P_{\rm win}\left(s\right)} \Pi_{\rm win}\left(s\right) + \left(1-{P_{\rm win}\left(s\right)}\right)r_L$. 

Since ${\tilde b}\left(r_L,C\right)=``{\rm N} \textquotedblright$, we have the following relation:
\begin{align}
\nonumber
& D r_L+\left(1-D\right) \frac{K-1+{\eta^{\rm APO}}}{K}{r_L} \ge \\
& {P_{\rm win}\left(s\right)} \Pi_{\rm win}\left(s\right) + \left(1-{P_{\rm win}\left(s\right)}\right)r_L.\label{equ:rL}
\end{align}

Adding (\ref{equ:rH}) and (\ref{equ:rL}), we have
\begin{align}
\left(1-\!{P_{\rm win}\left(s\right)-\!D-\!\left(1-\!D\right) \frac{K-\!1+\!{\eta^{\rm APO}}}{K}}\right)\left(r_H-\!r_L\right)\!\ge \!0.\label{equ:combine}
\end{align}

Recall that $1-D$ is the probability that $b_{\min}^{-k}=``{\rm N} \textquotedblright$, \emph{i.e.}, all the other $K-1$ APOs bid $``{\rm N} \textquotedblright$. 
If APO $k$ bids $s$ and all the other $K-1$ APOs bid $``{\rm N} \textquotedblright$, APO $k$ definitely wins. Hence, we have $P_{\rm win}\left(s\right)\ge 1-D$. Based on this, we further have the following inequality:
\begin{align}
\nonumber
& \left(1-{P_{\rm win}\left(s\right)-D-\left(1-D\right) \frac{K-1+{\eta^{\rm APO}}}{K}}\right)\left(r_H-r_L\right)\\ 
& \le -\left(1-D\right) \frac{K-1+{\eta^{\rm APO}}}{K}\left(r_H-r_L\right).\label{equ:relax} 
\end{align}

We discuss the following two cases:

\emph{Case A:} $D<1$. In this case, the right hand side of (\ref{equ:relax}) is smaller than $0$, which contradicts with (\ref{equ:combine}).

\emph{Case B:} $D=1$. In this case, we can always find an APO type $r_M\in\left(r_L,r_H\right)$ such that $t\triangleq{\tilde b}\left(r_M,C\right)\in\left[0,C\right]$ and $P_{\rm win}\left(t\right)>0$, where $P_{\rm win}\left(t\right)$ is defined as the probability of winning the auction under bid $t$ (assuming the other $K-1$ APOs adopt strategy $\tilde b$). 
We can prove this by contradiction. We assume that we cannot find such an APO type, then the probability of any APO with type in $\left(r_L,r_H\right)$ winning the auction under its bid given by strategy $\tilde b$ is zero (assuming the other $K-1$ APOs adopt strategy $\tilde b$). 
When we randomly pick the $K$ APOs according to the probability distribution $f\left(\cdot\right)$, there is a chance that all the APOs have types in interval $\left(r_L,r_H\right)$. Based on our result, in this case, when all the APOs adopt strategy $\tilde b$, there is no winner in the auction with probability one. 
On the other hand, $D=1$ implies that, when all the APOs adopt strategy $\tilde b$, these $K$ APOs bid from $\left[0,C\right]$ and there is a winner in the auction with probability one. Hence, there is a contradiction and we can always find an APO type $r_M$ as described above.

Similar as $\Pi_{\rm win}\left(s\right)$, we define $\Pi_{\rm win}\left(t\right)$ as the type-$r_M$ APO's payoff when it wins with bid $t$, assuming the other APOs adopt strategy $\tilde b$. For the APO with type $r_M$, ${\tilde b}\left(r_M,C\right)=t$. Therefore, we have:
\begin{align}
\nonumber
&{P_{\rm win}\left(t\right)} \Pi_{\rm win}\left(t\right) + \left(1-{P_{\rm win}\left(t\right)}\right)r_M\\
& \ge D r_M+\left(1-D\right) \frac{K-1+{\eta^{\rm APO}}}{K}{r_M}= r_M.
\end{align}

After arrangement, we have
\begin{align}
P_{\rm win}\left(t\right) r_M \le {P_{\rm win}\left(t\right)} \Pi_{\rm win}\left(t\right).
\end{align}
Because $P_{\rm win}\left(t\right)>0$, we obtain $r_M\le \Pi_{\rm win}\left(t\right)$. Since $r_M\in\left(r_L,r_H\right)$, we further have
\begin{align}
r_L< \Pi_{\rm win}\left(t\right).\label{equ:rL1}
\end{align}
Recall that we have assumed ${\tilde b}\left(r_L,C\right)=``{\rm N} \textquotedblright$. Hence, we have the following relation:
\begin{align}
\nonumber
& D r_L+\left(1-D\right) \frac{K-1+{\eta^{\rm APO}}}{K}{r_L} \\
& \ge {P_{\rm win}\left(t\right)} \Pi_{\rm win}\left(t\right) + \left(1-{P_{\rm win}\left(t\right)}\right)r_L.
\end{align}
After applying $D=1$ and $P_{\rm win}\left(t\right)>0$, we obtain
\begin{align} 
r_L \ge \Pi_{\rm win}\left(t\right).
\end{align}
This contradicts with (\ref{equ:rL1}).

Therefore, there are contradictions for both \emph{Case A} and \emph{Case B}. We conclude that, if $r_L<r_H$ and ${\tilde b}\left(r_H,C\right)\ne``{\rm N} \textquotedblright$, we have ${\tilde b}\left(r_L,C\right)\ne``{\rm N} \textquotedblright$, which completes the proof of {\bf Part I}.

{\bf Part II:} Now we show that, if $r_L<r_H$ and ${\tilde b}\left(r_H,C\right)\ne``{\rm N} \textquotedblright$, we have ${\tilde b}\left(r_L,C\right) \le {\tilde b}\left(r_H,C\right)$.

Similar as {\bf Part I}, we define $s\triangleq {\tilde b}\left(r_H,C\right)$, $P_{\rm win}\left(s\right)$, and $\Pi_{\rm win}\left(s\right)$. We also define $o\triangleq {\tilde b}\left(r_L,C\right)$, $P_{\rm win}\left(o\right)$, and $\Pi_{\rm win}\left(o\right)$. We prove $o\le s$ by contradiction, \emph{i.e.}, we assume that $o>s$.

For an APO with type $r_H$, we have the following relation:
\begin{align}
\nonumber
& {P_{\rm win}\left(s\right)} \Pi_{\rm win}\left(s\right) + \left(1-{P_{\rm win}\left(s\right)}\right)r_H \\
& \ge {P_{\rm win}\left(o\right)} \Pi_{\rm win}\left(o\right) + \left(1-{P_{\rm win}\left(o\right)}\right)r_H.\label{equ:partB:rH}
\end{align}

For an APO with type $r_L$, we have the following relation:
\begin{align}
\nonumber
&{P_{\rm win}\left(o\right)} \Pi_{\rm win}\left(o\right) + \left(1-{P_{\rm win}\left(o\right)}\right)r_L\\ 
& \ge {P_{\rm win}\left(s\right)} \Pi_{\rm win}\left(s\right) + \left(1-{P_{\rm win}\left(s\right)}\right)r_L.\label{equ:partB:rL}
\end{align}

Adding (\ref{equ:partB:rH}) and (\ref{equ:partB:rL}), we have:
\begin{align}
\left({P_{\rm win}\left(o\right)}-{P_{\rm win}\left(s\right)}\right)\left(r_H - r_L\right) \ge 0.
\end{align} 
From $r_H>r_L$, we have ${P_{\rm win}\left(o\right)}\ge {P_{\rm win}\left(s\right)}$. Since $s<o$, naturally we have ${P_{\rm win}\left(s\right)}\ge{P_{\rm win}\left(o\right)}$, \emph{i.e.}, given the fact that the other APOs use strategy $\tilde b$, the probability that an APO wins with bid $s$ is not smaller than the probability that it wins with a higher bid $o$. Hence, the only possibility is that ${P_{\rm win}\left(o\right)}= {P_{\rm win}\left(s\right)}$, which means the probability for a particular APO to bid between interval $\left[s,o\right]$ under strategy $\tilde b$ is zero.

We can find a type $r_G\in\left(r_L,r_H\right)$ such that $P_{\rm win}\left(u\right)\ne P_{\rm win}\left(s\right),P_{\rm win}\left(o\right)$, where $u\triangleq {\tilde b}\left(r_G,C\right)$ and $P_{\rm win}\left(u\right)$ is the probability that an APO of type $r_G$ wins the auction under bid $u$, assuming the other APOs choose strategy ${\tilde b}$. 
We prove this by contradiction. We suppose that for any APO type $r_G\in\left(r_L,r_H\right)$, we have $P_{\rm win}\left(u\right)=P_{\rm win}\left(s\right)=P_{\rm win}\left(o\right)$. 
When we randomly pick an APO according to the probability distribution $f\left(\cdot\right)$, there is a chance that the APO's type is from $\left(r_L,r_H\right)$. Based on our assumption, when it bids based on ${\tilde b}$, it should bid the same value with probability one. 
We denote this bid value as $v$. Furthermore, we can always find an APO type $r_F\in\left(r_L,r_H\right)$ such that ${\tilde b}\left(r_F,C\right)=v$. Based on our assumption, we have $P_{\rm win}\left(v\right)=P_{\rm win}\left(s\right)=P_{\rm win}\left(o\right)$. However, this contradicts with the fact that there is a positive chance that a randomly picked APO bids $v$ and $s\ne o$. Hence, we can always find an APO type $r_G$ as described above.

We discuss the following two cases.

\emph{Case A:} $u\in\left[0,o\right)$. We have the following relation for a type $r_G$ APO:
\begin{align}
\nonumber
& {P_{\rm win}\left(u\right)} \Pi_{\rm win}\left(u\right) + \left(1-{P_{\rm win}\left(u\right)}\right)r_G\\ 
& \ge {P_{\rm win}\left(o\right)} \Pi_{\rm win}\left(o\right) + \left(1-{P_{\rm win}\left(o\right)}\right)r_G.
\end{align}
After rearrangement, we have
\begin{align}
{P_{\rm win}\left(u\right)} \Pi_{\rm win}\left(u\right) \!- \!{P_{\rm win}\left(o\right)} \Pi_{\rm win}\left(o\right)  \ge \left(P_{\rm win}\left(u\right)\!-\!{P_{\rm win}\left(o\right)}\right)r_G.\label{equ:partB:rK}
\end{align}

Since $u<o$ and $P_{\rm win}\left(u\right)\ne P_{\rm win}\left(o\right)$, we have $P_{\rm win}\left(u\right)> P_{\rm win}\left(o\right)$. From (\ref{equ:partB:rK}) and $r_G>r_L$, we have
\begin{align}
\nonumber
& {P_{\rm win}\left(u\right)} \Pi_{\rm win}\left(u\right) - {P_{\rm win}\left(o\right)} \Pi_{\rm win}\left(o\right) \\
& > \left(P_{\rm win}\left(u\right)-{P_{\rm win}\left(o\right)}\right)r_L.
\end{align}

After rearrangement, we have
\begin{align}
\nonumber
& {P_{\rm win}\left(u\right)} \Pi_{\rm win}\left(u\right) + \left(1-{P_{\rm win}\left(u\right)}\right)r_L \\
&> {P_{\rm win}\left(o\right)} \Pi_{\rm win}\left(o\right) + \left(1-{P_{\rm win}\left(o\right)}\right)r_L.
\end{align}
This means an APO of type $r_L$ prefers the bid $u$ over $o$, which contradicts with the fact that $o={\tilde b}\left(r_L,C\right)$. Therefore, $u$ can not be in interval $\left[0,s\right)$.

\emph{Case B:} $u\in\left[o,C\right]$. We have the following relation for a type $r_G$ APO:
\begin{align}
\nonumber
&{P_{\rm win}\left(u\right)} \Pi_{\rm win}\left(u\right) + \left(1-{P_{\rm win}\left(u\right)}\right)r_G\\
&\ge {P_{\rm win}\left(s\right)} \Pi_{\rm win}\left(s\right) + \left(1-{P_{\rm win}\left(s\right)}\right)r_G.
\end{align}
After rearrangement, we have
\begin{align}
\left(P_{\rm win}\left(s\right)-{P_{\rm win}\left(u\right)}\right)r_G  \ge {P_{\rm win}\left(s\right)} \Pi_{\rm win}\left(s\right) -{P_{\rm win}\left(u\right)} \Pi_{\rm win}\left(u\right) .\label{equ:partB:rK2}
\end{align}
Since $u>s$ and $P_{\rm win}\left(u\right)\ne P_{\rm win}\left(s\right)$, we have $P_{\rm win}\left(u\right)<P_{\rm win}\left(s\right)$. From (\ref{equ:partB:rK2}) and $r_G<r_H$, we have
\begin{align}
\nonumber
& \left(P_{\rm win}\left(s\right)-{P_{\rm win}\left(u\right)}\right)r_H \\
& > {P_{\rm win}\left(s\right)} \Pi_{\rm win}\left(s\right) -{P_{\rm win}\left(u\right)} \Pi_{\rm win}\left(u\right).
\end{align}
After rearrangement, we have
\begin{align}
\nonumber
& {P_{\rm win}\left(u\right)} \Pi_{\rm win}\left(u\right) + \left(1-{P_{\rm win}\left(u\right)}\right)r_H \\
& > {P_{\rm win}\left(s\right)} \Pi_{\rm win}\left(s\right) + \left(1-{P_{\rm win}\left(s\right)}\right)r_H.
\end{align}
This means an APO of type $r_H$ prefers the bid $u$ over $s$, which contradicts with the fact that $s={\tilde b}\left(r_H,C\right)$. Therefore, $u$ can not be in interval $\left[o,C\right]$.

Based on \emph{Case A} and \emph{Case B}, we prove that $o\le s$. In other words, if $r_L<r_H$ and ${\tilde b}\left(r_H,C\right)\ne``{\rm N} \textquotedblright$, we have ${\tilde b}\left(r_L,C\right) \le {\tilde b}\left(r_H,C\right)$. This completes the proof of {\bf Part II}.

Combining {\bf Part I} and {\bf Part II}, we complete the proof of the lemma.
\end{proof}

\begin{lemma}\label{lemma:activity}
For a strategy function $\tilde b$ which constitutes an SBNE, we use $Q$ to denote the probability that an APO bids a value from set $\left[0,C\right]$ (\emph{i.e.}, not bid ``{\rm N} \textquotedblright) under strategy $\tilde b$. We have $Q<1$.
\end{lemma}

\begin{proof}
We prove it by contradiction: we assume that $Q=1$. Then we can always find APO $k$ with type $r_Z\in\left(C,r_{\max}\right]$ that bids a value from $\left[0,C\right]$. We denote $s\triangleq {\tilde b}\left(r_Z,C\right)$. Therefore, we have the following relation for APO $k$'s payoffs under bid $s$ and $``{\rm N} \textquotedblright$:
\begin{align}
{P_{\rm win}\left(s\right)} \Pi_{\rm win}\left(s\right) + \left(1-{P_{\rm win}\left(s\right)}\right)r_Z\ge r_Z,\label{equ:rZ}
\end{align}
where the definitions of $P_{\rm win}\left(s\right)$ and $\Pi_{\rm win}\left(s\right)$ can be found in the proof of Lemma \ref{lemma:monotonicity}. By Lemma \ref{lemma:monotonicity} and ${\tilde b}\left(r_Z,C\right)=s$, all APOs with types in $\left(r_Z,r_{\max}\right]$ will bid $``{\rm N}\textquotedblright$ or values no smaller than $s$. This means $P_{\rm win}\left(s\right)>0$. Furthermore, based on the payment rule, we have $\Pi_{\rm win}\left(s\right)\le C<r_Z$. Hence, we have
\begin{align}
{P_{\rm win}\left(s\right)} \Pi_{\rm win}\left(s\right) + \left(1-{P_{\rm win}\left(s\right)}\right)r_Z < r_Z,
\end{align}
which contradicts with (\ref{equ:rZ}). Therefore, $Q$ can only be smaller than $1$, which completes the proof.
\end{proof}

\begin{lemma}\label{lemma:rA}
For a strategy function $\tilde b$ which constitutes an SBNE, there exists $r_A\in\left(r_{\min},r_{\max}\right)$ such that ${\tilde b}\left(r,C\right)=``{\rm N} \textquotedblright$ for all $r\in\left(r_A,r_{\max}\right]$ and ${\tilde b}\left(r,C\right)\in\left[0,C\right]$ for all $r\in\left[r_{\min},r_{A}\right)$.
\end{lemma}

\begin{proof}
According to Lemma \ref{lemma:activity}, there is a positive measure of types that bid $``{\rm N} \textquotedblright$ under $\tilde b$. Based on Lemma \ref{lemma:monotonicity}, the set of these APO types should have the form of $\left(r_A,r_{\max}\right]$ or $\left[r_A,r_{\max}\right]$, where $r_A\in\left[r_{\min},r_{\max}\right)$. 

Suppose $r_A=r_{\min}$, \emph{i.e.}, ${\tilde b}\left(r,C\right)=``{\rm N}\textquotedblright$ for all $r\in\left[r_{\min},r_{\max}\right]$. 
Since $C\ge r_{\min}>\frac{K-1+{\eta^{\rm APO}}}{K}r_{\min}$, we can find an APO type $\hat r\in\left(r_{\min},r_{\max}\right)$ such that $C>\frac{K-1+{\eta^{\rm APO}}}{K}{\hat r}$. Then such an APO can bid a value from $\left[0,C\right]$ and win a payoff of $C$. Since $C>\frac{K-1+{\eta^{\rm APO}}}{K}{\hat r}$, the obtained payoff is greater than the APO's payoff $\frac{K-1+{\eta^{\rm APO}}}{K}{\hat r}$ when it bids $``{\rm N} \textquotedblright$. In other words, the APO with type $\hat r$ will deviate from $``{\rm N} \textquotedblright$. This means $r_A$ can not be equal to $r_{\min}$.

Therefore, we have shown that, there exists $r_A\in\left(r_{\min},r_{\max}\right)$ such that ${\tilde b}\left(r,C\right)=``{\rm N} \textquotedblright$ for all $r\in\left(r_A,r_{\max}\right]$ and ${\tilde b}\left(r,C\right)\in\left[0,C\right]$ for all $r\in\left[r_{\min},r_{A}\right)$. 
\end{proof}

\begin{lemma}\label{lemma:rB}
For a strategy function $\tilde b$ which constitutes an SBNE, there exists $r_B\in\left[r_{\min},r_A\right)$ such that ${\tilde b}\left(r,C\right)=C$ for all $r\in\left(r_B,r_{A}\right)$ and ${\tilde b}\left(r,C\right)\in\left[0,C\right)$ for all $r\in\left[r_{\min},r_{B}\right)$.
\end{lemma}

\begin{proof} We explain the proof by the following two parts.

{\bf Part I:} We first prove that, if there exist $r_L,r_H\in\left[r_{\min},r_{\max}\right]$ and $r_L\ne r_H$ such that  ${\tilde b}\left(r,C\right)=X$ for all $r\in\left(r_{L},r_{H}\right)$, then $X=C$. We show this by contradiction, \emph{i.e.}, we assume that $X<C$. We discuss the following two cases.

\emph{Case A:} $X<r_H$. We can find ${\bar r}_H\in\left(r_L,r_H\right)$ such that $X<{\bar r}_H$. Furthermore, we can find $\epsilon>0$ such that $X+\epsilon <{\bar r}_H$ and $X+\epsilon <C$. Next we show that APO $k$ with type ${\bar r}_H$ has the incentive to deviate from bid $X$ to bid $X+\epsilon$. 
Note that bid $X$ and bid $X+\epsilon$ generate a difference to APO $k$ only when $b_{\min}^{-k}\in\left[X,X+\epsilon\right]$. In particular, we discuss the following two situations:
\begin{itemize}
\item When $b_{\min}^{-k}=X$, bidding $X$ generates an expected payoff with a value in $\left(X,{\bar r}_H\right)$, and bidding $X+\epsilon$ generates a payoff of ${\bar r}_H$. Because $X<{\bar r}_H$, in this case, bidding $X+\epsilon$ is strictly better;
\item When $b_{\min}^{-k}\in\left(X,X+\epsilon\right]$, bidding $X$ generates a payoff that is not greater than $X+\epsilon$, and bidding $X+\epsilon$ generates an expected payoff that is greater than $X+\epsilon$. Hence, in this case, bidding $X+\epsilon$ is strictly better.
\end{itemize}
Notice that, since ${\tilde b}\left(r,C\right)=X$ for all $r\in\left(r_{L},r_{H}\right)$, $b_{\min}^{-k}=X$ happens with a non-zero probability. Therefore, it is strictly better for the type-${\bar r}_H$ APO to deviate to bid $X+\epsilon$. This violates the assumption that ${\tilde b}\left(r,C\right)=X$ for all $r\in\left(r_{L},r_{H}\right)$.

\emph{Case B:} $X\ge r_H$. We can find ${\hat r}_H\in\left(r_L,r_H\right)$ such that $X>{\hat r}_H$. Furthermore, we can find $\xi>0$ such that  $X-\xi>{\hat r}_H$. Next we show that APO $k$ with type ${\hat r}_H$ has the incentive to deviate from bid $X$ to bid $X-\xi$. Note that bid $X$ and bid $X-\xi$ generate a difference to APO $k$ only when $b_{\min}^{-k}\in\left[X-\xi,X\right]$. In particular, we discuss the following two situations:
\begin{itemize}
\item When $b_{\min}^{-k}\in\left[X-\xi,X\right)$, bidding $X$ generates a payoff of ${\hat r}_H$, and bidding $X-\xi$ generates an expected payoff that is greater than ${\hat r}_H$. Hence, in this case, bidding $X-\xi$ is strictly better;
\item When $b_{\min}^{-k}=X$, bidding $X$ generates an expected payoff with a value in $\left({\hat r}_H,X\right)$, while bidding $X-\xi$ generates a payoff of $X$. Because $X>{\hat r}_H$, in this case, bidding $X-\xi$ is strictly better.
\end{itemize}
Notice that, since ${\tilde b}\left(r,C\right)=X$ for all $r\in\left(r_{L},r_{H}\right)$, $b_{\min}^{-k}=X$ happens with a non-zero probability. Therefore, it is strictly better for the type-${\hat r}_H$ APO to deviate to bid $X-\xi$. This violates the assumption that ${\tilde b}\left(r,C\right)=X$ for all $r\in\left(r_{L},r_{H}\right)$.

Combining \emph{Case A} and \emph{Case B}, we complete the proof of {\bf Part I}.

{\bf Part II:} We then prove that, there exists $\omega\in\left(0,r_A-r_{\min}\right]$ such that ${\tilde b}\left(r,C\right)=C$ for all $r\in\left[r_A-\omega,r_A \right)$. We show it by contradiction, \emph{i.e.}, we assume that for all $r<r_A$, we have ${\tilde b}\left(r,C\right)<C$ (note that from Lemma \ref{lemma:rA}, for all $r<r_A$, we have ${\tilde b}\left(r,C\right)\le C$).

We discuss the following two cases.

\emph{Case A:} $r_A\le C$. Naturally, we have $\frac{K-1+{\eta^{\rm APO}}}{K}r_A< C$ and we can find an APO type $r_U\in\left(r_A,r_{\max}\right)$ such that $\frac{K-1+{\eta^{\rm APO}}}{K}r_U<C$. According to Lemma \ref{lemma:rA}, we have ${\tilde b}\left(r_U,C\right)=``{\rm N} \textquotedblright$. Next we show that APO $k$ with type $r_U$ has an incentive to switch its bid from $``{\rm N} \textquotedblright$ to $C$. 

Note that bid $C$ and bid $``{\rm N} \textquotedblright$ generate a difference payoff to APO $k$ only when $b_{\min}^{-k}=C$ or $``{\rm N} \textquotedblright$. 
We do not need to consider the case where $b_{\min}^{-k}=C$, because all APO types in $\left[r_{\min},r_A\right)\cup\left(r_A,r_{\max}\right]$ do not bid $C$ based on our assumption, and the probability for an APO to be of type $r_A$ is zero. 
Therefore, we only need to consider $b_{\min}^{-k}=``{\rm N} \textquotedblright$. In this case, the payoffs of APO $k$ under bid $``{\rm N} \textquotedblright$ and bid $C$ are $\frac{K-1+{\eta^{\rm APO}}}{K}r_U$ and $C$, respectively. Since $C>\frac{K-1+{\eta^{\rm APO}}}{K}r_U$, in this case, bid $C$ is strictly better than bid $``{\rm N} \textquotedblright$. 
Based on Lemma \ref{lemma:rA}, there is a non-zero probability that $b_{\min}^{-k}=``{\rm N} \textquotedblright$. Therefore, we can conclude that APO $k$ with type $r_U$ has an incentive to switch its bid from $``{\rm N} \textquotedblright$ to $C$, which violates ${\tilde b}\left(r_U,C\right)=``{\rm N} \textquotedblright$.

\emph{Case B:} $r_A>C$. We can find an APO type $r_S\in\left(C,r_A\right)$. According to the assumption, we have $w\triangleq {\tilde b}\left(r_S,C\right)<C$. Next we show that APO $k$ with type $r_S$ has the incentive to switch from bidding $w$ to bidding $C$.
 
Note that bid $w$ and bid $C$ generate a difference payoff to APO $k$ only when $b_{\min}^{-k}\in\left[w,C\right]$. 
Furthermore, based on the conclusion in {\bf Part I} and Lemma \ref{lemma:monotonicity}, the probability that $b_{\min}^{-k}=w$ is zero. Based on the assumption in {\bf Part II}, the probability that $b_{\min}^{-k}=C$ is zero. Therefore, we do not need to consider the case where $b_{\min}^{-k}=w$ or $C$, and only need to compare bid $w$ and bid $C$ under the situation where $b_{\min}^{-k}\in\left(w,C\right)$. 
When $b_{\min}^{-k}\in\left(w,C\right)$, bid $w$ generates a payoff that is smaller than $C$, and bid $C$ generates a payoff of $r_S$. Since $r_S>C$, bid $C$ is strictly better than bid $w$ for APO $k$ with type $r_S$, which violates ${\tilde b}\left(r_S,C\right)=w$.

Based on \emph{Case A} and \emph{Case B}, the assumption that ${\tilde b}\left(r,C\right)<C$ for all $r<r_A$ does not hold. Therefore, we can find an APO type $r'\in\left[r_{\max},r_A\right)$ such that ${\tilde b}\left(r',C\right)=C$. Based on Lemma \ref{lemma:monotonicity}, we have ${\tilde b}\left(r,C\right)=C$ for all $r\in\left[r',r_A\right)$, where we complete the proof of {\bf Part II}.

Based on {\bf Part I}, {\bf Part II}, and Lemma \ref{lemma:monotonicity}, it is easy to show that there exists $r_B\in\left[r_{\min},r_A\right)$ such that ${\tilde b}\left(r,C\right)=C$ for all $r\in\left(r_B,r_{A}\right)$ and ${\tilde b}\left(r,C\right)\in\left[0,C\right)$ for all $r\in\left[r_{\min},r_{B}\right)$.
\end{proof}

\begin{lemma}\label{lemma:continuous}
For a strategy function $\tilde b$ which constitutes an SBNE, (1) ${\tilde b}\left(r_B,C\right)=C$; (2) ${\tilde b}\left(r,C\right)=r$ for all $r\in\left(r_{\min},r_B\right)$.
\end{lemma}

\begin{proof} We show the proof by the following two parts.

{\bf Part I:} We first prove ${\tilde b}\left(r_B,C\right)=C$ by contradiction, \emph{i.e.}, we assume that ${\tilde b}\left(r_B,C\right)=X<C$. Based on Lemma \ref{lemma:monotonicity}, Lemma \ref{lemma:rA}, and Lemma \ref{lemma:rB}, we have ${\tilde b}\left(r,C\right)< X$ for $r\in\left[r_{\min},r_B\right)$, ${\tilde b}\left(r,C\right)= C$ for $r\in\left(r_B,r_A\right)$, and ${\tilde b}\left(r,C\right)=``{\rm N} \textquotedblright$ for $r\in\left(r_A,r_{\max}\right]$. 

We discuss the following three cases.

\emph{Case A: $r_B>C$.} We compare bid $X$ and bid $C$ for APO $k$ with type $r_B$. When $b_{\min}^{-k}<X$ or $b_{\min}^{-k}=``{\rm N} \textquotedblright$, bid $X$ and bid $C$ generate the same payoff to APO $k$. Furthermore, $b_{\min}^{-k}\in\left[X,C\right)$ with probability zero. We only need to consider the situation where $b_{\min}^{-k}=C$. In this situation, bid $X$ generates $C$ and bid $C$ generates an expected payoff that strictly lies between $r_B$ and $C$. Since $r_B>C$, bid $C$ generates a higher payoff than bid $C$, which violates ${\tilde b}\left(r_B,C\right)=X$.

\emph{Case B: $r_B=C$.} We can find a type ${\bar r}_B\in\left[r_{\min},r_B\right)$ such that ${\bar r}_B>X$. According to ${\tilde b}\left(r_B,C\right)=X$ and Lemma \ref{lemma:monotonicity}, we have $s\triangleq {\tilde b}\left({\bar r}_B,C\right)<X$. Next we compare bid $X$ and bid $s$ for APO $k$ with type ${\bar r}_B$. 
Similar as the explanation before, we only need to consider the situation where $b_{\min}^{-k}\in\left(s,X\right)$. In this situation, bid $X$ generates ${\bar r}_B$, and bid $s$ generates a payoff that is smaller than $X$. Since ${\bar r}_B > X$, we find bid $X$ is strictly better than bid $s$ for APO $k$ with type ${\bar r}_B$, which violates ${\tilde b}\left({\bar r}_B,C\right)=s$.

\emph{Case C: $r_B<C$.} We can find a type ${\hat r}_B\in\left(r_B,r_A\right)$ such that $C>{\hat r}_B$. According to Lemma \ref{lemma:rB}, ${\tilde b}\left({\hat r}_B,C\right)= C$. Next we compare bid $X$ and bid $C$ for APO $k$ with type ${\hat r}_B$. Similar as the explanation before, we only need to consider the situation where $b_{\min}^{-k}=C$. In this situation, bid $X$ generates $C$ and bid $C$ generates an expected payoff smaller than $C$. Therefore, bid $X$ is strictly better than bid $C$ for APO $k$ with type ${\hat r}_B$, which violates ${\tilde b}\left({\hat r}_B,C\right)=C$.

Based on \emph{Case A}, \emph{Case B}, and \emph{Case C}, we conclude that ${\tilde b}\left(r_B,C\right)=C$, which completes the proof of {\bf Part I}.

{\bf Part II:} We next show that ${\tilde b}\left(r,C\right)=r$ for all $r\in\left(r_{\min},r_B\right)$. We consider an APO with type $r_M\in\left(r_{\min},r_B\right)$ and denote ${x}\triangleq {\tilde b}\left(r_M,C\right)$. From Lemma \ref{lemma:monotonicity}, function ${\tilde b}\left(r,C\right)$ is increasing for $r\in\left(r_{\min},r_B\right)$. Hence, the left-hand limit of ${\tilde b}\left(r,C\right)$ at $r=r_M$ exists, and we denote $u\triangleq \lim\limits_{r\uparrow {r_M}}{{\tilde b}\left(r,C\right)}$ (recall that $C$ is assumed to be fixed). Next we prove $u=r_M$ by contradiction. We assume that $u>r_M$ or $u<r_M$, and consider the following two cases.

\emph{Case A:} $u>r_M$. Since the left-hand limit of ${\tilde b}\left(r,C\right)$ at $r=r_M$ exists and equals $u$, we can find a $\zeta>0$ such that ${\tilde b}\left(r,C\right)\in\left(r_M,u\right)$ for all $r\in\left(r_M-\zeta,r_M\right)$. We then pick an APO type $r_L$ from interval $\left(r_M-\zeta,r_M\right)$ and denote ${\bar u}\triangleq {\tilde b}\left(r_L,C\right)$. Naturally, we have ${\bar u}>r_M$. 
Based on the monotonicity of ${\tilde b}\left(r,C\right)$ in Lemma \ref{lemma:monotonicity} and Lemma \ref{lemma:rB}, we also have $x>{\bar u}$. 
Furthermore, APOs with type in $\left(r_L,r_M\right)$ have ${\tilde b}\left(r,C\right)\in\left({\bar u},u\right)$. Next we compare bid $x$ and bid $\bar u$ for APO $k$ with type $r_M$. 
Notice that these two bids generate a difference to APO $k$ only when $b_{\min}^{-k}\in\left[{\bar u},x\right]$. 
Furthermore, the probability for $b_{\min}^{-k}={\bar u}$ or $x$ is zero. Hence, we only need to consider the situation where $b_{\min}^{-k}\in\left({\bar u},x\right)$. In this situation, bid $x$ causes a payoff of $r_M$, and bid $\bar u$ causes a payoff that is larger than $\bar u$. Since we have ${\bar u}>r_M$, we conclude that bid $\bar u$ is strictly better than bid $x$, which violates ${\tilde b}\left(r_M,C\right)=x$.

\emph{Case B:} $u<r_M$. We can find an APO type $r_S\in\left(u,r_M\right)$ and denote ${v}\triangleq {\tilde b}\left(r_S,C\right)$. Since $u$ is the left-hand limit of ${\tilde b}\left(r,C\right)$ at $r=r_M$, we have $ {\tilde b}\left(r,C\right)\in\left(v,u\right)$ for all $r\in\left(r_S,r_M\right)$. Next we compare bid $v$ and bid $u$ for APO $k$ with type $r_S$. 
Notice that these two bids generate a difference to APO $k$ only when $b_{\min}^{-k}\in\left[{v},u\right]$. 
Furthermore, the probability for $b_{\min}^{-k}={v}$ or $u$ is zero. Hence, we only need to consider the situation where $b_{\min}^{-k}\in\left({v},u\right)$. In this situation, bid $v$ causes a payoff that is smaller than $u$, and bid $u$ causes a payoff of $r_S$. 
Since we have $r_S>u$, we conclude that bid $u$ is strictly better than bid $v$, which violates ${\tilde b}\left(r_S,C\right)=v$.

We conclude that $u=r_M$, \emph{i.e.}, the left-hand limit of function ${\tilde b}\left(r,C\right)$ at $r=r_M$ equals $r_M$. With the similar approach, we can show that the right-hand limit of function ${\tilde b}\left(r,C\right)$ at $r=r_M$ also equals $r_M$. Together with the monotonicity property in Lemma \ref{lemma:monotonicity}, we can easily obtain that ${\tilde b}\left(r,C\right)$ is continuous at point $r=r_M$ and ${\tilde b}\left(r_M,C\right)=r_M$. Because $r_M$ can be any point from interval $\left(r_{\min},r_B\right)$, we conclude that ${\tilde b}\left(r,C\right)=r$ for all $r\in\left(r_{\min},r_B\right)$, which completes the proof of {\bf Part II}.

Combining {\bf Part I} and {\bf Part II}, we prove the lemma.
\end{proof}

\begin{lemma}\label{lemma:rBrA}
For a strategy function $\tilde b$ which constitutes an SBNE, $r_B=C$ and $r_A=r_T\left(C\right)$, where $r_T\left(C\right)$ is defined in Lemma \ref{lemma:rT}.
\end{lemma}

\begin{proof} We explain the proof in the following two parts.

{\bf Part I:} We first prove $r_B=C$. According to Lemma \ref{lemma:continuous}, we can easily show that $r_B \le C$. Next we assume $r_B=X<C$. Based on Lemma \ref{lemma:rB} and Lemma \ref{lemma:continuous}, we have ${\tilde b}\left(r,C\right)=C$ for all $r\in\left(r_B,r_A\right)$, and ${\tilde b}\left(r,C\right)\in\left[0,X\right)$ for all $r\in\left(r_{\min},r_B\right)$. Since $r_B<C$, we can find an APO type $r_M$ such that $r_M\in\left(r_B,C\right)$. Naturally, we have ${\tilde b}\left(r_M,C\right)=C$. Next we compare bid $C$ and bid $X$ for APO $k$ with type $r_M$. 
Note that bid $C$ and bid $X$ generate a difference to APO $k$ only when $b_{\min}^{-k}\in\left[X,C\right]$. 
Furthermore, the probability that $b_{\min}^{-k}\in\left[X,C\right)$ is zero. Hence, we only need to consider the situation where $b_{\min}^{-k}=C$. 
In this situation, bid $C$ generates an expected payoff in $\left(r_M,C\right)$ and bid $X$ generates $C$. Since we have $C>r_M$, we conclude that bid $X$ is strictly better than bid $C$, which violates ${\tilde b}\left(r_M,C\right)=C$. Therefore, the only possible relation between $r_B$ and $C$ is that $r_B=C$. We complete the proof of {\bf Part I}.

{\bf Part II:} We then prove $r_A=r_T\left(C\right)$, where $r_T\left(C\right)$ is defined in Lemma \ref{lemma:rT}. We first denote the payoff of an arbitrary APO type $r\in\left(r_{\min},r_{\max}\right)$ under bid $C$ and bid $``{\rm N} \textquotedblright$ as $\Pi_1$ and $\Pi_2$, respectively. 
We compute $\Pi_1$ as (\ref{makeup:a}), and compute $\Pi_2$ as (\ref{makeup:b}). Furthermore, we denote their difference as $W\left(r\right)$, and compute it as (\ref{makeup:c}).

\begin{figure*}
\begin{align}
\nonumber
 \Pi_1 \triangleq & \left(1-\left(1-F\left(C\right)\right)^{K-1}\right)r +\left(1-F\left(r_A\right)\right)^{K-1}C\\
&+\sum_{n=1}^{K-1}{\binom{K-1}{n} \left(F\left(r_A\right)-F\left(C\right)\right)^n\left(1-F\left(r_A\right)\right)^{K-1-n}\frac{C+nr}{n+1}}.\label{makeup:a}
\end{align}
\hrulefill
\begin{align}
\nonumber
 \Pi_2 \triangleq & \left(1-\left(1-F\left(C\right)\right)^{K-1}\right)r+\left(1-F\left(r_A\right)\right)^{K-1}\frac{K-1+{\eta^{\rm APO}}}{K}r\\
&+\sum_{n=1}^{K-1}{
\binom{K-1}{n} \left(F\left(r_A\right)-F\left(C\right)\right)^n\left(1-F\left(r_A\right)\right)^{K-1-n}r},\label{makeup:b}
\end{align}
\hrulefill
\begin{align}
\nonumber
W\left(r\right) & \triangleq \Pi_1-\Pi_2\\
&=\left(1-F\left(r_A\right)\right)^{K-1}\left(C-\frac{K-1+{\eta^{\rm APO}}}{K}r\right)+\sum_{n=1}^{K-1}{\binom{K-1}{n} \left(F\left(r_A\right)-F\left(C\right)\right)^n\left(1-F\left(r_A\right)\right)^{K-1-n}\frac{C-r}{n+1}}.\label{makeup:c}
\end{align}
\hrulefill
\end{figure*}

According to Lemma \ref{lemma:rA} and Lemma \ref{lemma:rB}, ${\tilde b}\left(r,C\right)=C$ for all $r\in\left[r_B,r_A\right)$ and ${\tilde b}\left(r,C\right)=``{\rm N} \textquotedblright$ for all $r\in\left(r_A,r_{\max}\right]$. Hence, $W\left(r\right)\ge0$ for all $r\in\left[r_B,r_A\right)$ and $W\left(r\right)\le0$ for all $r\in\left(r_A,r_{\max}\right]$. It is easy to find that $W\left(r\right)$ is strictly decreasing and continuous for $r\in\left[r_B,r_{\max}\right]$. Hence, we have $W\left(r_A\right)=0$, which means $r_A$ needs to satisfy the following equation:
\begin{multline}
\!\!\sum_{n=1}^{K-1}\!{\binom{K-1}{n} \!\left(F\left(r_A\right)\!-\!F\left(C\right)\right)^n\left(1\!-\!F\left(r_A\right)\right)^{K-1-n}\frac{C\!-\!r_A}{n+1}}\\
+\left(1-F\left(r_A\right)\right)^{K-1}\left(C-\frac{K-1+{\eta^{\rm APO}}}{K}r_A\right)=0.
\end{multline}
Furthermore, from Lemma \ref{lemma:rA}, Lemma \ref{lemma:rB}, and $r_B=C$, we have $r_A\in\left(C,r_{\max}\right)$.

In Lemma \ref{lemma:rT}, we define $r_T\left(C\right)\in\left(C,r_{\max}\right)$ as the solution to 
\begin{multline}
\sum_{n=1}^{K-1}{\binom{K-1}{n} \left(F\left(r\right)-F\left(C\right)\right)^n\left(1-F\left(r\right)\right)^{K-1-n}\frac{C-r}{n+1}}\\
+\left(1-F\left(r\right)\right)^{K-1}\left(C-\frac{K-1+{\eta^{\rm APO}}}{K}r\right)=0.
\end{multline} 
Hence, $r_A=r_T\left(C\right)$, and we complete the proof of {\bf Part II}.

Combining {\bf Part I} and {\bf Part II}, we complete the proof of the lemma.
\end{proof}

\begin{lemma}\label{lemma:necessityall}
For a strategy function $\tilde b$ which constitutes an SBNE, it needs to satisfy the following conditions:\\
(1) ${\tilde b}\left(r_{\min},C\right)\in \left[0,r_{\min}\right]$;\\
(2) ${\tilde b}\left(r,C\right)=r$ for all $r\in\left(r_{\min},C\right)$;\\
(3) ${\tilde b}\left(r,C\right)=C$ for all $r\in\left[C,r_T\left(C\right)\right)$;\\
(4) ${\tilde b}\left(r_T\left(C\right),C\right)=C{~\rm or~}``{\rm N} \textquotedblright$;\\
(5) ${\tilde b}\left(r,C\right)=``{\rm N} \textquotedblright$ for all $r\in\left(r_T\left(C\right),r_{\max}\right]$.
\end{lemma}

\begin{proof}
Items (2), (3), and (5) can be directly obtained from Lemma \ref{lemma:rA}, Lemma \ref{lemma:rB}, Lemma \ref{lemma:continuous}, and Lemma \ref{lemma:rBrA}. For item (1), since ${\tilde b}\left(r,C\right)=r$ for all $r\in\left(r_{\min},C\right)$, it is natural that the bidder of type $r_{\min}$ should not bid a value that is larger than $r_{\min}$ (based on the monotonicity in Lemma \ref{lemma:monotonicity}). For item (4), since ${\tilde b}\left(r,C\right)=C$ for all $r\in\left[C,r_T\left(C\right)\right)$ and ${\tilde b}\left(r,C\right)=``{\rm N} \textquotedblright$ for all $r\in\left(r_T\left(C\right),r_{\max}\right]$, it is easy to show that the bidder of type $r_T\left(C\right)$ can only bid $C$ or $``{\rm N} \textquotedblright$ (from the monotonicity in Lemma \ref{lemma:monotonicity}).
\end{proof}

\subsection{Proof of Theorem \ref{theorem:combine:unique}}\label{appendix:sec:theorem2}
From Lemma \ref{lemma:necessityall}, the bidding strategy in (\ref{equ:equilibrium}) is the unique strategy that satisfies all conditions for an equilibrium strategy. Furthermore, from Theorem \ref{theorem:equilibrium}, the bidding strategy in (\ref{equ:equilibrium}) constitutes an SBNE. Therefore, it is the unique form of the bidding strategy under an SBNE.

\subsection{Proofs of Theorem \ref{theorem:lowC}, Theorem \ref{theorem:middleC}, and Theorem \ref{theorem:highC}}

Notice that in Section \ref{subsec:stageII:1}, we consider the situation where $C$ is from $\left[r_{\min},r_{\max}\right)$, and show the three-region structure of the bidding strategies: some APO types bid their types, some APO types bid the reserve rate, and some APO types bid $``\rm N \textquotedblright$. In Sections \ref{subsec:lowC}, \ref{subsec:stageII:3} and \ref{subsec:stageII:4}, we consider the situation where $C$ is from $C\in\left[0,r_{\min}\right)\cup\left[r_{\max},\infty\right)$. Since in this situation, the reserve rate is either very small or very large, among the three regions of APOs' strategies introduced in Section \ref{subsec:stageII:1}, only one or two will appear in this case. Hence, the equilibrium analysis in Sections \ref{subsec:lowC}, \ref{subsec:stageII:3} and \ref{subsec:stageII:4} can be treated as the special cases of the analysis in Section \ref{subsec:stageII:1}.

Therefore, the proofs of Theorem \ref{theorem:lowC}, Theorem \ref{theorem:middleC}, and Theorem \ref{theorem:highC} are similar to those of Theorem \ref{theorem:equilibrium} and Theorem \ref{theorem:combine:unique}. The details are omitted.

\subsection{Proof of Lemma \ref{lemma:rX}}
Let $H\left(r\right)$ be the left hand side of equation (\ref{equ:rX}). It is easy to find that $H\left(r\right)$ is continuous in $\left[r_{\min},r_{\max}\right]$. Furthermore, we have
\begin{align}
& H\left(r_{\min}\right)=C-\frac{K-1+{\eta^{\rm APO}}}{K}r_{\min}>0,\\
& H\left(r_{\max}\right)=\frac{C-r_{\max}}{K}<0.
\end{align}
Hence, there is at least one solution in $\left(r_{\min},r_{\max}\right)$ satisfying $H\left(r\right)=0$.

\subsection{Preliminary Lemma}

In order to prove Theorem \ref{theorem:optimalCcase}, we introduce a preliminary lemma in this section.

\begin{lemma}\label{lemma:local:etarmin}
(a) When $R_{\rm LTE}>\frac{K-1+{\eta^{\rm APO}}}{K\left(1-{\delta^{\rm LTE}}\right)}r_{\min}$, $C=\frac{K-1+{\eta^{\rm APO}}}{K}r_{\min}$ is a local minimum of ${\bar \Pi}^{\rm LTE}\left(C\right)$; (b) when $R_{\rm LTE}\le \frac{K-1+{\eta^{\rm APO}}}{K\left(1-{\delta^{\rm LTE}}\right)}r_{\min}$, $C=\frac{K-1+{\eta^{\rm APO}}}{K}r_{\min}$ is a global maximum of ${\bar \Pi}^{\rm LTE}\left(C\right)$.
\end{lemma}

\begin{proof}
We prove the lemma by the following two parts.

{\bf Part I:} We show that when $R_{\rm LTE}>\frac{K-1+{\eta^{\rm APO}}}{K\left(1-{\delta^{\rm LTE}}\right)}r_{\min}$, $C=\frac{K-1+{\eta^{\rm APO}}}{K}r_{\min}$ is a local minimum of ${\bar \Pi}^{\rm LTE}\left(C\right)$. First, it is easy to show that ${\bar \Pi}^{\rm LTE}\left(C\right)$ is continuous at point $C=\frac{K-1+{\eta^{\rm APO}}}{K}r_{\min}$. 
According to (\ref{equ:LTEpayoff:relow}), we can compute (\ref{equ:rightlimit}). 
We take the derivates of both sides of equation (\ref{equ:rX}) over $r$, and obtain (\ref{equ:rightlimit:follow}).

\begin{figure*}
\begin{align}
\lim_{C\downarrow {{\frac{K-1+{\eta^{\rm APO}}}{K}r_{\min}}}}\frac{d {\bar \Pi}^{\rm LTE}\left(C\right)}{d C}= \lim_{C\downarrow {{\frac{K-1+{\eta^{\rm APO}}}{K}r_{\min}}}} K f\left(r_{\min}\right) \frac{d r_X\left(C\right)}{d C} \left(\left(1-{\delta^{\rm LTE}}\right)R_{\rm LTE}-\frac{K-1+{\eta^{\rm APO}}}{K} {r_{\min}}\right).\label{equ:rightlimit}
\end{align}
\hrulefill
\begin{align}
\left(\frac{K-1+\eta^{\rm APO}}{K}+\frac{\left(K-1\right)}{2}f\left(r_{\min}\right){\frac{1-\eta^{\rm APO}}{K}r_{\min}}\right) \lim_{C\downarrow {{\frac{K-1+{\eta^{\rm APO}}}{K}r_{\min}}}}\frac{d r_X\left(C\right)}{d C} = 1.\label{equ:rightlimit:follow}
\end{align}
\hrulefill
\end{figure*}

Therefore, we have $\lim_{C\downarrow {{\frac{K-1+{\eta^{\rm APO}}}{K}r_{\min}}}}\frac{d r_X\left(C\right)}{d C}>0$. Based on this, condition $R_{\rm LTE}>\frac{K-1+{\eta^{\rm APO}}}{K\left(1-{\delta^{\rm LTE}}\right)}r_{\min}$, and (\ref{equ:rightlimit}), we conclude that $\lim_{C\downarrow \frac{K-1+{\eta^{\rm APO}}}{K}r_{\min}} \frac{d {{\bar \Pi}^{\rm LTE}}\left(\!C\right)}{dC}>0$. 
Furthermore, $\frac{d {\bar \Pi}^{\rm LTE}\left(C\right)}{d C}$ is continuous for $C\in\left(\frac{K-1+{\eta^{\rm APO}}}{K}r_{\min},r_{\min}\right)$. Hence, there exists an $\epsilon>0$ such that $\frac{d {\bar \Pi}^{\rm LTE}\left(C\right)}{d C}>0$ for all $C\in\left(\frac{K-1+{\eta^{\rm APO}}}{K}r_{\min},\frac{K-1+{\eta^{\rm APO}}}{K}r_{\min}+\epsilon\right)$. On the other hand, since ${\bar \Pi}^{\rm LTE}\left(C\right)$ is a constant for $C\in\left[0,\frac{K-1+{\eta^{\rm APO}}}{K}\!r_{\min}\right]$ based on (\ref{equ:LTEpayoff:smallC}), we have $\frac{d {\bar \Pi}^{\rm LTE}\left(C\right)}{d C}=0$ for $C\in\left[0,\frac{K-1+{\eta^{\rm APO}}}{K}r_{\min}\right)$. Therefore, we conclude that $C=\frac{K-1+{\eta^{\rm APO}}}{K}r_{\min}$ is the local minimum of function ${\bar \Pi}^{\rm LTE}\left(C\right)$. This completes the proof of {\bf Part I}.

{\bf Part II:} we show that when $R_{\rm LTE}\le \frac{K-1+{\eta^{\rm APO}}}{K\left(1-{\delta^{\rm LTE}}\right)}r_{\min}$, $C=\frac{K-1+{\eta^{\rm APO}}}{K}r_{\min}$ is a global maximum of ${\bar \Pi}^{\rm LTE}\left(C\right)$. Notice that when $C\in\left[0,\frac{K-1+{\eta^{\rm APO}}}{K}r_{\min}\right]$, ${\bar \Pi}^{\rm LTE}\left(C\right)={\delta^{\rm LTE}} R_{\rm LTE}$. We compare this value with the LTE's payoff under other intervals of $C$ as follows:

\emph{Case A:} $C\in\left(\frac{K-1+{\eta^{\rm APO}}}{K}r_{\min},r_{\min}\right)$. In this case, the LTE provider's payoff contains two parts: $\left(1-F\left(r_X\left(C\right)\right)\right)^2 {\delta^{\rm LTE}} R_{\rm LTE}$ and $\left(1-\left(1-F\left(r_X\left(C\right)\right)\right)^2 \right) \left(R_{\rm LTE}-C\right)$, and both of the coefficients (\emph{i.e.}, $\left(1-F\left(r_X\left(C\right)\right)\right)^2$ and $\left(1-\left(1-F\left(r_X\left(C\right)\right)\right)^2 \right)$) are positive. Because $R_{\rm LTE}\le \frac{K-1+{\eta^{\rm APO}}}{K\left(1-{\delta^{\rm LTE}}\right)}r_{\min}$, we have $R_{\rm LTE}-  \frac{K-1+{\eta^{\rm APO}}}{K}r_{\min} \le {\delta^{\rm LTE}} R_{\rm LTE}$. Moreover, we have $C>\frac{K-1+{\eta^{\rm APO}}}{K}r_{\min}$. Hence, we can obtain $R_{\rm LTE}- C < {\delta^{\rm LTE}} R_{\rm LTE}$, which means the LTE's payoff under $C\in\left(\frac{K-1+{\eta^{\rm APO}}}{K}r_{\min},r_{\min}\right)$ is always smaller than that under $C=\frac{K-1+{\eta^{\rm APO}}}{K}r_{\min}$;

\emph{Case B:} $C\in\left[r_{\min},r_{\max}\right)$. In this case, if there is a competition, LTE obtains ${\delta^{\rm LTE}} R_{\rm LTE}$, otherwise, it obtains the difference between $R_{\rm LTE}$ and the payment to APOs. In the latter case, the LTE's payoff is no larger than $R_{\rm LTE}-r_{\min}$. Since $r_{\min}>\frac{K-1+{\eta^{\rm APO}}}{K}r_{\min}\ge \left(1-\delta\right)R_{\rm LTE}$, the LTE's payoff in the latter case is smaller than ${\delta^{\rm LTE}} R_{\rm LTE}$, which is the payoff under $C=\frac{K-1+{\eta^{\rm APO}}}{K}r_{\min}$. In other words, the LTE's payoff under $C\in\left[r_{\min},r_{\max}\right)$ is always smaller than that under $C=\frac{K-1+{\eta^{\rm APO}}}{K}r_{\min}$;

\emph{Case C:} $C\in\left[r_{\max},\infty\right)$. In this case, based on a similar analysis as \emph{Case B} and \emph{Case C}, we can show that the LTE's payoff is also smaller than that under $C=\frac{K-1+{\eta^{\rm APO}}}{K}r_{\min}$.

Based on \emph{Case A}, \emph{Case B}, and \emph{Case C}, we prove {\bf Part II}.

Combing {\bf Part I} and {\bf Part II}, we complete the proof of the lemma.
\end{proof}

\subsection{Proof of Theorem \ref{theorem:optimalCcase}}
We discuss the optimal reserve rate in the following three cases:

\emph{Case A:} $R_{\rm LTE}\le\frac{K-1+{\eta^{\rm APO}}}{K\left(1-{\delta^{\rm LTE}}\right)}r_{\min}$. Based on Lemma \ref{lemma:local:etarmin}, in this case, the LTE provider achieves the highest payoff (\emph{i.e.}, ${\delta^{\rm LTE}} R_{\rm LTE}$) with $C\in\left[0,\frac{K-1+{\eta^{\rm APO}}}{K}r_{\min}\right]$. 
Because the LTE provider works in the competition mode and does not need to allocate rate to APOs, constraint $b_{\max}\left(C\right) \le R_{\rm LTE}$ in Problem (\ref{equ:optABC}) is automatically satisfied. Hence, the optimal $C^*$ can be any value from $\left[0,\frac{K-1+{\eta^{\rm APO}}}{K}r_{\min}\right]$;

\emph{Case B:} $\frac{K-1+{\eta^{\rm APO}}}{K\left(1-{\delta^{\rm LTE}}\right)}r_{\min}<R_{\rm LTE}\le r_{\max}$. Based on Lemma \ref{lemma:local:etarmin}, in this case, point $C=\frac{K-1+{\eta^{\rm APO}}}{K}r_{\min}$ is a local minimum, and the LTE provider will not consider any $C\in\left[0,\frac{K-1+{\eta^{\rm APO}}}{K}r_{\min}\right]$ as the optimal reserve rate. Therefore, the optimal $C^*$ should be chosen from set $\left(\frac{K-1+{\eta^{\rm APO}}}{K}r_{\min},\infty\right)$. Moreover, if $C^*$ is larger than $R_{\rm LTE}$, it is possible that an APO bids a value larger than $R_{\rm LTE}$, which implies the violation of constraint $b_{\max}\left(C\right) \le R_{\rm LTE}$. Hence, $C^*$ should be set no larger than $R_{\rm LTE}$. To conclude, $C^*$ can be chosen from $\left(\frac{K-1+{\eta^{\rm APO}}}{K}r_{\min},R_{\rm LTE}\right]$;

\emph{Case C:} $R_{\rm LTE}\!>\!\max\left\{r_{\max},\frac{K-1+{\eta^{\rm APO}}}{K\left(1-{\delta^{\rm LTE}}\right)}r_{\min}\right\}$. Similar as \emph{Case B}, the LTE provider will not consider any $C\in\left[0,\frac{K-1+{\eta^{\rm APO}}}{K}r_{\min}\right]$ as the optimal reserve rate. Furthermore, the LTE's payoff is a constant for $C\in\left[r_{\max},\infty\right)$. Hence, it is sufficient for the LTE to consider $C^*$ from set $\left(\frac{K-1+{\eta^{\rm APO}}}{K}r_{\min},r_{\max}\right]$. 
In this situation, because $R_{\rm LTE}\!>\!\max\left\{r_{\max},\frac{K-1+{\eta^{\rm APO}}}{K\left(1-{\delta^{\rm LTE}}\right)}r_{\min}\right\}$, the condition $C^*\le R_{\rm LTE}$ is automatically satisfied. 
Therefore, we conclude that $C^*$ can be chosen from $\left(\frac{K-1+{\eta^{\rm APO}}}{K}r_{\min},r_{\max}\right]$.

Based on \emph{Case A}, \emph{Case B}, and \emph{Case C}, we complete the proof of Theorem \ref{theorem:optimalCcase}.

{{
\subsection{Assumption \ref{assumption:unique} Under Uniform Distribution ($K=2$)}\label{appendix:sec:K2}
In this section, we prove that when $K=2$, Assumption \ref{assumption:unique} holds under the uniform distribution. 


{\bf Part A:} we prove that equation (\ref{equ:rT}) has a unique solution in $\left(C,r_{\max}\right)$. 
For the uniform distribution, the cumulative distribution function is $F\left(r\right)=\frac{r-r_{\min}}{r_{\max}-r_{\min}}$, where $r\in\left[r_{\min},r_{\max}\right]$. 
By using the expression of $F\left(r\right)$ and the fact that $K=2$, we can rewrite (\ref{equ:rT}) as
\begin{align}
\frac{r-C}{r_{\max}-r_{\min}}\frac{C-r}{2} +\frac{r_{\max}-r}{r_{\max}-r_{\min}}\left(C-\frac{1+{\eta^{\rm APO}}}{2}r\right)=0.
\end{align}
After rearrangement, we have
\begin{align}
\frac{\eta^{\rm APO}}{2}r^2-\frac{1+\eta^{\rm APO}}{2}r_{\max}r+r_{\max}C-\frac{C^2}{2}=0.\label{appendix:quadratic:a}
\end{align}
We define function $H\left(r\right)$ as the left hand side of (\ref{appendix:quadratic:a}):
\begin{align}
H\left(r\right)\triangleq \frac{\eta^{\rm APO}}{2}r^2-\frac{1+\eta^{\rm APO}}{2}r_{\max}r+r_{\max}C-\frac{C^2}{2}.
\end{align}
Since $H\left(r\right)$ is a quadratic function, we can show that its value decreases with $r$ for $r\in\left[r_{\min},\frac{1+\eta^{\rm APO}}{2\eta^{\rm APO}}r_{\max}\right)$. Because $C\in\left[r_{\min},r_{\max}\right)$ and $\frac{1+\eta^{\rm APO}}{2\eta^{\rm APO}}>1$, we have $r_{\min}\le C<r_{\max}<\frac{1+\eta^{\rm APO}}{2\eta^{\rm APO}}r_{\max}$. Therefore, $H\left(r\right)$ decreases with $r$ for $r\in\left[C,r_{\max}\right]$. 

When $r=C$ and $r=r_{\max}$, the values of function $H\left(r\right)$ are as follows:
\begin{align}
& H\left(C\right)=\frac{1-\eta^{\rm APO}}{2}\left(r_{\max}-C\right)C>0,\\
& H\left(r_{\max}\right)=-\frac{1}{2}\left(r_{\max}-C\right)^2<0.
\end{align}
Considering the fact that $H\left(r\right)$ is decreasing for $r\in\left[C,r_{\max}\right]$, we can conclude that there is a unique $r\in\left(C,r_{\max}\right)$ that satisfies equation $H\left(r\right)=0$. In other words, when $F\left(r\right)$ follows the uniform distribution and $K=2$, equation (\ref{equ:rT}) has a unique solution in $\left(C,r_{\max}\right)$.



{\bf Part B:} we prove that equation equation (\ref{equ:rX}) has a unique solution in $\left(r_{\min},r_{\max}\right)$. 
By using $F\left(r\right)=\frac{r-r_{\min}}{r_{\max}-r_{\min}}$ and the fact that $K=2$, we can rewrite (\ref{equ:rX}) as
\begin{align}
\frac{r-r_{\min}}{r_{\max}-r_{\min}}\frac{C-r}{2}+\frac{r_{\max}-r}{r_{\max}-r_{\min}}\left(C-\frac{1+\eta^{\rm APO}}{2}r\right)=0.
\end{align}
After rearrangement, we have
\begin{multline}
\frac{\eta^{\rm APO}}{2}r^2-\frac{C}{2}r+\frac{r_{\min}}{2}r-\frac{1+\eta^{\rm APO}}{2}r_{\max}r\\
+r_{\max}C-\frac{Cr_{\min}}{2}=0.\label{appendix:quadratic:b}
\end{multline}
We define function $J\left(r\right)$ as the left hand side of (\ref{appendix:quadratic:b}):
\begin{align}
\nonumber
J\left(r\right)\triangleq & \frac{\eta^{\rm APO}}{2}r^2-\frac{C}{2}r+\frac{r_{\min}}{2}r\\
& -\frac{1+\eta^{\rm APO}}{2}r_{\max}r+r_{\max}C-\frac{Cr_{\min}}{2}.
\end{align}
We can show that
\begin{align}
& J\left(r_{\min}\right)=\left({r_{\max}-r_{\min}}\right)\left(C-\frac{1+\eta^{\rm APO}}{2}r_{\min}\right)>0,\label{appendix:Jfunction:a}\\
& J\left(r_{\max}\right)=\frac{1}{2}\left(r_{\max}-r_{\min}\right)\left(C-r_{\max}\right)<0.\label{appendix:Jfunction:b}
\end{align}

Next we discuss two situations: $r_{\max}\le\frac{\left(1+\eta^{\rm APO}\right)r_{\max}+C-r_{\min}}{2\eta^{\rm APO}}$ and $r_{\max}>\frac{\left(1+\eta^{\rm APO}\right)r_{\max}+C-r_{\min}}{2\eta^{\rm APO}}$. 

When $r_{\max}\le\frac{\left(1+\eta^{\rm APO}\right)r_{\max}+C-r_{\min}}{2\eta^{\rm APO}}$, the quadratic function $J\left(r\right)$ is decreasing for $r\in\left[r_{\min},r_{\max}\right]$. Based on (\ref{appendix:Jfunction:a}) and (\ref{appendix:Jfunction:b}), we can conclude that there is a unique $r\in\left(r_{\min},r_{\max}\right)$ that satisfies equation $J\left(r\right)=0$. 

When $r_{\max}>\frac{\left(1+\eta^{\rm APO}\right)r_{\max}+C-r_{\min}}{2\eta^{\rm APO}}$, we can show that the quadratic function $J\left(r\right)$ is decreasing for $r\in\left[r_{\min},\frac{\left(1+\eta^{\rm APO}\right)r_{\max}+C-r_{\min}}{2\eta^{\rm APO}}\right)$ and increasing for $r\in\left(\frac{\left(1+\eta^{\rm APO}\right)r_{\max}+C-r_{\min}}{2\eta^{\rm APO}},r_{\max}\right]$. If $J\left(\frac{\left(1+\eta^{\rm APO}\right)r_{\max}+C-r_{\min}}{2\eta^{\rm APO}}\right)\ge0$, we have $J\left(r\right)\ge0$ for all $r\in\left[r_{\min},r_{\max}\right]$, which contradicts with (\ref{appendix:Jfunction:b}). Therefore, we have $J\left(\frac{\left(1+\eta^{\rm APO}\right)r_{\max}+C-r_{\min}}{2\eta^{\rm APO}}\right)<0$. Together with (\ref{appendix:Jfunction:a}) and the fact that $J\left(r\right)$ is decreasing for $r\in\left[r_{\min},\frac{\left(1+\eta^{\rm APO}\right)r_{\max}+C-r_{\min}}{2\eta^{\rm APO}}\right)$, we can show that there is a unique $r\in\left(r_{\min},\frac{\left(1+\eta^{\rm APO}\right)r_{\max}+C-r_{\min}}{2\eta^{\rm APO}}\right)$ that satisfies equation $J\left(r\right)=0$. 
Moreover, since $J\left(r\right)$ is increasing for $r\in\left(\frac{\left(1+\eta^{\rm APO}\right)r_{\max}+C-r_{\min}}{2\eta^{\rm APO}},r_{\max}\right]$, $J\left(\frac{\left(1+\eta^{\rm APO}\right)r_{\max}+C-r_{\min}}{2\eta^{\rm APO}}\right)<0$, and $J\left(r_{\max}\right)<0$, function $J\left(r\right)$ is always smaller than zero for $r\in\left[\frac{\left(1+\eta^{\rm APO}\right)r_{\max}+C-r_{\min}}{2\eta^{\rm APO}},r_{\max}\right)$. 
That is to say, there is a unique $r\in\left(r_{\min},r_{\max}\right)$ that satisfies equation $J\left(r\right)=0$.

From the results of the two situations, we can conclude that when $F\left(r\right)$ follows the uniform distribution and $K=2$, equation (\ref{equ:rX}) has a unique solution in $\left(r_{\min},r_{\max}\right)$. 

Combining {\bf Part A} and {\bf Part B}, we complete the proof.
}}

\subsection{Assumption \ref{assumption:unique} Under Uniform Distribution and Truncated Normal Distribution (General $K$)}\label{appendix:sec:Kgeneral}
In this section, we show that Assumption \ref{assumption:unique} holds under the uniform distribution and truncated normal distribution for a general $K$ through simulation. We show the following four cases:

\begin{figure*}[t]
  \centering
  \begin{minipage}[t]{.48\linewidth}
  \includegraphics[scale=0.3]{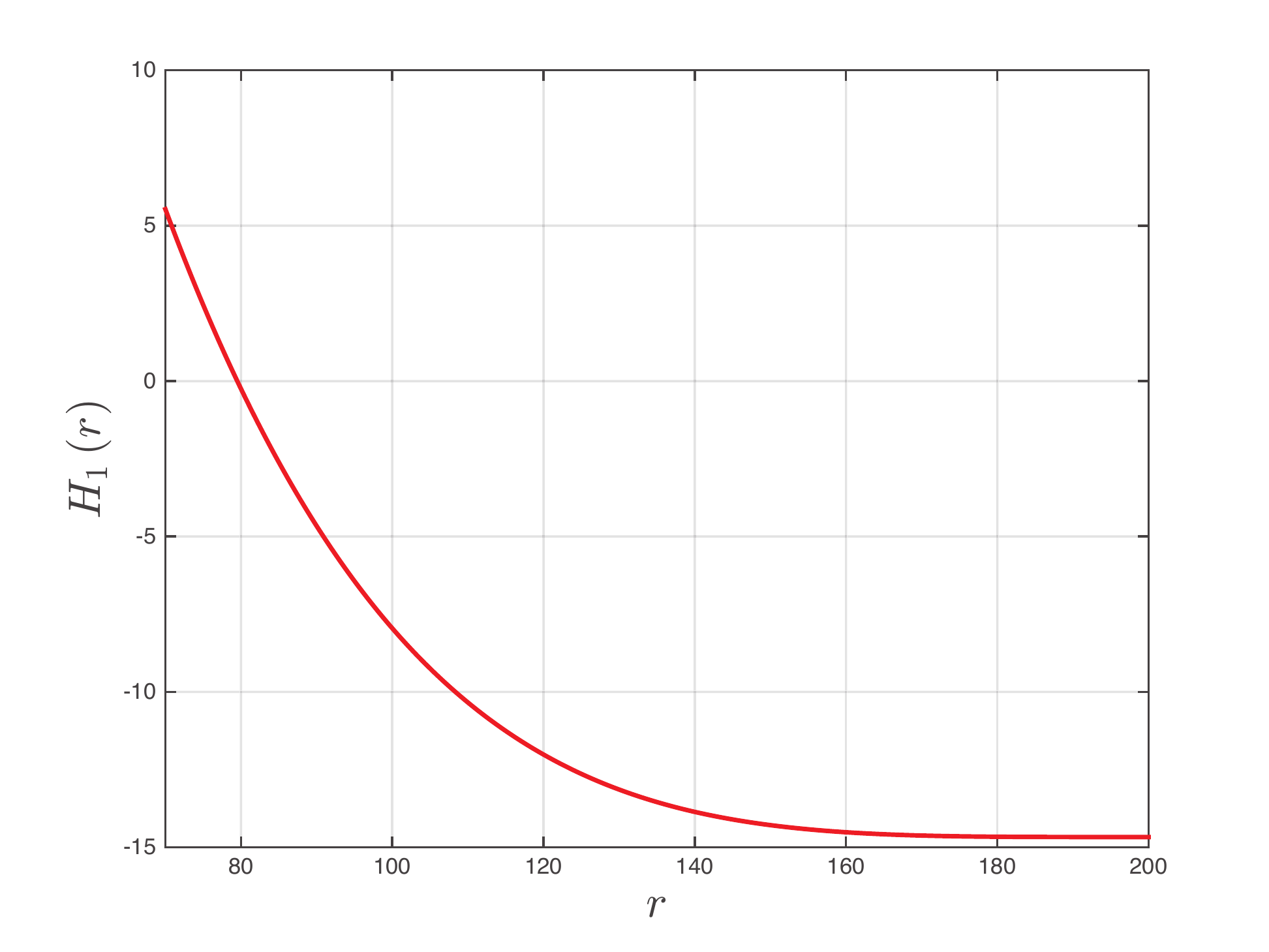}
  \caption{Uniqueness of $r_T\left(C\right)$ Under Uniform Distribution.}
  \label{fig:app:1}
  \vspace{-4mm}
  \end{minipage}
  \begin{minipage}[t]{.48\linewidth}
  \centering
  \includegraphics[scale=0.3]{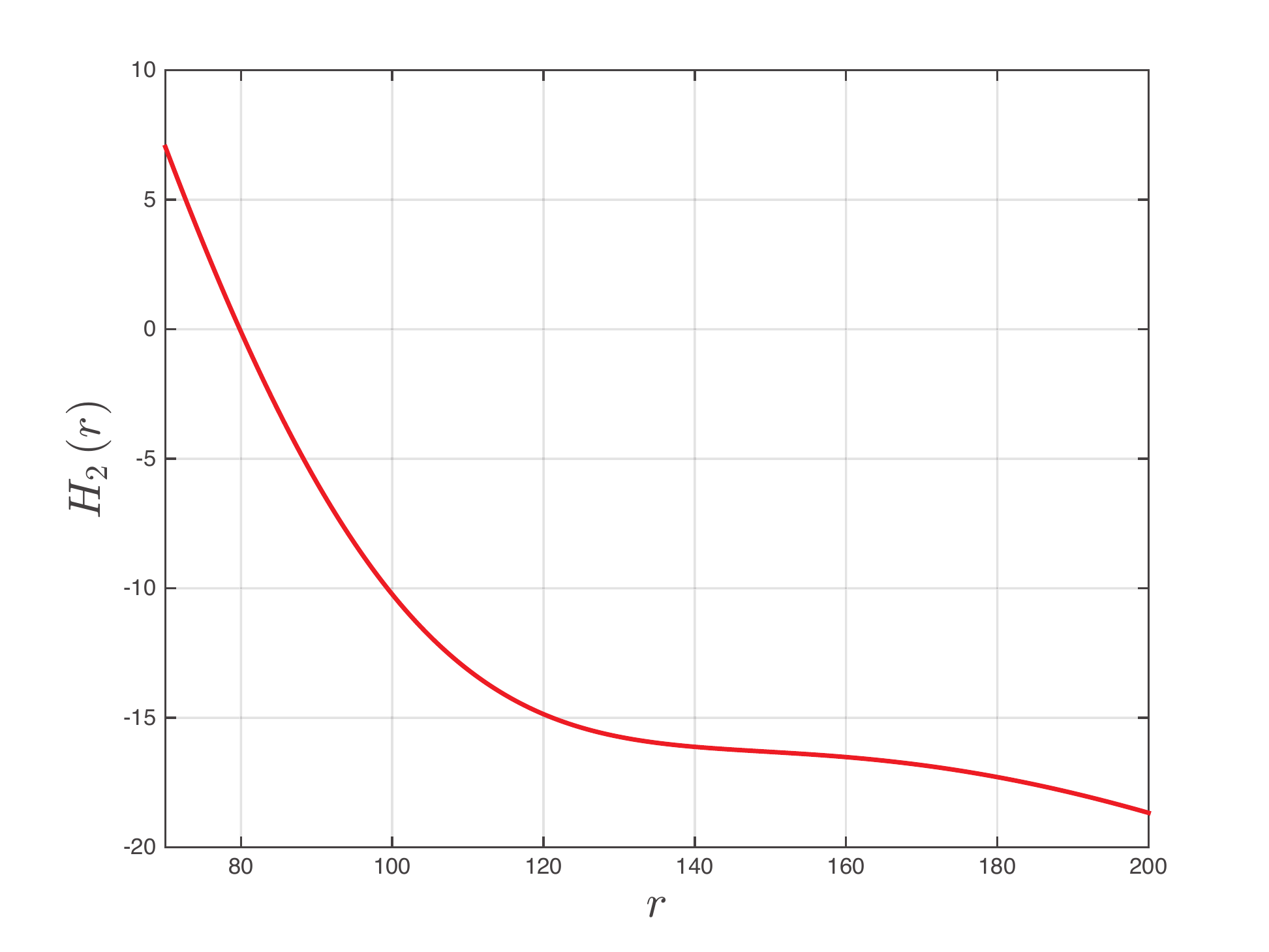}
  \caption{Uniqueness of $r_T\left(C\right)$ Under Truncated Normal Distribution.}
  \label{fig:app:2}
  \vspace{-4mm}
  \end{minipage}
\end{figure*}
\begin{figure*}[t]
  \centering
  \begin{minipage}[t]{.48\linewidth}
  \includegraphics[scale=0.3]{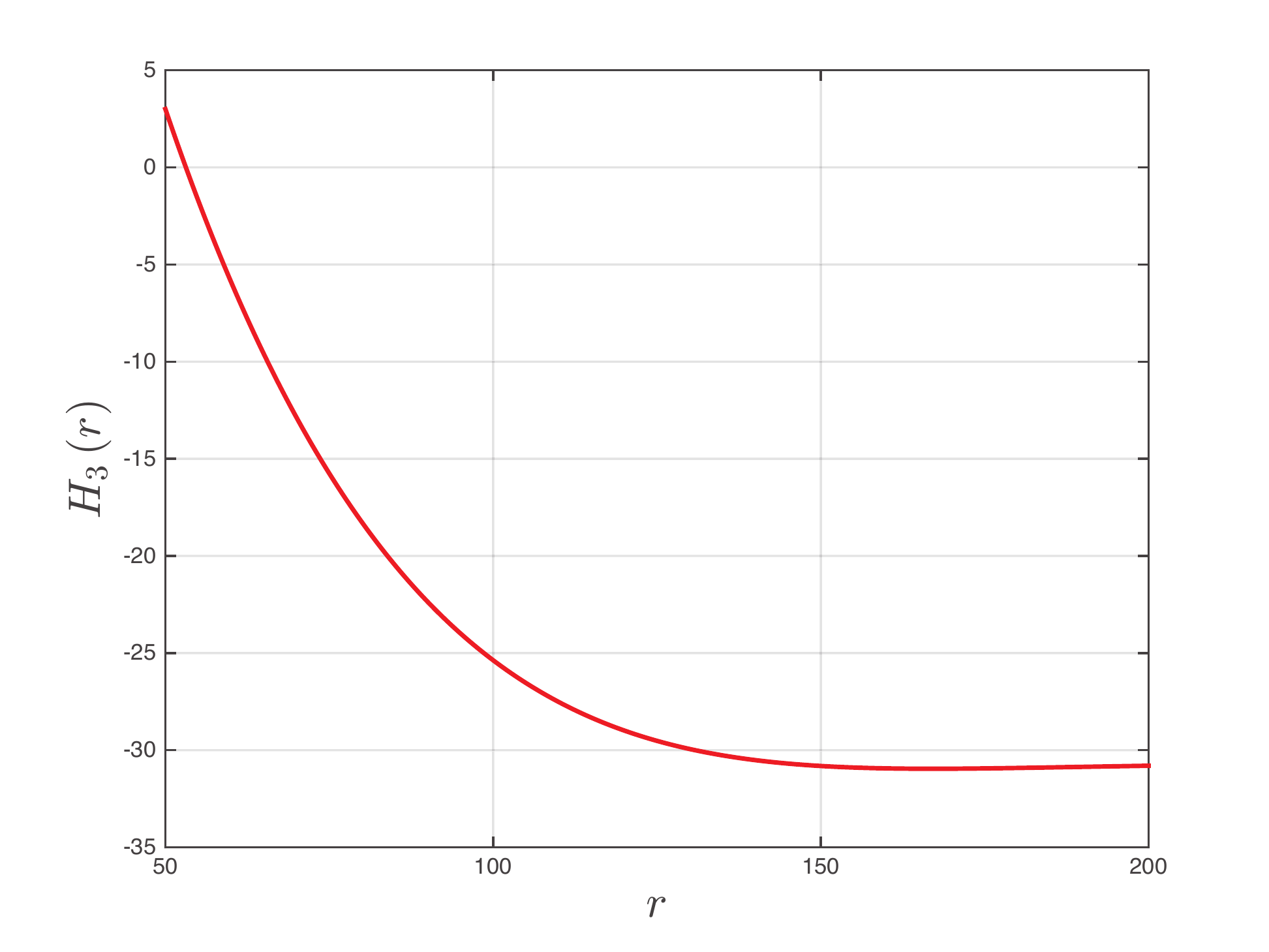}
  \caption{Uniqueness of $r_X\left(C\right)$ Under Uniform Distribution.}
  \label{fig:app:3}
  \vspace{-4mm}
  \end{minipage}
  \begin{minipage}[t]{.48\linewidth}
  \centering
  \includegraphics[scale=0.3]{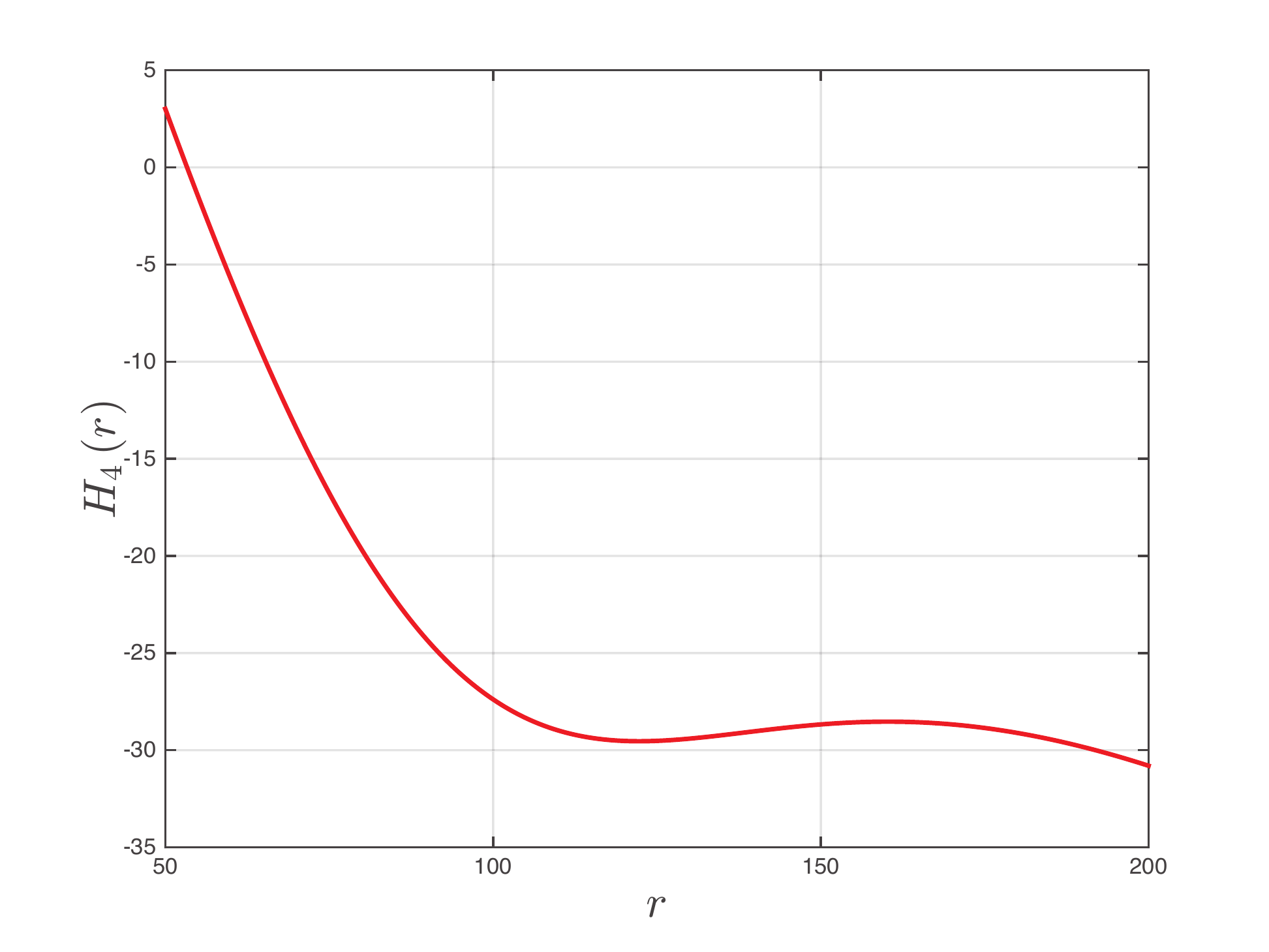}
  \caption{Uniqueness of $r_X\left(C\right)$ Under Truncated Normal Distribution.}
  \label{fig:app:4}
  \vspace{-4mm}
  \end{minipage}
\end{figure*}

\emph{Case A:} Uniqueness of $r_T\left(C\right)$ under the uniform distribution. We choose $K=5$, $\eta^{\rm APO}=0.3$, $C=70{\rm~Mbps}$, and $r_k\sim {\cal U}\left[50{\rm~Mbps},200{\rm~Mbps}\right]$ for all $k\in{\cal K}$. We denote the left-hand-side of equation (\ref{equ:rT}) as $H_1\left(r\right)$, and plot it against $r$ in Fig. \ref{fig:app:1}. We can find that there is a unique solution for $H_1\left(r\right)=0$. In other words, there is a unique $r_T\left(C\right)$ in this case;

\emph{Case B:} Uniqueness of $r_T\left(C\right)$ under the truncated normal distribution. We choose $K=5$, $\eta^{\rm APO}=0.3$, and $C=70{\rm~Mbps}$. We obtain the distribution of $r_k$, $k\in{\cal K}$, by truncating the normal distribution ${\cal N}\left(125~{\rm Mbps},2500~{{\rm Mbps}^2}\right)$ to interval $\left[50~{\rm Mbps},200~{\rm Mbps}\right]$. 
We denote the left-hand-side of equation (\ref{equ:rT}) as $H_2\left(r\right)$, and plot it against $r$ in Fig. \ref{fig:app:2}. We can find that there is a unique solution for $H_2\left(r\right)=0$. In other words, there is a unique $r_T\left(C\right)$ in this case;

\emph{Case C:} Uniqueness of $r_X\left(C\right)$ under the uniform distribution. We choose $K=5$, $\eta^{\rm APO}=0.3$, $C=46{\rm~Mbps}$, and $r_k\sim {\cal U}\left[50{\rm~Mbps},200{\rm~Mbps}\right]$ for all $k\in{\cal K}$. We denote the left-hand-side of equation (\ref{equ:rX}) as $H_3\left(r\right)$, and plot it against $r$ in Fig. \ref{fig:app:3}. We can find that there is a unique solution for $H_3\left(r\right)=0$. In other words, there is a unique $r_X\left(C\right)$ in this case;

\emph{Case D:} Uniqueness of $r_T\left(C\right)$ under the truncated normal distribution. We choose $K=5$, $\eta^{\rm APO}=0.3$, and $C=46{\rm~Mbps}$. We obtain the distribution of $r_k$, $k\in{\cal K}$, by truncating the normal distribution ${\cal N}\left(125~{\rm Mbps},2500~{{\rm Mbps}^2}\right)$ to interval $\left[50~{\rm Mbps},200~{\rm Mbps}\right]$. 
We denote the left-hand-side of equation (\ref{equ:rX}) as $H_4\left(r\right)$, and plot it against $r$ in Fig. \ref{fig:app:4}. We can find that there is a unique solution for $H_4\left(r\right)=0$. In other words, there is a unique $r_X\left(C\right)$ in this case.


\subsection{Computation of ${\bar r}_{\rm pay}\left(C\right)$}\label{appendix:sec:rpay}
We first compute the probability distribution of $b_{\min}^{-k}$. We denote the CDF of $b_{\min}^{-k}$ as $G\left(\cdot\right)$ and compute it as
\begin{align}
G\left(r\right)=1-\left(1-F\left(r\right)\right)^{K-1},r\in\left[r_{\min},r_{\max}\right].\label{equ:Gr}
\end{align}
Hence, we denote the PDF of $b_{\min}^{-k}$ as $g\left(\cdot\right)$ and compute it as
\begin{align}
g\left(r\right)=\frac{d G\left(r\right)}{dr}=\left(K-1\right)f\left(r\right)\left(1-F\left(r\right)\right)^{K-2},r\in\left[r_{\min},r_{\max}\right].\label{equ:gr}
\end{align}

We focus on the expected payment received by APO $k$. APO $k$ wins the auction under the following three cases:
\begin{itemize}
\item $r_k\in\left[r_{\min},C\right)$ and $b_{\min}^{-k}\in\left[r_k,C\right)$. In this case, APO $k$ receives $b_{\min}^{-k}$ from the LTE;
\item $r_k\in\left[r_{\min},C\right)$ and $b_{\min}^{-k}=C$ or $``{\rm N} \textquotedblright$. In this case, APO $k$ receives $C$ from the LTE;
\item $r_k\in\left[C,r_T\left(C\right)\right]$ and $b_{\min}^{-k}=C$ or $``{\rm N} \textquotedblright$. In this case, the expected payment that APO $k$ receives depends on the number of APOs bidding $C$.
\end{itemize}
Based on this discussion, we can compute the expected payment that APO $k$ receives as (\ref{equ:app:long}). Furthermore, it is easy to find that relation (\ref{equ:help}) holds. Based on $G\left(r\right)$ in (\ref{equ:Gr}), $g\left(r\right)$ in (\ref{equ:gr}), and (\ref{equ:help}), we rewrite (\ref{equ:app:long}) as (\ref{equ:APOk:rpay}). 
Notice that (\ref{equ:APOk:rpay}) shows the expected payment that APO $k$ receives from the LTE. Since we consider $K$ APOs, we can compute ${\bar r}_{\rm pay}\left(C\right)$ as
\begin{figure*}
\begin{multline}
\int_{r_{\min}}^{C} r g\left(r\right) F\left(r\right) dr +CF\left(C\right)\left(1-G\left(C\right)\right) +\\
\left(F\left(r_T\left(C\right)\right)-F\left(C\right)\right)\sum_{n=0}^{K-1}{\binom{K-1}{n} \left(F\left(r_T\right)-F\left(C\right)\right)^n\left(1-F\left(r_T\right)\right)^{K-1-n}\frac{C}{n+1}}.\label{equ:app:long}
\end{multline}
\hrulefill
\begin{multline}
\frac{1}{K}C\left(\left(1-F\left(C\right)\right)^K-\left(1-F\left(r_T\left(C\right)\right)\right)^K\right)=\\
\left(F\left(r_T\left(C\right)\right)-F\left(C\right)\right)\sum_{n=0}^{K-1}{\binom{K-1}{n} \left(F\left(r_T\right)-F\left(C\right)\right)^n\left(1-F\left(r_T\right)\right)^{K-1-n}\frac{C}{n+1}}.\label{equ:help}
\end{multline}
\hrulefill
\begin{multline}
\left(K-1\right) \int_{r_{\min}}^{C} r f\left(r\right) F\left(r\right)  \left(1-F\left(r\right)\right)^{K-2} dr +CF\left(C\right)\left(1-F\left(C\right)\right)^{K-1} \\
+\frac{1}{K}C\left(\left(1-F\left(C\right)\right)^K-\left(1-F\left(r_T\left(C\right)\right)\right)^K\right).\label{equ:APOk:rpay}
\end{multline}
\hrulefill
\end{figure*}
\begin{align}
\nonumber
{\bar r}_{\rm pay}\left(C\right)\triangleq & K\left(K-1\right) \int_{r_{\min}}^{C} r f\left(r\right) F\left(r\right)  \left(1-F\left(r\right)\right)^{K-2} dr \\
\nonumber
&+KCF\left(C\right)\left(1-F\left(C\right)\right)^{K-1} \\
&+C\left(\left(1-F\left(C\right)\right)^K-\left(1-F\left(r_T\left(C\right)\right)\right)^K\right).
\end{align}



\subsection{Example About Inefficiency in Auction}\label{appendix:sec:example}
We present an example to show that even if the cooperation mutually benefits the LTE provider and the APOs, these two types of networks may not reach an agreement on the cooperation. 

In the example, $R_{\rm LTE}=95{\rm~Mbps}$, $K=4$, $\delta^{\rm LTE}=0.4$, $\eta^{\rm APO}=0.3$, and ${r_k}=64{\rm~Mbps}$ for all $k\in{\cal K}$. Notice that APO $k$'s type $r_k$ is not known by the LTE and APOs, while the distribution of $r_k$ is the common knowledge. We assume that $r_k$ follows a truncated normal distribution, and we obtain the distribution of $r_k$ by truncating the normal distribution ${\cal N}\left(125~{\rm Mbps},2500~{{\rm Mbps}^2}\right)$ to interval $\left[50~{\rm Mbps},200~{\rm Mbps}\right]$.

For this example, based on the numerical method mentioned in Section \ref{subsec:stageI:3}, we can compute the LTE's optimal reserve rate as $C^*=49.4{\rm~Mbps}$. Furthermore, we can compute $r_X\left(C^*\right)=59.3~{\rm Mbps}$, and conclude that the APOs with types in $\left[50~{\rm Mbps},59.3~{\rm Mbps}\right]$ bid $C^*$, and the APOs with types in $\left(59.3~{\rm Mbps},200~{\rm Mbps}\right]$ bid $``{\rm N} \textquotedblright$ (based on the analysis in Section{\ref{subsec:stageII:3}}). Since in our example, ${r_k}=64{\rm~Mbps}$ for all $k\in{\cal K}$, all of these four APOs bid $``{\rm N} \textquotedblright$. In other words, the LTE cannot find any APO to cooperate with, and it will eventually work in the competition mode. In this situation, the LTE's payoff is $\delta R_{\rm LTE}=38~{\rm Mbps}$ and APO $k$'s ($k\in{\cal K}$) expected payoff is $\frac{K-1+\eta^{\rm APO}}{K}r_k=52.8~{\rm Mbps}$.

Next we consider the situation where the LTE changes its reserve rate from $C^*=49.4{\rm~Mbps}$ to ${\bar C}=55{\rm~Mbps}$. When ${\bar C}=55{\rm~Mbps}$, we can find $r_T\left({\bar C}\right)=65.8{\rm~Mbps}$ based on the analysis in Section{\ref{subsec:stageII:1}. That is to say, the APOs with types in $\left[50~{\rm Mbps},55~{\rm Mbps}\right]$ bid their types, the APOs with types in $\left(55~{\rm Mbps},65.8{\rm~Mbps}\right]$ bid $\bar C$, and the APOs with types in $\left(65.8~{\rm Mbps},200~{\rm Mbps}\right]$ bid $``{\rm N} \textquotedblright$ (based on the analysis in Section{\ref{subsec:stageII:1}}). 
In our example, ${r_k}=64{\rm~Mbps}$ for all $k\in{\cal K}$. 
Hence, when the LTE chooses its reserve rate as ${\bar C}=55{\rm~Mbps}$, all of these four APOs bid $\bar C$. As a result, the LTE will randomly pick one APO to cooperate with. In this situation, the LTE's payoff is $R_{\rm LTE}-{\bar C}=40{\rm~Mbps}$ and APO $k$'s ($k\in{\cal K}$) expected payoff is $\frac{1}{K}{\bar C}+\frac{K-1}{K} r_k=61.75~{\rm Mbps}$. We find that when the LTE changes its reserve rate from $C^*=49.4{\rm~Mbps}$ to ${\bar C}=55{\rm~Mbps}$, the LTE switches from the competition mode to the cooperation mode, and both the LTE and APOs obtain higher payoffs.

To summarize, under reserve rate $55{\rm~Mbps}$, the LTE and APOs can cooperate with each other. However, due to the incomplete information in the auction, the LTE determines its reserve rate by only considering the probability distribution of APOs' types ${\bm r}$. As a result, the LTE will choose $C^*=49.4{\rm~Mbps}$ as its reserve rate, which makes the LTE lose a cooperation chance that mutually benefits both the LTE and APOs.



\section{Supplementary Material: Multi-LTE Analysis}\label{sec:supplementary}

We extend our model and analysis to a more general scenario, where there are multiple LTE providers in the system. Based on our discussion in Section \ref{subsec:discussion:multiLTE} of the paper, the LTE providers take turns to organize the auctions. Without loss of generality, we consider the case where a particular LTE provider organizes an auction. 

Note that we focus on the cooperation between the LTE and APOs, and do not consider the cooperation among the LTE providers. LTE-U Forum evaluated the performance of two LTE networks (owned by different providers) when they share the same unlicensed channel, and the results showed that the interference between the LTE networks is much lower than that between the LTE and Wi-Fi APOs \cite{forum}.


\subsection{Basic Settings}

We consider a general system model illustrated in Fig. \ref{fig:system:multiLTE}, where LTE provider $0$ is the auctioneer and negotiates with all the APOs. We define ${\cal K}_S\triangleq \left\{1,2,\ldots,K_S\right\}$ ($K_S\ge 2$) and ${\cal K}_A\triangleq \left\{K_S+1,K_S+2,\ldots,K_S+K_A\right\}$ ($K_A\ge 2$). Specifically, APO $k\in{\cal K}_S$ \emph{shares} channel $k$ with LTE provider $k$, and APO $k\in{\cal K}_A$ occupies channel $k$ \emph{alone}. 
Hence, there are a total of $K_S+1$ LTE providers in the system: LTE provider $0$ organizes the auction, and the remaining $K_S$ LTE providers share the channels with the APOs from set ${\cal K}_S$.{\footnote{{Here, we assume that each APO from set ${\cal K}_S$ shares channel $k$ with only one LTE provider. This assumption helps us simplify the presentation of the analysis (\emph{e.g.}, reduce the number of notations). The technique that we develop in this section can be easily applied to the scenario where there are multiple LTE providers coexisting with an APO from set ${\cal K}_S$.}}}

\begin{figure}[h]
  \centering
  \includegraphics[scale=0.4]{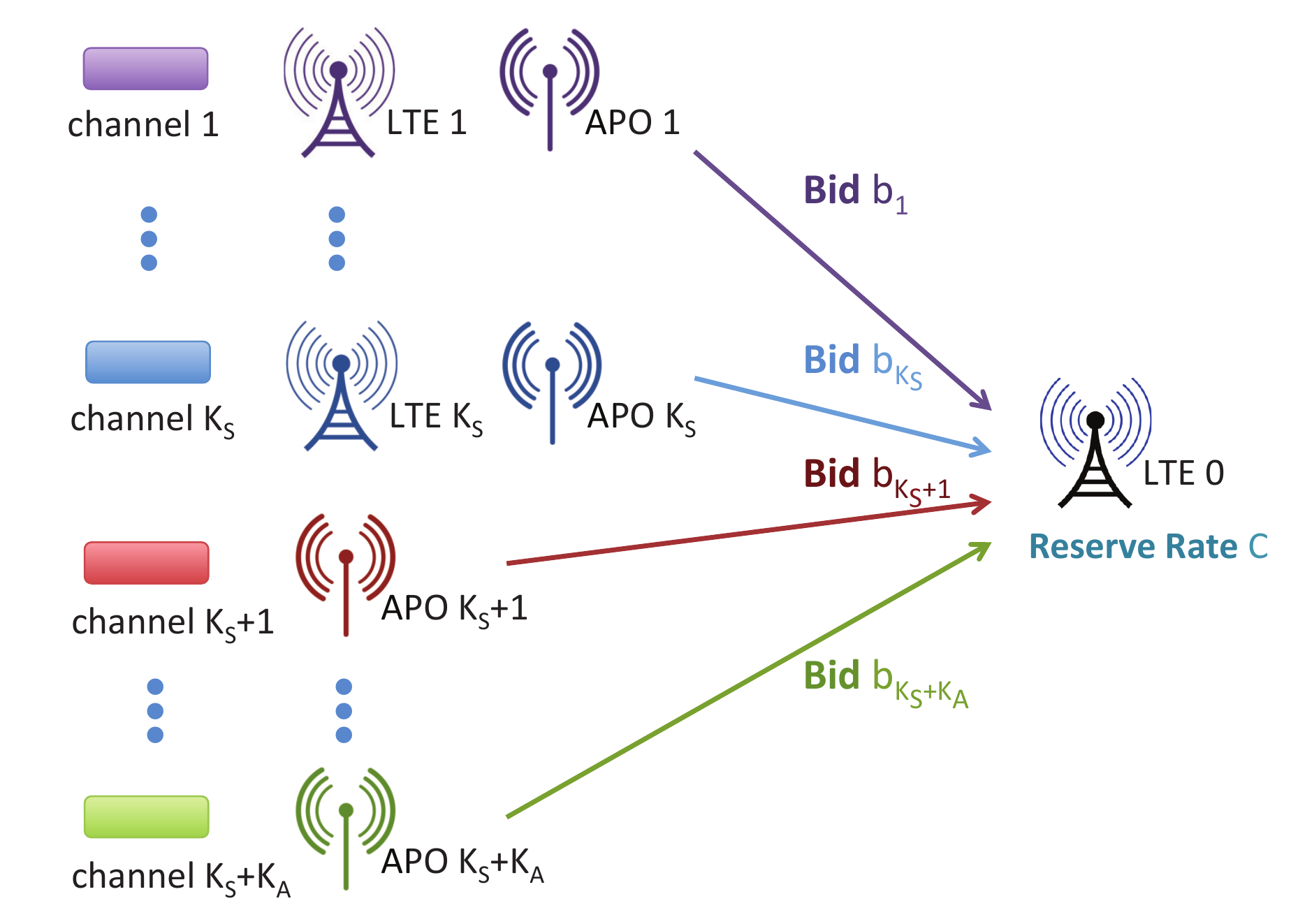}\\
  \caption{Generalized System Model.}
  \label{fig:system:multiLTE}
  \vspace{-0.2cm}
\end{figure}  

{\bf APOs' Rates:} {{We use $r_k$ to denote the throughput that APO $k\in{\cal K}_S\cup{\cal K}_A$ can achieve to serve its users when it \emph{exclusively} occupies a channel (without the interference from the LTE networks).}} 
{{The value of $r_k$ is the private information of APO $k$. The LTE providers and the other APOs only know the probability distribution of $r_k$. Specifically, we assume that $r_k$ is a continuous random variable drawn from interval $\left[r_{\min},r_{\max}\right]$ ($r_{\min},r_{\max}\ge0$), and follows a probability distribution function (PDF) $f\left(\cdot\right)$ and a cumulative distribution function (CDF) $F\left(\cdot\right)$. Moreover, we assume that $f\left(\cdot\right)>0$ for all $r\in\left[r_{\min},r_{\max}\right]$.}} 

{\bf LTE provider $0$'s Dual Modes:} We assume that LTE provider $0$ achieves a channel independent throughput of $R_{\rm LTE}>0$ when it \emph{exclusively} occupies a channel (without the interference from the other LTE providers and APOs). Note that we do not need to know other LTE providers' throughputs, as we focus on the interactions between LTE provider $0$ and APOs in the auction. 
LTE provider $0$ can operate its network in one of the following modes:

-\emph{Competition mode:} LTE provider $0$ randomly chooses channel $k\in{\cal K}_A$ with an equal probability and coexists with APO $k$ in that channel. 
The co-channel interference decreases both the data rates of LTE provider $0$ and APO $k$. We use ${\delta^{\rm LTE}}\in\left(0,1\right)$ and ${\eta^{\rm APO}}\in\left(0,1\right)$ to denote LTE provider $0$'s and APO $k$'s data rate discounting factors, respectively. 
Notice that in the competition mode, LTE provider $0$ will only access a channel from set ${\cal K}_A$, and will not access a channel from set ${\cal K}_S$. 
This is because any channel from set ${\cal K}_S$ is shared by an APO and an LTE provider. Compared with accessing a channel from set ${\cal K}_A$, accessing a channel from set ${\cal K}_S$ incurs more interference to LTE provider $0$. 

-\emph{Cooperation mode:} LTE provider $0$ reaches an agreement with APO $k\in{\cal K}_S\cup{\cal K}_A$, accesses channel $k$, and allocates $r_{\rm pay}$ to the onloaded APO $k$'s traffic. Specifically, if LTE provider $0$ cooperates with APO $k\in{\cal K}_A$, LTE provider $0$ will fully occupy channel $k$. However, if LTE provider $0$ cooperates with APO $k\in{\cal K}_S$, LTE provider $0$ will share channel $k$ with LTE provider $k$. In this situation, the co-channel interference decreases the data rates of both LTE providers, and we use $\theta^{\rm LTE}\in\left(0,1\right)$ to denote the LTE providers' data rate discounting factor.

\subsection{Second-Price Reverse Auction Design}\label{revision:auctionframe}
We design a second-price reverse auction, where LTE provider $0$ is the auctioneer (buyer) and the $K_S+K_A$ APOs are the bidders (sellers of onloading opportunities). 

{\bf Reserve Rate and Bids:} {{In Stage I of the auction,}} LTE provider $0$ announces its reserve rate $C\in\left[0,\infty\right)$. {{In Stage II of the auction,}} after observing the reserve rate $C$, APO $k\in{\cal K}_S\cup{\cal K}_A$ submits a bid $b_k$. We define the vector of all APOs' bids as ${\bm b}\triangleq \left(b_k,\forall k\in{\cal K}_S\cup{\cal K}_A \right)$. 
Specifically, APO $k\in{\cal K}_S$ can submit a bid $b_k\in\left[0,C-\left(1-\theta^{\rm LTE}\right)R_{\rm LTE}\right]\cup\left\{``\rm N \textquotedblright\right\}$, and APO $k\in{\cal K}_A$ can submit a bid $b_k\in\left[0,C\right]\cup\left\{``\rm N \textquotedblright\right\}$. 

Note that APO $k\in{\cal K}_S$ can only submit $C-\left(1-\theta^{\rm LTE}\right)R_{\rm LTE}$ as its highest bid except $``\rm N \textquotedblright$. This is because the benefit for LTE provider $0$ of accessing channel $k\in{\cal K}_S$ is lower than the benefit of accessing channel a channel from set ${\cal K}_A$, due to the interference from LTE provider $k$ in channel $k$. 
When LTE provider $0$ onloads the users of APO $k\in{\cal K}_S$ and accesses channel $k$, LTE provider $0$ can achieve a data rate of $\theta^{\rm LTE}R_{\rm LTE}$ (due to the interference from LTE provider $k$). When LTE provider $0$ onloads the users of an APO from set ${\cal K}_A$ and accesses the corresponding channel, LTE provider $0$ can achieve a data rate of $R_{\rm LTE}$. Hence, the difference of LTE provider $0$'s benefits of accessing channel $k\in{\cal K}_S$ and a channel from set ${\cal K}_A$ is $\left(1-\theta^{\rm LTE}\right)R_{\rm LTE}$. As a result, LTE provider $0$ will not be willing to cooperate with APO $k\in{\cal K}_S$ if APO $k$'s bid is greater than $C-\left(1-\theta^{\rm LTE}\right)R_{\rm LTE}$. 

Since the benefits for LTE provider $0$ of accessing a channel from set ${\cal K}_S$ and set ${\cal K}_A$ are different, LTE provider $0$ needs to normalize the bid vector $\bm b$ to a virtual bid vector $\tilde{\bm b}\triangleq \left(\tilde b_k,\forall k\in{\cal K}_S\cup{\cal K}_A \right)$ to fairly compare all APOs' requests. 
Specifically, for $k\in{\cal K}_A$, we define $\tilde b_k \triangleq b_k$; for $k\in{\cal K}_S$, we define $\tilde b_k$ as 
\begin{align}
\nonumber
&\tilde b_k \triangleq \\
& \left\{\begin{array}{ll}
{\!b_k\!+\!\left(1\!-\!\theta^{\rm LTE}\right)\!R_{\rm LTE},} &  {{\rm if~} b_k\in \left[0,C\!-\!\left(1\!-\!\theta^{\rm LTE}\right)\!R_{\rm LTE}\right]\!,}\\
{``\rm N \textquotedblright,} & {{\rm if~}b_k=``\rm N \textquotedblright.}
\end{array} \right.
\end{align}
Hence, we have $\tilde b_k\in \left[\left(1-\theta^{\rm LTE}\right)R_{\rm LTE},C\right]\cup\left\{``\rm N \textquotedblright\right\}$ for $k\in{\cal K}_S$. LTE provider $0$ utilizes the virtual bid vector $\tilde{\bm b}$ to determine the auction outcome. 


{\bf Auction Outcome:} Next we discuss the auction outcome based on the different values of $\tilde{\bm b}$ and $C$. For ease of exposition, we define the comparison between $``\rm N \textquotedblright$ and any virtual bid $\tilde b_k$, $k\in{\cal K}_S\cup{\cal K}_A$, as
\begin{align}
\min\left\{{``\rm N \textquotedblright},\tilde b_k\right\}=\left\{\begin{array}{ll}
{\tilde b_k,} & {\rm if~}{\tilde b_k\in\left[0,C\right],}\\
{{``\rm N \textquotedblright},} & {{\rm if~}{\tilde b_k={``\rm N \textquotedblright}}.}
\end{array} \right.
\end{align} 
Furthermore, we define ${\cal I}_{\min}$ as the set of APOs with the minimum virtual bid:
\begin{align}
{\cal I}_{\min}\triangleq \left\{i\in{\cal K}_S\cup{\cal K}_A:i=\arg\min_{k\in{\cal K}_S\cup{\cal K}_A}{\tilde b_k} \right\}.
\end{align}

The auction has the following possible outcomes:

(a) When $\left|{\cal I}_{\min}\right|=1$,{\footnote{{Note that condition $\left|{\cal I}_{\min}\right|=1$ implies that $\min_{k\in{\cal K}_S\cup{\cal K}_A}{\tilde b_k} \in \left[0,C\right]$, as we have $K_S+K_A>1$ APOs.}}} then APO $i= \arg\min_{k\in{\cal K}_S\cup{\cal K}_A}{\tilde b_k}$ is the winner. LTE provider $0$ works in the \emph{cooperation mode} and accesses channel $i$. If $i\in{\cal K}_S$, LTE provider $0$ serves APO $i$'s users with a rate $r_{\rm pay}=\min\left\{C,\min_{k\ne i,k\in{\cal K}_S\cup{\cal K}_A}{\tilde b}_k\right\}-\left(1-\theta^{\rm LTE}\right)R_{\rm LTE}$; if $i\in{\cal K}_A$, LTE provider $0$ serves APO $i$'s users with a rate $r_{\rm pay}=\min\left\{C,\min_{k\ne i,k\in{\cal K}_S\cup{\cal K}_A}{\tilde b}_k\right\}$;

(b) When $\min_{k\in{\cal K}_S\cup{\cal K}_A}{\tilde b_k} \in \left[0,C\right]$ and $\left|{\cal I}_{\min}\right|>1$, LTE provider $0$ works in the \emph{cooperation mode}, randomly chooses an APO from set ${\cal I}_{\min}$ with the probability $\textstyle \frac{1}{\left|{\cal I}_{\min}\right|}$, and accesses that corresponding channel. If the chosen APO is from set ${\cal K}_S$, LTE provider $0$ serves the chosen APO's users with a rate $r_{\rm pay}=\min_{k\in{\cal K}_S\cup{\cal K}_A}{\tilde b_k}-\left(1-\theta^{\rm LTE}\right)R_{\rm LTE}$; if the chosen APO is from set ${\cal K}_A$, LTE provider $0$ serves the chosen APO's users with a rate $r_{\rm pay}=\min_{k\in{\cal K}_S\cup{\cal K}_A}{\tilde b_k}$;

(c) When $\min_{k\in{\cal K}_S\cup{\cal K}_A}{\tilde b_k} =``\rm N \textquotedblright$, LTE provider $0$ works in the \emph{competition mode}, randomly chooses channel $k\in{\cal K}_A$ with the probability $\textstyle\frac{1}{K_A}$, and shares the channel with APO $k$.

\subsection{LTE provider $0$'s Payoff}
We define LTE provider $0$'s payoff as the data rate that it can allocate to its own users. Based on the summary of auction outcome in the last subsection, we can compute LTE provider $0$'s payoff as a function of the virtual bid vector $\tilde {\bm b}$ and $C$ in (\ref{revision:equ:LTEpayoff}). 

Note that the expression of $\Pi^{\rm LTE}\left(\tilde {\bm b},C\right)$ in (\ref{revision:equ:LTEpayoff}) captures both the situation that an APO from set ${\cal K}_S$ becomes the winner and the situation that an APO from set ${\cal K}_A$ becomes the winner. For example, when $\left|{\cal I}_{\min}\right|=1$ and APO $k\in{\cal K}_S$ becomes the winner, LTE provider $0$'s payoff is computed as
\begin{align}
\nonumber
& \Pi^{\rm LTE}\left(\tilde {\bm b},C\right)\\
\nonumber
& =\theta^{\rm LTE} R_{\rm LTE}\!-\!\left(\min\!\left\{\!C,\min_{k\ne i,k\in{\cal K}_S\cup{\cal K}_A}{\tilde b}_k\!\right\}\!-\!\left(1\!-\!\theta^{\rm LTE}\right)\!R_{\rm LTE}\!\right)\\
& =R_{\rm LTE}-\min\left\{C,\min_{k\ne i,k\in{\cal K}_S\cup{\cal K}_A}{\tilde b}_k\right\}.\label{revision:equ:consistent}
\end{align}
From this example, we show that equation (\ref{revision:equ:LTEpayoff}) has already captured the situation that an APO from set ${\cal K}_S$ becomes the winner. 

\begin{figure*}
\begin{align}
\Pi^{\rm LTE}\left(\tilde {\bm b},C\right)= \left\{\begin{array}{ll}
{R_{\rm LTE}-\min\left\{C,\min_{k\ne i,k\in{\cal K}_S\cup{\cal K}_A}{\tilde b}_k\right\},} & {\rm if~}{\left|{\cal I}_{\min}\right|=1,}\\
{R_{\rm LTE}-\min_{k\in{\cal K}_S\cup{\cal K}_A}{\tilde b_k},} & {\rm if~}{\min_{k\in{\cal K}_S\cup{\cal K}_A}{\tilde b_k} \in \left[0,C\right]{~\rm and~}\left|{\cal I}_{\min}\right|>1,}\\
{{\delta^{\rm LTE}} R_{\rm LTE},} & {\rm if~}{\min_{k\in{\cal K}_S\cup{\cal K}_A}{\tilde b_k} =``\rm N \textquotedblright.}
\end{array} \right.\label{revision:equ:LTEpayoff}
\end{align}
\hrulefill
\begin{align}
\Pi^{\rm APO}_{k}\left(\left({\tilde b}_{k},{\bm {\tilde b}}_{S,-k}\right),{\bm {\tilde b}}_{A},C\right)= \left\{\begin{array}{ll}
{\frac{1}{\left|{\cal I}_{\min}\right|}\!\left(\min\left\{C,\min_{n\ne k,n\in{\cal K}_S\cup{\cal K}_A}{\tilde b}_n\right\}-\!\left(1\!-\!\theta^{\rm LTE}\right)\!R_{\rm LTE}\!\right)} & {\rm if~}{\tilde b_k\!= \min_{n\in{\cal K}_S\cup{\cal K}_A}{\tilde b_n}\!\in\!\left[0,C\right],}\\
{+\frac{\left|{\cal I}_{\min}\right|-1}{\left|{\cal I}_{\min}\right|}{\eta^{\rm APO}r_k},} & {}\\
{\eta^{\rm APO}r_k,} & {\rm otherwise.}
\end{array} \right.\label{revision:equ:APO0payoff}
\end{align}
\hrulefill
\begin{align}
\Pi^{\rm APO}_{k}\left({\tilde {\bm b}}_S,\left({\tilde b}_k,{\tilde {\bm b}}_{A,-k}\right),C\right)= \left\{\begin{array}{ll}
{r_k,} & {\rm if~}{\tilde b_k>\min_{n\in{\cal K}_S\cup{\cal K}_A}{\tilde b_n},}\\
{\frac{1}{\left|{\cal I}_{\min}\right|}{\min\left\{C,\min_{n\ne k,n\in{\cal K}_S\cup{\cal K}_A}{\tilde b}_n\right\}}+\frac{\left|{\cal I}_{\min}\right|-1}{\left|{\cal I}_{\min}\right|}r_k,} & {\rm if~}{\tilde b_k\!=\! \min_{n\in{\cal K}_S\cup{\cal K}_A}{\tilde b_n}\!\in\left[0,C\right],}\\
{\frac{K_A-1+{\eta^{\rm APO}}}{K_A}r_k,} & {\rm if~}{ \min_{n\in{\cal K}_S\cup{\cal K}_A}{\tilde b_n}=``{\rm N} \textquotedblright.}
\end{array} \right.\label{revision:equ:APOkpayoff}
\end{align}
\hrulefill
\end{figure*}


\subsection{APOs' Payoffs}
We define the payoff of APO $k\in{\cal K}_S\cup{\cal K}_A$ as the data rate that its users receive: when APO $k$ cooperates with LTE provider $0$, these users are served by LTE provider $0$; otherwise, they are served by APO $k$. 

Based on the auction design, we summarize the expected payoff of APO $k\in{\cal K}_S$ as (\ref{revision:equ:APO0payoff}). Variable ${\tilde b}_k$ is APO $k$'s virtual bid. We use ${\bm {\tilde b}}_{S,-k}\triangleq \left(\tilde b_n,n\ne k,n\in{\cal K}_S\right)$ to represent the virtual bids of the remaining APOs in set ${\cal K}_S$ (other than APO $k$), and use ${\bm {\tilde b}}_{A}\triangleq \left(\tilde b_n,n\in{\cal K}_A\right)$ to represent the virtual bids of the APOs in set ${\cal K}_A$. 

Equation (\ref{revision:equ:APO0payoff}) summarizes two possible situations: (a) when $\tilde b_k\!=\! \min_{n\in{\cal K}_S\cup{\cal K}_A}{\tilde b_n}\!\in\left[0,C\right]$, LTE provider $0$ cooperates with APO $k$ and one of the other APOs with the minimum virtual bid with the probability ${\frac{1}{\left|{\cal I}_{\min}\right|}}$ and the probability $1-{\frac{1}{\left|{\cal I}_{\min}\right|}}$ ($1\le {\left|{\cal I}_{\min}\right|}\le K_S+K_A$), respectively. In the former case, APO $k$'s users are onloaded to LTE provider $0$ and receive rate $\min\left\{C,\min_{n\ne k,n\in{\cal K}_S\cup{\cal K}_A}{\tilde b}_n\right\}-\!\left(1\!-\!\theta^{\rm LTE}\right)\!R_{\rm LTE}$. In the latter case, APO $k$'s users are served by APO $k$ and receive rate $\eta^{\rm APO}r_k$; (b) otherwise, LTE provider $0$ does not access channel $k$, and APO $k$'s users are served by APO $k$ and receive rate $\eta^{\rm APO}r_k$.

Based on the auction design, we summarize the expected payoff of APO $k\in{\cal K}_A$ as (\ref{revision:equ:APOkpayoff}). We use ${\bm {\tilde b}}_{S}\triangleq \left(\tilde b_n,n\in{\cal K}_S\right)$ to represent the virtual bids of the APOs in set ${\cal K}_S$. Variable ${\tilde b}_k$ is APO $k$'s virtual bid, and we use ${\bm {\tilde b}}_{A,-k}\triangleq \left(\tilde b_n,n\ne k,n\in{\cal K}_A\right)$ to represent the virtual bids of the remaining APOs in set ${\cal K}_A$ (other than APO $k$). 

Equation (\ref{revision:equ:APOkpayoff}) summarizes three possible situations: (a) when $\tilde b_k>\min_{n\in{\cal K}_S\cup{\cal K}_A}{\tilde b_n}$, LTE provider $0$ accesses a channel from one of the APOs (other than APO $k$) with the minimum virtual bid. As a result, APO $k$ can exclusively occupy its own channel $k$, and serve its users with rate $r_k$; (b) when ${\tilde b_k\!=\! \min_{n\in{\cal K}_S\cup{\cal K}_A}{\tilde b_n}\!\in\left[0,C\right]}$, LTE provider $0$ cooperates with APO $k$ and one of the other APOs with the minimum virtual bid with the probability ${\frac{1}{\left|{\cal I}_{\min}\right|}}$ and the probability $1-{\frac{1}{\left|{\cal I}_{\min}\right|}}$ ($1\le {\left|{\cal I}_{\min}\right|}\le K_S+K_A$), respectively. Hence, APO $k$'s users receive rate $\min\left\{C,\min_{n\ne k,n\in{\cal K}_S\cup{\cal K}_A}{\tilde b}_n\right\}$ and rate $r_k$ with the probability ${\frac{1}{\left|{\cal I}_{\min}\right|}}$ and the probability $1-{\frac{1}{\left|{\cal I}_{\min}\right|}}$, respectively; (c) when $ \min_{n\in{\cal K}_S\cup{\cal K}_A}{\tilde b_n}=``{\rm N} \textquotedblright$, there is no winner in the auction, and LTE provider $0$ randomly chooses a channel from set ${\cal K}_A$ to coexist with the corresponding APO. With the probability ${\frac{1}{K_A}}$, APO $k$ coexists with LTE provider $0$ and has a data rate of ${{\eta^{\rm APO}}}r_k$; with the probability $1-\frac{1}{K_A}$, APO $k$ has a data rate of $r_k$ by exclusively occupying channel $k$. In this case, the expected data rate that APO $k$'s users receive is $\frac{K_A-1+{\eta^{\rm APO}}}{K_A}r_k$. 

For the parameters and distributions that characterize the APOs, $r_k$ is the private information of APO $k\in{\cal K}_S\cup{\cal K}_A$, and the remaining information, \emph{i.e.}, $K_S,K_A,r_{\min},r_{\max},f\left(\cdot\right),F\left(\cdot\right),$ and ${\eta^{\rm APO}}$, is publicly known to all the APOs and LTE provider $0$. For the parameters that characterize LTE provider $0$, $\theta^{\rm LTE}$ will be announced by LTE provider $0$ to all APOs, as $\theta^{\rm LTE}$ is included in the auction rule ($\theta^{\rm LTE}$ affects the feasibilities of the bids of APOs from set ${\cal K}_S$). For the remaining parameters that characterize LTE provider $0$, \emph{i.e.}, $R_{\rm LTE}$ and ${\delta^{\rm LTE}}$, as we will see in later sections, they will not affect the APOs' strategies. Therefore, they can be either made known or kept unknown to the APOs.

{{Next we analyze the auction by backward induction. In Section \ref{revision:sec:stageII:APO}, we analyze the APOs' equilibrium strategies in Stage II, given LTE provider $0$'s reserve rate $C$ in Stage I. In Section \ref{revision:sec:stageI:LTE}, we analyze LTE provider $0$'s equilibrium reserve rate $C^*$ in Stage I by anticipating the APOs' equilibrium strategies in Stage II. In Section \ref{revision:sec:simulation}, we provide the simulation results for the comparison between our auction-based scheme and the benchmark scheme.}}

\subsection{{{Stage II: APOs' Equilibrium Bidding Strategies}}}\label{revision:sec:stageII:APO}

In this section, we assume that the reserve rate $C$ of LTE provider $0$ in Stage I is given, and analyze the APOs' equilibrium strategies in Stage II. Note that studying the APOs' strategies in terms of vector $\bm b$ and vector $\tilde{\bm b}$ is equivalent. For ease of exposition, we characterize the APOs' equilibrium strategies in terms of $\tilde{\bm b}$.

\subsubsection{{Definition of Symmetric Bayesian Nash Equilibrium}}
We focus on the symmetric Bayesian Nash equilibrium (SBNE), which is defined as follows. 
\begin{definition}
Under a reserve rate $C$, bidding strategy functions $\tilde b_S^*\left(r,C\right)$ and $\tilde b_A^*\left(r,C\right)$, $r\in\left[r_{\min},r_{\max}\right]$, constitute a symmetric Bayesian Nash equilibrium if and only if (\ref{revision:equ:defineEQ:S}) holds for all $d_k\in \left[\left(1-\theta^{\rm LTE}\right)R_{\rm LTE},C\right]\cup\left\{{``{\rm N}\textquotedblright}\right\}$, all $r_k\in  \left[r_{\min},r_{\max}\right]$, and all $k\in{\cal K}_S$; and (\ref{revision:equ:defineEQ:A}) holds for all $d_k\in \left[0,C\right]\cup\left\{{``{\rm N}\textquotedblright}\right\}$, all $r_k\in  \left[r_{\min},r_{\max}\right]$, and all $k\in{\cal K}_A$.
\end{definition}

In (\ref{revision:equ:defineEQ:S}), we use $\tilde b_S^*\left(r_k,C\right)$ to represent the equilibrium bidding strategy of APO $k\in{\cal K}_S$, and use vector function ${ {{\bm{\tilde b}}_{S,-k}^*}}\left({\bm r}_{S,-k},C\right)\triangleq \left({\tilde b}_{S}^*\left(r_n,C\right),n\ne k,n\in{\cal K}_S\right)$ to represent the equilibrium bidding strategies of the remaining APOs in set ${\cal K}_S$ (other than APO $k$). We also use vector function ${ {{\bm{\tilde b}}_{A}^*}}\left({\bm r}_{A},C\right)\triangleq \left({\tilde b}_{A}^*\left(r_n,C\right),n\in{\cal K}_A\right)$ to represent the equilibrium bidding strategies of the APOs in set ${\cal K}_A$.

\begin{figure*}
\begin{align}
\nonumber
& \mathbb{E}_{{\bm r}_{S,-k},{\bm r}_{A}}\left\{\Pi^{\rm APO}_k\left(\left(\tilde b_S^*\left(r_k,C\right),{ {{\bm {\tilde b}}_{S,-k}^*}\left({\bm r}_{S,-k},C\right)}\right),{ {{\bm {\tilde b}}_{A}^*}\left({\bm r}_{A},C\right)},C\right)|r_k\right\}\\
&\ge \mathbb{E}_{{\bm r}_{S,-k},{\bm r}_{A}}\left\{\Pi^{\rm APO}_k\left(\left(d_k,{ {{\bm{\tilde b}}_{S,-k}^*}\left({\bm r}_{S,-k},C\right)}\right),{ {{\bm{\tilde b}}_{A}^*}\left({\bm r}_{A},C\right)},C\right)|r_k\right\},\label{revision:equ:defineEQ:S}
\end{align}
\hrulefill
\begin{align}
\nonumber
& \mathbb{E}_{{\bm r}_{S},{\bm r}_{A,-k}}\left\{\Pi^{\rm APO}_k\left({ {{\bm{\tilde b}}_{S}^*}\left({\bm r}_{S},C\right)},\left(\tilde b_A^*\left(r_k,C\right),{ {{\bm{\tilde b}}_{A,-k}^*}\left({\bm r}_{A,-k},C\right)}\right),C\right)|r_k\right\}\\
&\ge \mathbb{E}_{{\bm r}_{S},{\bm r}_{A,-k}}\left\{\Pi^{\rm APO}_k\left({ {{\bm{\tilde b}}_{S}^*}\left({\bm r}_{S},C\right)},\left(d_k,{ {{\bm{\tilde b}}_{A,-k}^*}\left({\bm r}_{A,-k},C\right)}\right),C\right)|r_k\right\},\label{revision:equ:defineEQ:A}
\end{align}
\hrulefill
\end{figure*}

The left hand side of inequality (\ref{revision:equ:defineEQ:S}) stands for the expected payoff of APO $k\in{\cal K}_S$ when it bids $\tilde b_S^*\left(r_k,C\right)$. The expectation is taken with respect to ${\bm r}_{S,-k}\triangleq \left(r_n,\forall n\ne k,n\in{\cal K}_S\right)$ and ${\bm r}_A\triangleq \left(r_n,n\in{\cal K}_A\right)$, which denote the types of the remaining APOs in set ${\cal K}_S$ (other than APO $k$) and all APOs in set ${\cal K}_A$, respectively. The vectors ${\bm r}_{S,-k}$ and ${\bm r}_A$ are unknown to APO $k$. Inequality (\ref{revision:equ:defineEQ:S}) implies that APO $k\in{\cal K}_S$ cannot improve its expected payoff by unilaterally changing its bid from $\tilde b_S^*\left(r_k,C\right)$ to any $d_k\in \left[\left(1-\theta^{\rm LTE}\right)R_{\rm LTE},C\right]\cup\left\{{``{\rm N}\textquotedblright}\right\}$.

In (\ref{revision:equ:defineEQ:A}), we use vector function ${ {{\bm{\tilde b}}_{S}^*}}\left({\bm r}_{S},C\right)\triangleq \left({\tilde b}_{S}^*\left(r_n,C\right),n\in{\cal K}_S\right)$ to represent the equilibrium bidding strategies of the APOs in set ${\cal K}_S$. Moreover, we use $\tilde b_A^*\left(r_k,C\right)$ to represent the equilibrium bidding strategy of APO $k\in{\cal K}_A$, and use vector function ${ {{\bm{\tilde b}}_{A,-k}^*}}\left({\bm r}_{A,-k},C\right)\triangleq \left({\tilde b}_{A}^*\left(r_n,C\right),n\ne k,n\in{\cal K}_A\right)$ to represent the equilibrium bidding strategies of the remaining APOs in set ${\cal K}_A$ (other than APO $k$). 

The left hand side of inequality (\ref{revision:equ:defineEQ:A}) stands for the expected payoff of APO $k\in{\cal K}_A$ when it bids $\tilde b_A^*\left(r_k,C\right)$. The expectation is taken with respect to ${\bm r}_{S}\triangleq \left(r_n,\forall n\in{\cal K}_S\right)$ and ${\bm r}_{A,-k}\triangleq \left(r_n,n\ne k,n\in{\cal K}_A\right)$, which denote the types of all APOs in set ${\cal K}_S$ and the remaining APOs in set ${\cal K}_A$ (other than APO $k$), respectively. The vectors ${\bm r}_{S}$ and ${\bm r}_{A,-k}$ are unknown to APO $k$. Inequality (\ref{revision:equ:defineEQ:A}) implies that APO $k\in{\cal K}_A$ cannot improve its expected payoff by unilaterally changing its bid from $\tilde b_A^*\left(r_k,C\right)$ to any $d_k\in \left[0,C\right]\cup\left\{{``{\rm N}\textquotedblright}\right\}$.



\subsubsection{Equilibrium Strategies $\tilde b_S^*\left(r,C\right)$ and $\tilde b_A^*\left(r,C\right)$} Next we characterize the APOs' equilibrium strategy functions $\tilde b_S^*\left(r,C\right)$ and $\tilde b_A^*\left(r,C\right)$. First, we introduce the following lemmas.
\begin{lemma}\label{revision:lemma:rX} 
The following equation admits at least one solution $r$ in $\left(r_{\min},r_{\max}\right)$:
\begin{multline}
\sum_{n=1}^{K_A-1}{\binom{K_A-1}{n} F^n\left(r\right)\left(1-F\left(r\right)\right)^{K_A-1-n}\frac{C-r}{n+1}}+\\
\left(1-F\left(r\right)\right)^{K_A-1}\left(C-\frac{K_A-1+{\eta^{\rm APO}}}{K_A}r\right)=0.\label{revision:equ:rX}
\end{multline}
We denote the solutions $r$ in $\left(r_{\min},r_{\max}\right)$ as $r_l^x\left(C\right)$ with $1\le l\le L$, where $L$ is the number of solutions.
\end{lemma}
\begin{lemma}\label{revision:lemma:rT} 
The following equation admits at least one solution $r$ in $\left(C,r_{\max}\right)$:
\begin{multline}
\!\!\sum_{n=1}^{K_A-1}\!{\binom{K_A-1}{n} \!\left(F\left(r\right)\!-\!F\left(C\right)\right)^n\left(1\!-\!F\left(r\right)\right)^{K_A-1-n}\frac{C-r}{n+1}} \\
+\left(1-F\left(r\right)\right)^{K_A-1}\left(C-\frac{K_A-1+{\eta^{\rm APO}}}{K_A}r\right)=0.\label{revision:equ:rT}
\end{multline}
We denote the solutions $r$ in $\left(C,r_{\max}\right)$ as $r_m^t\left(C\right)$ with $1\le m\le M$, where $M$ is the number of solutions.
\end{lemma}

Based on the definitions of $r_l^x\left(C\right)$ and $r_m^t\left(C\right)$, we introduce the APOs' equilibrium strategies. For any reserve rate $C\in\left[0,\infty\right)$, functions (\ref{revision:APOequilibrium:S}) and (\ref{revision:APOequilibrium:A}) constitute an SBNE, where $r_X\left(C\right)$ belongs to the set of $\left\{r_l^x\left(C\right),1\le l \le L\right\}$ and $r_T\left(C\right)$ belongs to the set of $\left\{r_m^t\left(C\right),1\le m \le M\right\}$.

\begin{figure*}
\begin{align}
& {\tilde b_S^*}\left(r_k,C\right)=\left\{\begin{array}{ll}
{\eta^{\rm APO}r_k+\left(1-\theta^{\rm LTE}\right)R_{\rm LTE},} & {{\rm if~}{ r_k\in\left[r_{\min},\frac{C-\left(1-\theta^{\rm LTE}\right)R_{\rm LTE}}{\eta^{\rm APO}}\right]},} \\ 
``{\rm N}\textquotedblright, & {{\rm if~} r_k\in\left(\frac{C-\left(1-\theta^{\rm LTE}\right)R_{\rm LTE}}{\eta^{\rm APO}},r_{\max}\right],}
\end{array} \right.\label{revision:APOequilibrium:S}
\end{align}
\hrulefill
\begin{align}
& {\tilde b_A^*}\left(r_k,C\right)=\left\{\begin{array}{ll}
{``{\rm N} \textquotedblright,} & {{\rm if~}{C\in\left[0,\frac{K_A-1+\eta^{\rm APO}}{K_A}r_{\min}\right]},}\\
{C,} & {{\rm if~}{C\in\left(\frac{K_A-1+\eta^{\rm APO}}{K_A}r_{\min},r_{\min}\right)}{\rm~and~}{ r_k\in\left[r_{\min},r_X\left(C\right)\right]},}\\
{``{\rm N} \textquotedblright,} & {{\rm if~}{C\in\left(\frac{K_A-1+\eta^{\rm APO}}{K_A}r_{\min},r_{\min}\right)}{\rm~and~}{r_k\in\left(r_X\left(C\right),r_{\max}\right]},}\\
{r_k,} & {\rm if~}{C\in\left[r_{\min},r_{\max}\right)}{\rm~and~}{r_k\in\left[r_{\min},C\right],}\\
C, & {{\rm if~}{C\in\left[r_{\min},r_{\max}\right)}{\rm~and~}{ r_k\in\left(C,r_T\left(C\right)\right]},}\\
``{\rm N} \textquotedblright, & {{\rm if~} {C\in\left[r_{\min},r_{\max}\right)}{\rm~and~}{r_k\in\left(r_T\left(C\right),r_{\max}\right]},}\\
{r_k,} & {{\rm if~}{C\in\left[r_{\max},\infty\right),}}
\end{array} \right.\label{revision:APOequilibrium:A}
\end{align}
\hrulefill
\begin{align}
& {{\tilde b_A^{\prime}}}\left(r_k,C\right)=\left\{\begin{array}{ll}
{``{\rm N} \textquotedblright,} & {{\rm if~}{C\in\left[0,\frac{K_A-1+\eta^{\rm APO}}{K_A}r_{\min}\right]},}\\
{C,} & {{\rm if~}{C\in\left(\frac{K_A-1+\eta^{\rm APO}}{K_A}r_{\min},r_{\min}\right)}{\rm~and~}{ r_k\in\left[r_{\min},r_X\left(C\right)\right]},}\\
{``{\rm N} \textquotedblright,} & {{\rm if~}{C\in\left(\frac{K_A-1+\eta^{\rm APO}}{K_A}r_{\min},r_{\min}\right)}{\rm~and~}{r_k\in\left(r_X\left(C\right),r_{\max}\right]},}\\
{r_k,} & {\rm if~}{C\in\left[r_{\min},r_{\max}\right)}{\rm~and~}{r_k\in\left[r_{\min},C\right],}\\
C, & {{\rm if~}{C\in\left[r_{\min},r_{\max}\right)}{\rm~and~}{ r_k\in\left(C,r_T\left(C\right)\right]},}\\
``{\rm N} \textquotedblright, & {{\rm if~} {C\in\left[r_{\min},r_{\max}\right)}{\rm~and~}{r_k\in\left(r_T\left(C\right),r_{\max}\right]},}\\
{r_k,} & {{\rm if~}{C\in\left[r_{\max},\infty\right){\rm~and~}r_k\in\left[r_{\min},r_{\max}\right),}}\\
{{\rm any~value~in}\left[r_{\max},C\right]\!\cup\!\left\{\!``{\rm N}\textquotedblright\!\right\},} & {{\rm if~}{C\in\left[r_{\max},\infty\right){\rm~and~}r_k=r_{\max},}}
\end{array} \right.\label{revision:APOequilibrium:A:prime}
\end{align}
\hrulefill
\end{figure*}

From (\ref{revision:APOequilibrium:S}) and (\ref{revision:APOequilibrium:A}), we observe that only the APOs from set ${\cal K}_A$ with some types $r_k$ may bid the reserve rate $C$. That is to say, the allocative externalities only exist among the APOs from set ${\cal K}_A$. This is because LTE provider $0$ only accesses the channels used by the APOs from set ${\cal K}_A$ under the competition mode. In other words, LTE provider $0$ will only interfere with the APOs from set ${\cal K}_A$ under the competition mode, which pushes the APOs from set ${\cal K}_A$ with some types to bid the reserve rate.

Note that there exist other equilibrium strategy functions that deviate from (\ref{revision:APOequilibrium:S}) or (\ref{revision:APOequilibrium:A}) at some special regions and also constitute an SBNE. For example, if we revise $ {{\tilde b_A^{*}}}\left(r_k,C\right)$ as function ${{\tilde b_A^{\prime}}}\left(r_k,C\right)$ in (\ref{revision:APOequilibrium:A:prime}), then functions ${{\tilde b_S^{*}}}\left(r_k,C\right)$ and ${{\tilde b_A^{\prime}}}\left(r_k,C\right)$ also constitute an SBNE. However, we can easily show that functions ${{\tilde b_S^{*}}}\left(r_k,C\right)$ and ${{\tilde b_A^{\prime}}}\left(r_k,C\right)$ generate the same auction outcome (\emph{e.g.}, APOs' expected payoffs) as functions ${{\tilde b_S^{*}}}\left(r_k,C\right)$ and ${{\tilde b_A^{*}}}\left(r_k,C\right)$. Therefore, we can focus on functions ${{\tilde b_S^{*}}}\left(r_k,C\right)$ and ${{\tilde b_A^{*}}}\left(r_k,C\right)$, and analyze LTE provider $0$'s equilibrium strategy accordingly.

\subsection{{{Stage I: LTE Provider $0$'s Reserve Rate}}}\label{revision:sec:stageI:LTE}

In this section, we analyze LTE provider $0$'s optimal reserve rate by anticipating APOs' equilibrium strategies in Stage II. We first make the following assumption on the CDF of an APO's type, which implies that $r_X\left(C\right)$ and $r_T\left(C\right)$ in (\ref{revision:APOequilibrium:A}) are unique.
\begin{assumption}\label{revision:assumption:unique}
Under the cumulative distribution function $F\left(\cdot\right)$, (a) equation (\ref{revision:equ:rX}) has a unique solution in $\left(r_{\min},r_{\max}\right)$, {i.e.}, $L=1$, and (b) equation (\ref{revision:equ:rT}) has a unique solution in $\left(C,r_{\max}\right)$, {i.e.}, $M=1$.
\end{assumption}

We define LTE provider $0$'s expected payoff as
\begin{align}
{\bar \Pi}^{\rm LTE}\left(C\right) \triangleq \mathbb{E}_{{\bm r}_S,{\bm r}_A}\left\{{\Pi^{\rm LTE}\!\left(\left({\bm{\tilde b}}_S^*\left({{\bm r}_S},C\right),{\bm {\tilde b}}_A^*\left({{\bm r}_A},C\right)\right),C\right)}\!\right\}.\label{revision:equ:expectedpayoff}
\end{align}
Recall that ${\bm r}_S$ and ${\bm r}_A$ are the types of the APOs from set ${\cal K}_S$ and set ${\cal K}_A$, respectively. Moreover, ${ {{\bm{\tilde b}}_{S}^*}}\left({\bm r}_{S},C\right)=\left({\tilde b}_{S}^*\left(r_n,C\right),n\in{\cal K}_S\right)$ and ${ {{\bm{\tilde b}}_{A}^*}}\left({\bm r}_{A},C\right)=\left({\tilde b}_{A}^*\left(r_n,C\right),n\in{\cal K}_A\right)$ are the equilibrium strategies of the APOs from set ${\cal K}_S$ and set ${\cal K}_A$, respectively. Specifically, ${ {{\bm{\tilde b}}_{S}^*}}\left({\bm r}_{S},C\right)$ is determined based on (\ref{revision:APOequilibrium:S}), and ${ {{\bm{\tilde b}}_{A}^*}}\left({\bm r}_{A},C\right)$ is determined based on (\ref{revision:APOequilibrium:A}). 

The LTE provider determines the optimal reserve rate by solving
\begin{align}
& \max_{C\in\left[0,\infty\right)} {~\bar \Pi}^{\rm LTE}\left(C\right)\label{revision:equ:optABC:a}\\
&{\rm s.t.}\! \max\!\left\{\tilde b_{S,\max}\left(C\right)\!-\!\left(1\!-\!\theta^{\rm LTE} \right) R_{\rm LTE},\tilde b_{A,\max}\left(C\right)\right\}\!\le\! R_{\rm LTE},\label{revision:equ:optABC:b}
\end{align}
where we define 
\begin{align}
\nonumber
& \tilde b_{S,\max}\left(C\right) \triangleq \\
& \max \!\left\{ \tilde b_S^*\left(r_k,C\right)\!\in\!\left[\left(1\!-\!\theta^{\rm LTE}\right)R_{\rm LTE},\!C\right]\!:{r_k\!\in\left[r_{\min},r_{\max}\right]} \right\},\\
& \tilde b_{A,\max}\left(C\right) \triangleq \max \left\{ \tilde b_A^*\left(r_k,C\right)\in\left[0,C\right]:{r_k\in\left[r_{\min},r_{\max}\right]} \right\}.
\end{align}
Here, $\tilde b_{S,\max}\left(C\right)$ is the maximum possible virtual bid (except $``{\rm N}\textquotedblright$) from the APOs from set ${\cal K}_S$ at the SBNE under $C$, and $\tilde b_{A,\max}\left(C\right)$ is the maximum possible virtual bid (except $``{\rm N}\textquotedblright$) from the APOs from set ${\cal K}_A$ at the SBNE under $C$. Constraints (\ref{revision:equ:optABC:b}) ensures that LTE provider $0$ has enough capacity to satisfy the bid from the winning APO. 

\begin{figure}[h]
  \centering
  \includegraphics[scale=0.35]{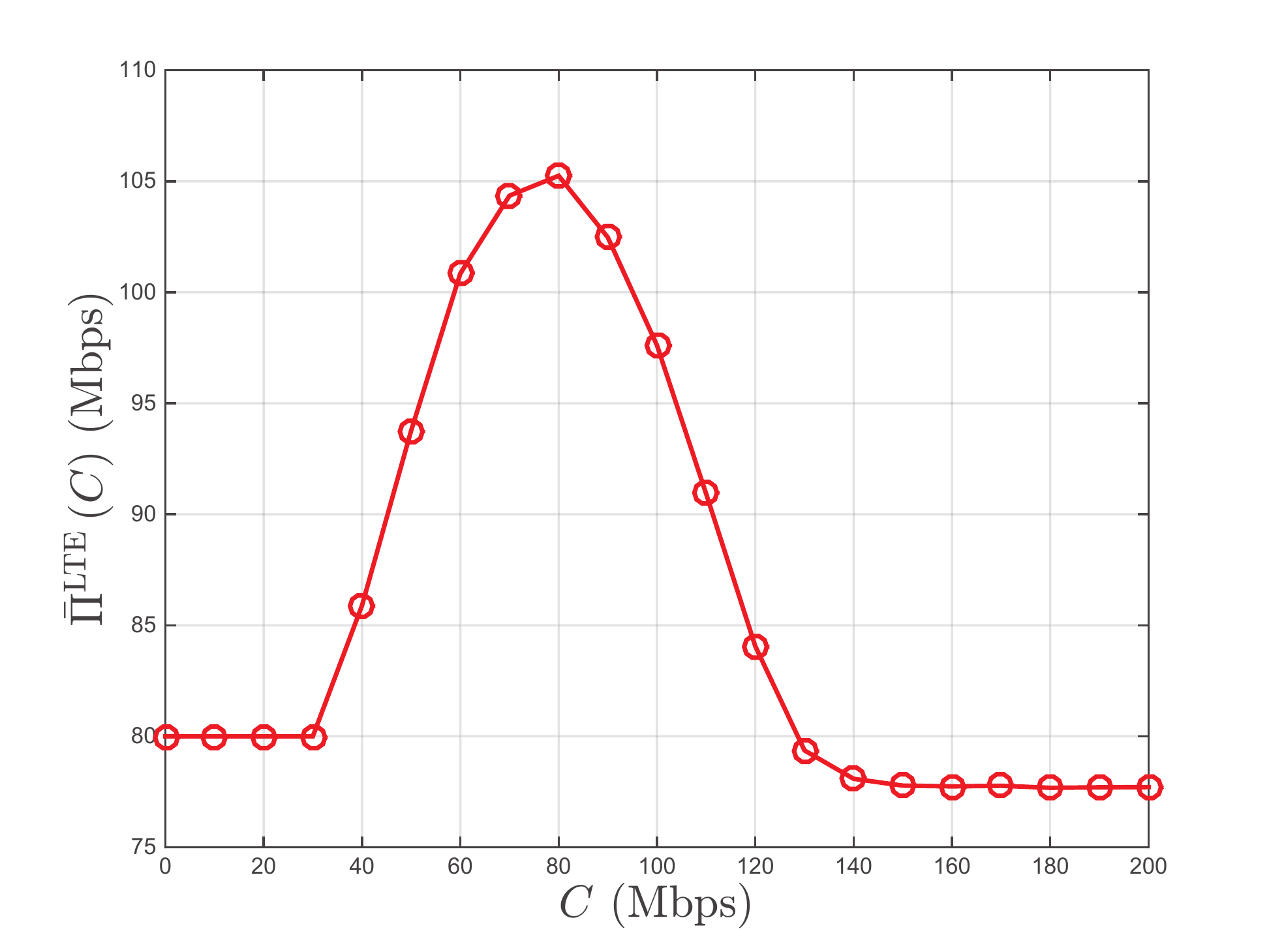}\\
  \caption{Example of Function ${\bar \Pi}^{\rm LTE}\left(C\right)$ (Multi-LTE Case).}
  \label{fig:simu:multiLTE}
  \vspace{-0.2cm}
\end{figure} 

We can solve problem (\ref{revision:equ:optABC:a})-(\ref{revision:equ:optABC:b}) through numerical methods. 
In Fig. \ref{fig:simu:multiLTE}, we illustrate an example of ${\bar \Pi}^{\rm LTE}\left(C\right)$, where we choose $K_S=4$, $K_A=2$, ${\delta^{\rm LTE}}=0.4$, ${\eta^{\rm APO}}=0.3$, $\theta^{\rm LTE}=0.5$, and $R_{\rm LTE}=200{\rm~Mbps}$, and obtain the distribution of $r_k$ by truncating the normal distribution ${\cal N}\left(125~{\rm Mbps},2500~{{\rm Mbps}^2}\right)$ to interval $\left[50~{\rm Mbps},200~{\rm Mbps}\right]$. We observe that ${\bar \Pi}^{\rm LTE}\left(C\right)$ is a unimodal function, and we can apply the Golden Section method to determine the optimal $C^*$.

\subsection{Numerical Results}\label{revision:sec:simulation}
In this section, we compare our auction-based spectrum sharing scheme with a state-of-the-art benchmark scheme for the multi-LTE scenario. Specifically, we randomly pick the APOs, and implement both schemes. Given the APO set ${\cal K}_S$ and APO set ${\cal K}_A$, the two schemes work as follows:
\begin{itemize}
\item \emph{Our auction-based scheme:} First, LTE provider $0$ determines $C^*$ numerically based on Section \ref{revision:sec:stageI:LTE}. Second, each APO $k\in{\cal K}_S$ submits its bid based on the equilibrium strategy $\tilde b_S^*\left(r_k,C^*\right)$ in (\ref{revision:APOequilibrium:S}), and each APO $k\in{\cal K}_A$ submits its bid based on the equilibrium strategy $b_A^*\left(r_k,C^*\right)$ in (\ref{revision:APOequilibrium:A}). Third, LTE provider $0$ determines its working mode, the winning APO, and the allocated data rate based on the auction rule in Section \ref{revision:auctionframe}. 
\item \emph{Benchmark scheme:} LTE provider $0$ randomly picks an APO from set ${\cal K}_A$, and shares the corresponding channel with the APO.
\end{itemize}
For particular APO sets ${\cal K}_S$ and ${\cal K}_A$, we denote LTE provider $0$'s payoff under our auction-based and the benchmark schemes as $\pi_{\rm a}^{\rm LTE}$ and $\pi_{\rm b}^{\rm LTE}$, respectively. Furthermore, we denote all $K_S+K_A$ APOs' total payoffs under our auction-based and the benchmark schemes as $\pi_{\rm a}^{\rm APO}$ and $\pi_{\rm b}^{\rm APO}$, respectively. For particular APO sets ${\cal K}_S$ and ${\cal K}_A$, we compute the relative performance gains of our auction-based scheme over the benchmark scheme in terms of the LTE's payoff and the APOs' total payoff as
\begin{align}
\rho^{\rm LTE} \triangleq \frac{\pi_{\rm a}^{\rm LTE}-\pi_{\rm b}^{\rm LTE}}{\pi_{\rm b}^{\rm LTE}}{\rm ~and~} \rho^{\rm APO} \triangleq \frac{\pi_{\rm a}^{\rm APO}-\pi_{\rm b}^{\rm APO}}{\pi_{\rm b}^{\rm APO}}.
\end{align}

We consider four pairs of $\left({\delta^{\rm LTE}},{\eta^{\rm APO}}\right)$: $\left(0.4,0.1\right),\left(0.4,0.3\right),\left(0.4,0.7\right)$, and $\left(0.6,0.3\right)$, and change $R_{\rm LTE}$ from $30{\rm~Mbps}$ to $370{\rm~Mbps}$. Moreover, we choose $K_S=2$, $K_A=2$, $\theta^{\rm LTE}=0.5$, and obtain the distribution of $r_k$, $k\in{\cal K}_S\cup{\cal K}_A$, by truncating the normal distribution ${\cal N}\left(125~{\rm Mbps},2500~{{\rm Mbps}^2}\right)$ to interval $\left[50~{\rm Mbps},200~{\rm Mbps}\right]$. We illustrate the simulation scenario in Fig. \ref{fig:simu:multiLTEscenario}.

\begin{figure}[h]
  \centering
  \includegraphics[scale=0.4]{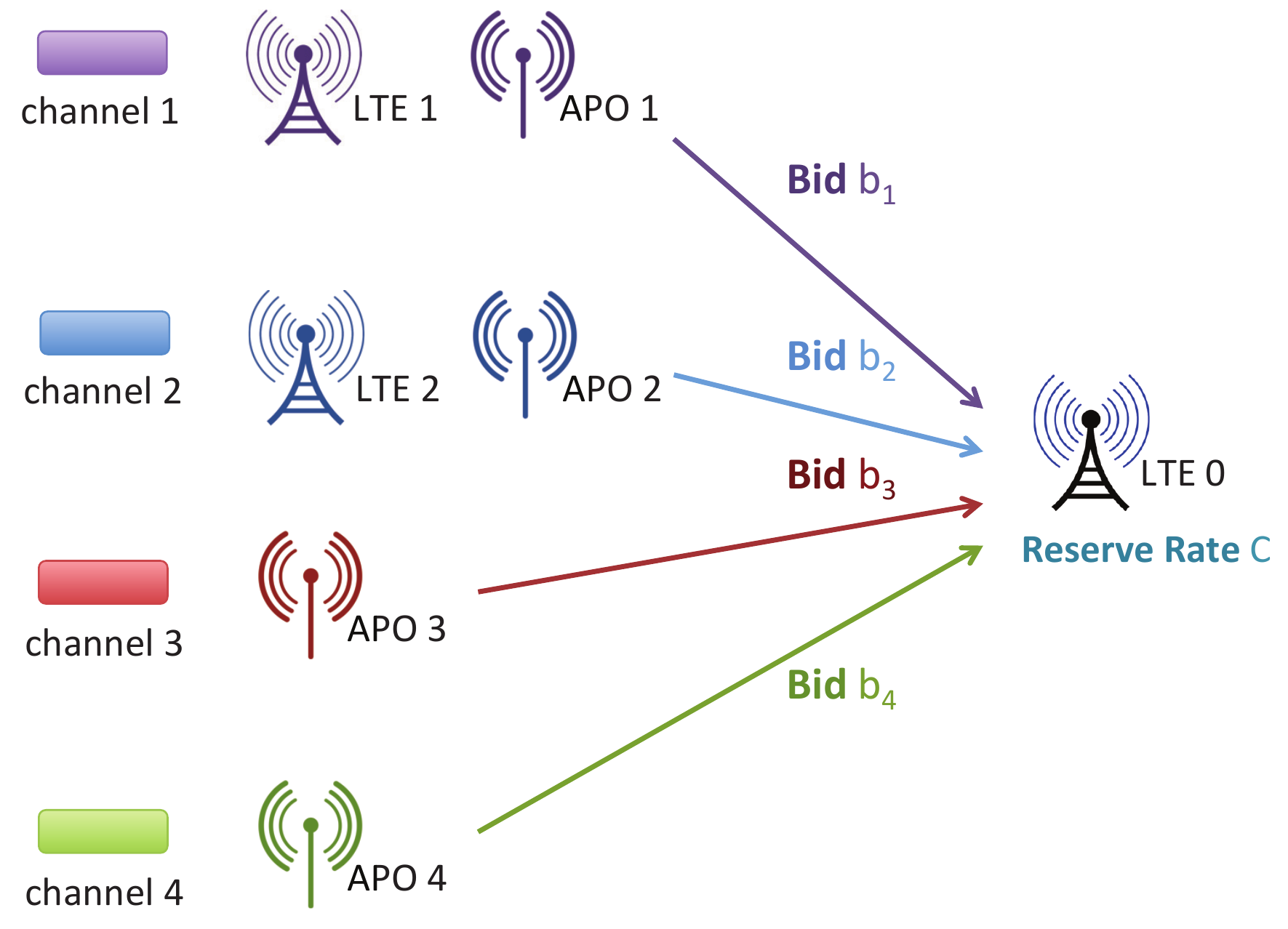}\\
  \caption{{In the simulation, we choose $K_S=2$ and $K_A=2$. There are three LTE providers and four APOs in the system: LTE provider $0$ organizes the auction, LTE provider $1$ coexists with APO $1$ in channel $1$, LTE provider $2$ coexists with APO $2$ in channel $2$, and APOs $3$ and $4$ occupy channels $3$ and $4$ alone, respectively.}}
  \label{fig:simu:multiLTEscenario}
  \vspace{-0.2cm}
\end{figure}  

Given a pair of $\left({\delta^{\rm LTE}},{\eta^{\rm APO}}\right)$ and a particular value of $R_{\rm LTE}$, we randomly choose $r_k$, $k\in{\cal K}_S\cup{\cal K}_A$, based on the truncated normal distribution, implement our auction-based scheme and the benchmark scheme separately, and record the corresponding values of $\rho^{\rm LTE}$ and $\rho^{\rm APO}$. For each pair of $\left({\delta^{\rm LTE}},{\eta^{\rm APO}}\right)$ and each value of $R_{\rm LTE}$, we run the experiment $20,000$ times, and obtain the corresponding average values of $\rho^{\rm LTE}$ and $\rho^{\rm APO}$.

\begin{figure}[t]
  \centering
  \includegraphics[scale=0.35]{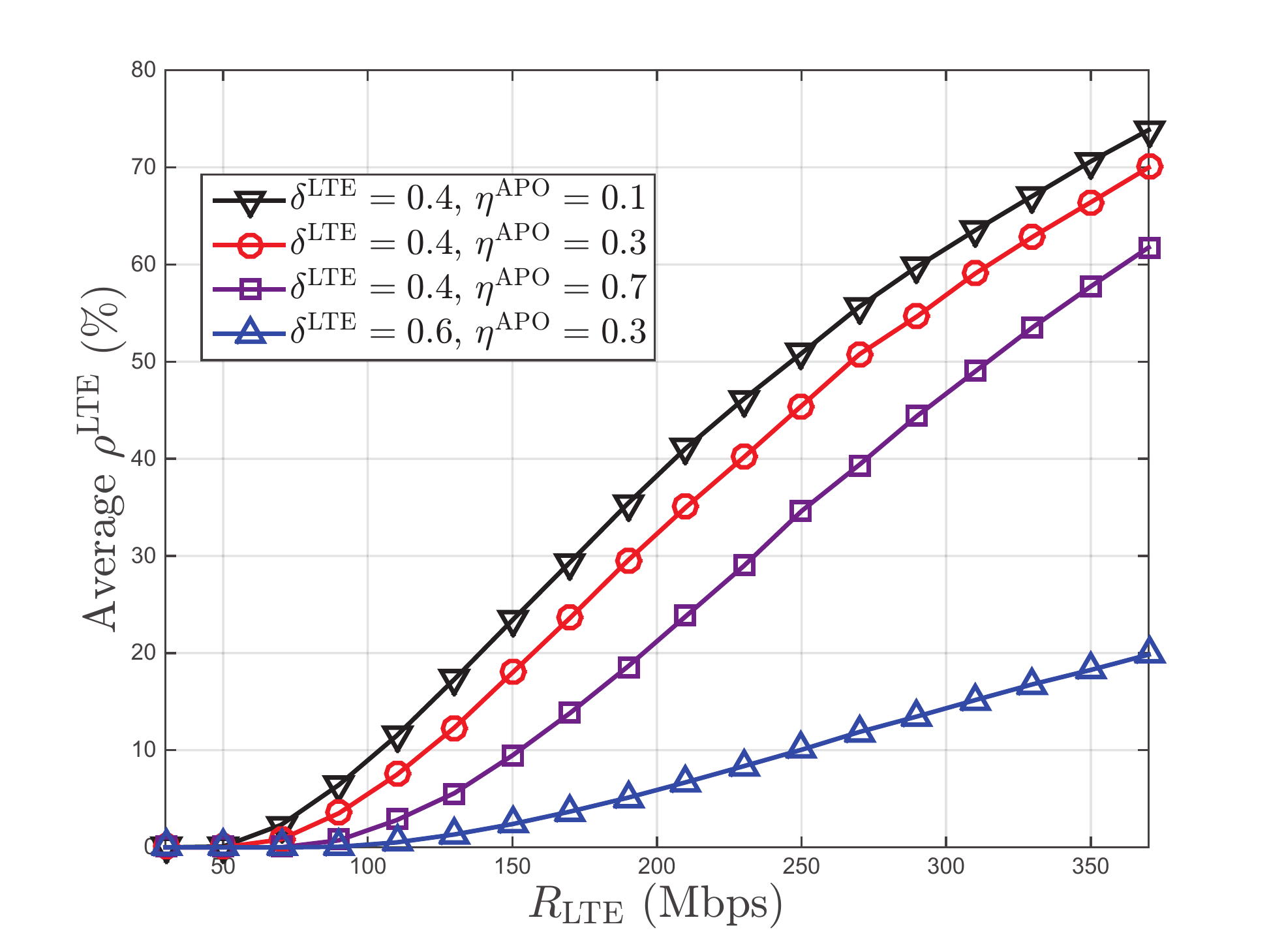}
  \caption{Comparison on LTE's Payoff (Multi-LTE Case).}
  \label{revision:simu:rhoLTE}
  \vspace{-4mm}
\end{figure} 
 

In Fig. \ref{revision:simu:rhoLTE}, we plot the average $\rho^{\rm LTE}$ against $R_{\rm LTE}$ for different $\left({\delta^{\rm LTE}},{\eta^{\rm APO}}\right)$ pairs. 
When $R_{\rm LTE}=370{\rm~Mbps}$ and ${\delta^{\rm LTE}}=0.4$, we observe that all the average $\rho^{\rm LTE}$ (\emph{i.e.}, the black, red, and purple curves) are between $60\%\sim75\%$. Notice that for the simulation of the single-LTE scenario, we also have four APOs, but each of them occupies a channel alone. We have provided the results of $\rho^{\rm LTE}$ for the single-LTE scenario in Fig. \ref{fig:simu:E} of our paper, where all the average $\rho^{\rm LTE}$ with $R_{\rm LTE}=370{\rm~Mbps}$ and ${\delta^{\rm LTE}}=0.4$ are around $70\%$. Comparing Fig. \ref{revision:simu:rhoLTE} with Fig. \ref{fig:simu:E}, we conclude that our auction can significantly improve the LTE's payoff for the multi-LTE scenario as well as the single-LTE scenario. 
We summarize the observations in Fig. \ref{revision:simu:rhoLTE} as follows.
\begin{observation}\label{revision:observation:LTE}
For the multi-LTE scenario, our auction-based scheme improves the LTE's payoff by more than 60\% on average under a large $R_{\rm LTE}$ and a small ${\delta^{\rm LTE}}$, compared with the benchmark scheme.
\end{observation}

In Fig. \ref{revision:simu:rhoAPO}, we plot the average $\rho^{\rm APO}$ against $R_{\rm LTE}$ for different $\left({\delta^{\rm LTE}},{\eta^{\rm APO}}\right)$ pairs. When $R_{\rm LTE}=370{\rm~Mbps}$, ${\delta^{\rm LTE}}=0.4$, and $\eta^{\rm APO}=0.1$, we observe that the corresponding average $\rho^{\rm APO}$ is around $60\%$. In contrast, Fig. \ref{fig:simu:F} of our paper shows the results of $\rho^{\rm APO}$ for the single-LTE scenario, where the average $\rho^{\rm APO}$ with $R_{\rm LTE}=370{\rm~Mbps}$, ${\delta^{\rm LTE}}=0.4$, and $\eta^{\rm APO}=0.1$ is around $32\%$. Therefore, the improvement of the APOs' total payoff under our auction-based scheme for the multi-LTE scenario can be greater than that for the single-LTE scenario. This is because compared with the single-LTE scenario, the APOs are more severely interfered by the LTE networks in the multi-LTE scenario. Under the benchmark scheme, the APOs' total payoff (\emph{i.e.}, ${\pi_{\rm b}^{\rm APO}}$) for the multi-LTE scenario is smaller than that for the single-LTE scenario. Therefore, when implementing our auction-based scheme and allowing the LTE and APOs to cooperate with each other, the percentage of the improvement of the APOs' total payoff (\emph{i.e.}, $\rho^{\rm APO} \triangleq \frac{\pi_{\rm a}^{\rm APO}-\pi_{\rm b}^{\rm APO}}{\pi_{\rm b}^{\rm APO}}$) for the multi-LTE scenario will be greater than that for the single-LTE scenario.   
We summarize the observations in Fig. \ref{revision:simu:rhoAPO} as follows.
\begin{observation}\label{revision:observation:APO}
The improvement of the APOs' total payoff under our auction-based scheme for the multi-LTE scenario can be greater than that for the single-LTE scenario.
\end{observation}

Combing Observations \ref{revision:observation:LTE} and \ref{revision:observation:APO}, we conclude that our auction-based scheme can significantly improve both the LTE's and APOs' payoffs for the multi-LTE scenario. In other words, our auction-based scheme works well for the multi-LTE scenario.

\begin{figure}[t]
  \centering
  \includegraphics[scale=0.35]{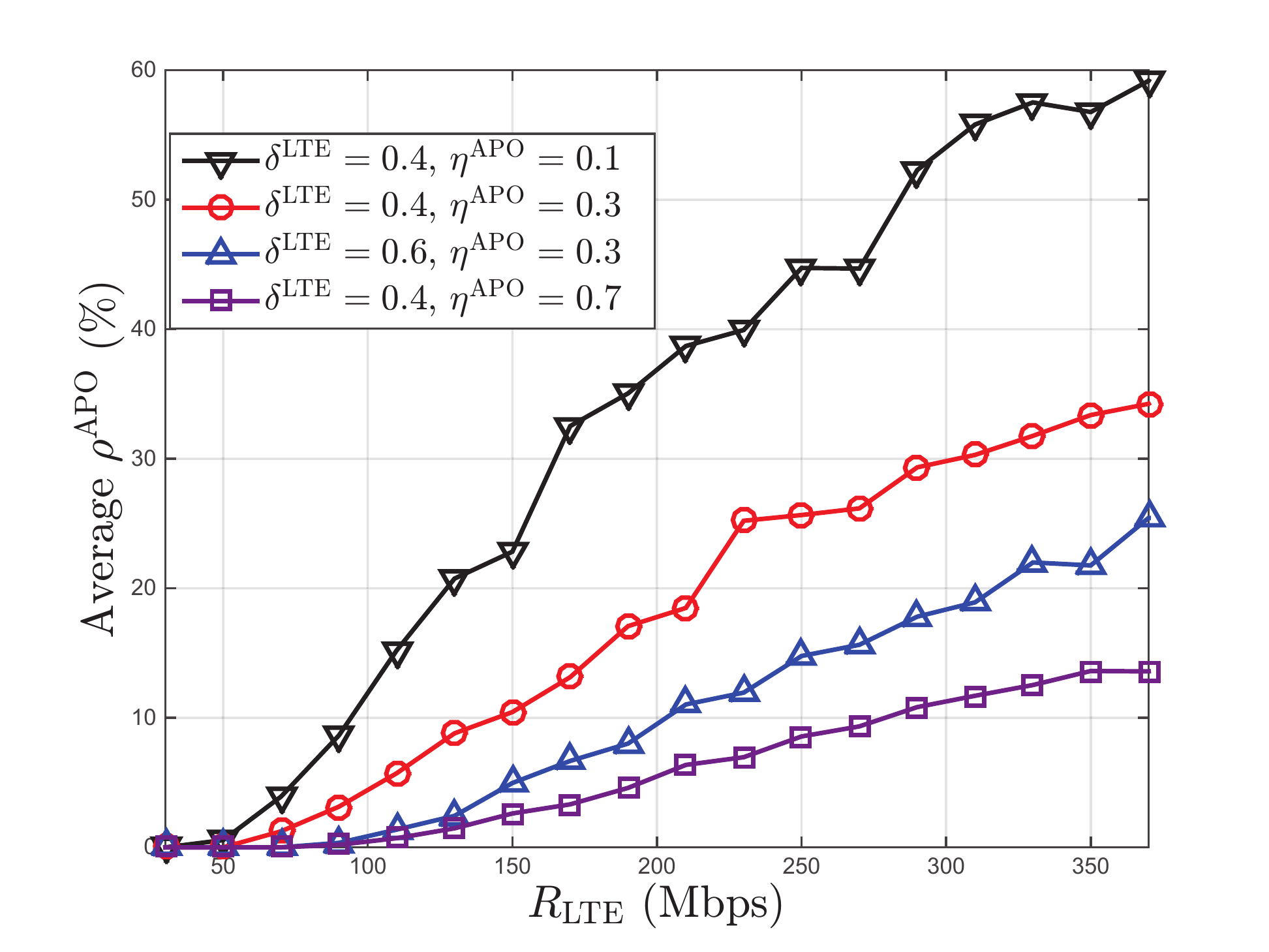}
  \caption{Comparison on APOs' Payoffs (Multi-LTE Case).}
  \label{revision:simu:rhoAPO}
  \vspace{-4mm}
\end{figure}

\end{document}